\documentclass[reprint,superscriptaddress,nofootinbib,amsmath,amssymb,pra]{revtex4-2}

\usepackage{graphicx}
\usepackage{dcolumn}
\usepackage{bm}
\usepackage[colorlinks,linkcolor=blue,urlcolor=blue,citecolor=blue,breaklinks=true]{hyperref}

\usepackage{amsthm}
\usepackage{physics}
\usepackage{dsfont}
\usepackage{multirow}
\usepackage{mathtools}

\newcommand{\1}{\mathds{1}}
\newcommand{\isometry}[2]{\mathbb{V}_\mathrm{iso}(#1, #2)}
\newcommand{\SU}{\mathrm{SU}}
\newcommand{\CC}{\mathbb{C}}
\newcommand{\RR}{\mathbb{R}}
\newcommand{\NN}{\mathbb{N}}
\newcommand{\Syoung}{\mathbb{S}_\mathrm{Young}}

\newcommand{\young}[2]{{\mathbb{Y}^{#1}_{#2}}}

\newcommand{\mcA}{\mathcal{A}}
\newcommand{\mcB}{\mathcal{B}}
\newcommand{\mcU}{\mathcal{U}}
\newcommand{\mcI}{\mathcal{I}}
\newcommand{\mcO}{\mathcal{O}}
\newcommand{\mcL}{\mathcal{L}}
\newcommand{\mcS}{\mathcal{S}}
\newcommand{\mcH}{\mathcal{H}}
\newcommand{\mcR}{\mathcal{R}}
\newcommand{\mcV}{\mathcal{V}}
\newcommand{\mcP}{\mathcal{P}}
\newcommand{\mcX}{\mathcal{X}}
\newcommand{\mcY}{\mathcal{Y}}
\newcommand{\mcZ}{\mathcal{Z}}
\newcommand{\mcT}{\mathcal{T}}
\newcommand{\mcW}{\mathcal{W}}
\newcommand{\mcD}{\mathcal{D}}

\newcommand{\sfT}{\mathsf{T}}
\newcommand{\mfS}{\mathfrak{S}}

\newcommand{\mbS}{\mathbb{S}}

\newcommand{\dket}[1]{\vert {#1} \rangle\!\rangle}
\newcommand{\dbraket}[2]{\langle\!\langle {#1} \vert {#2} \rangle\!\rangle}
\newcommand{\dketbra}[2]{\vert {#1} \rangle\!\rangle\!\langle\!\langle {#2} \vert}

\makeatletter
\def\dbraket#1{%
    \@ifnextchar\bgroup{%
        \dbraket@{#1}%
    }{%
        \langle\!\langle {#1} \vert {#1} \rangle\!\rangle%
    }%
}
\def\dbraket@#1#2{%
    \langle\!\langle {#1} \vert {#2} \rangle\!\rangle%
}
\makeatother

\makeatletter
\def\dketbra#1{%
    \@ifnextchar\bgroup{%
        \dketbra@{#1}%
    }{%
        \vert {#1} \rangle\!\rangle\!\langle\!\langle {#1} \vert%
    }%
}
\def\dketbra@#1#2{%
    \vert {#1} \rangle\!\rangle\!\langle\!\langle {#2} \vert%
}
\makeatother

\newtheorem{definition}{Definition}
\newtheorem{theorem}[definition]{Theorem}
\newtheorem{corollary}[definition]{Corollary}
\newtheorem{lemma}[definition]{Lemma}
\newtheorem{conjecture}[definition]{Conjecture}

\begin{document}

\preprint{APS/123-QED}

\title{Quantum Advantage in Storage and Retrieval of Isometry Channels}

\author{Satoshi Yoshida}
\email{satoshiyoshida.phys@gmail.com}
\affiliation{Department of Physics, Graduate School of Science, The University of Tokyo, 7-3-1 Hongo, Bunkyo-ku, Tokyo 113-0033, Japan}
\author{Jisho Miyazaki}
\affiliation{Department of Physics, Graduate School of Science, The University of Tokyo, 7-3-1 Hongo, Bunkyo-ku, Tokyo 113-0033, Japan}
\affiliation{Ritsumeikan University BKC Research Organization of Social Sciences, 1-1-1 Noji-Higashi, Kusatsu, Shiga 525-8577, Japan}
\author{Mio Murao}
\affiliation{Department of Physics, Graduate School of Science, The University of Tokyo, 7-3-1 Hongo, Bunkyo-ku, Tokyo 113-0033, Japan}
\affiliation{Trans-scale Quantum Science Institute, The University of Tokyo, Bunkyo-ku, Tokyo 113-0033, Japan}

\date{\today}

\begin{abstract}
Storage and retrieval refer to the task of encoding an unknown quantum channel $\Lambda$ into a quantum state, known as the program state, such that the channel can later be retrieved.
There are two strategies for this task: \emph{classical} and \emph{quantum strategies}.
The classical strategy uses multiple queries to $\Lambda$ to estimate $\Lambda$ and retrieves the channel based on the estimate represented in classical bits.
The classical strategy turns out to offer the optimal performance for the storage and retrieval of unitary channels.
In this work, we analyze the asymptotic performance of the classical and quantum strategies for the storage and retrieval of isometry channels.
We show that the optimal fidelity for isometry estimation is given by $F = 1-{d(D-d)\over n} + O(n^{-4/3})$, where $d$ and $D$ denote the input and output dimensions of the isometry, and $n$ is the number of queries.
This result indicates that, unlike in the case of unitary channels, the classical strategy is suboptimal for the storage and retrieval of isometry channels, which requires $n = \Theta(\epsilon^{-1})$ to achieve the diamond-norm error $\epsilon$.
We propose a more efficient quantum strategy based on port-based teleportation, which stores the isometry channel in a program state using only $n = \Theta(1/\sqrt{\epsilon})$ queries, achieving a quadratic improvement over the classical strategy.
As an application, we extend our approach to general quantum channels, achieving improved program cost compared to prior results by  Gschwendtner, Bluhm, and Winter \href{https://doi.org/10.22331/q-2021-06-29-488}{[Quantum \textbf{5}, 488 (2021)]}.
\end{abstract}

\maketitle

{\it Introduction}.---
Universal programming is the task to store the action of a quantum channel $\Lambda$ to a quantum state called the program state $\phi_\Lambda$~\cite{nielsen1997programmable}.
It is aimed to establish a quantum analogue of a classical program, where bit strings represent channels on bit strings.
The size of the program state is called \emph{program cost}, and the no-programming theorem prohibits deterministic and exact implementation of universal programming using a finite program cost~\cite{nielsen1997programmable}.
To circumvent this no-go theorem, researchers developed probabilistic or approximate protocols~\cite{nielsen1997programmable, kim2001storing, vidal2002storing, hillery2002probabilistic, hillery2002implementation, winter2002scalable, yu2002milti, hillery2004improving, brazier2005probabilistic, hillery2006approximate, ishizaka2008asymptotic, ishizaka2009quantum, sedlak2019optimal, kubicki2019resource, sedlak2020probabilistic, yang2020optimal, banchi2020convex, gschwendtner2021programmability, pavlivcko2022robustness, schoute2024quantum} with generalization to measurements~\cite{duvsek2002quantum, fiuravsek2002universal, paz2003quantum, rovsko2003generalized, fiuravsek2004probabilistic, dariano2005efficient, bergou2006programmable, zhang2006universal, perez2006optimality, he2007programmable, sentis2010multicopy, bisio2011quantum, zhou2012multicopy, zhou2014success, jafarizadeh2017designing, chabaud2018optimal, lewandowska2022storage}, infinite-dimensional systems~\cite{gschwendtner2021infinite}, and generalized probabilistic theory~\cite{miyadera2023programming}.
Previous works construct storage-and-retrieval (SAR) protocols,  where the program state $\phi_\Lambda$ is prepared using multiple queries to $\Lambda$, and the number of queries is called \emph{query complexity}.
SAR protocols allow asynchronous quantum information processing~\cite{kim2025resource}, where the input state can be chosen after the application of the quantum channel $\Lambda$.
This feature is beneficial for the attack of quantum position verification protocols~\cite{beigi2011simplified} and remote quantum computing in the blind setting~\cite{yang2021representation}.

\begin{figure}
    \centering
    \includegraphics[width=\linewidth]{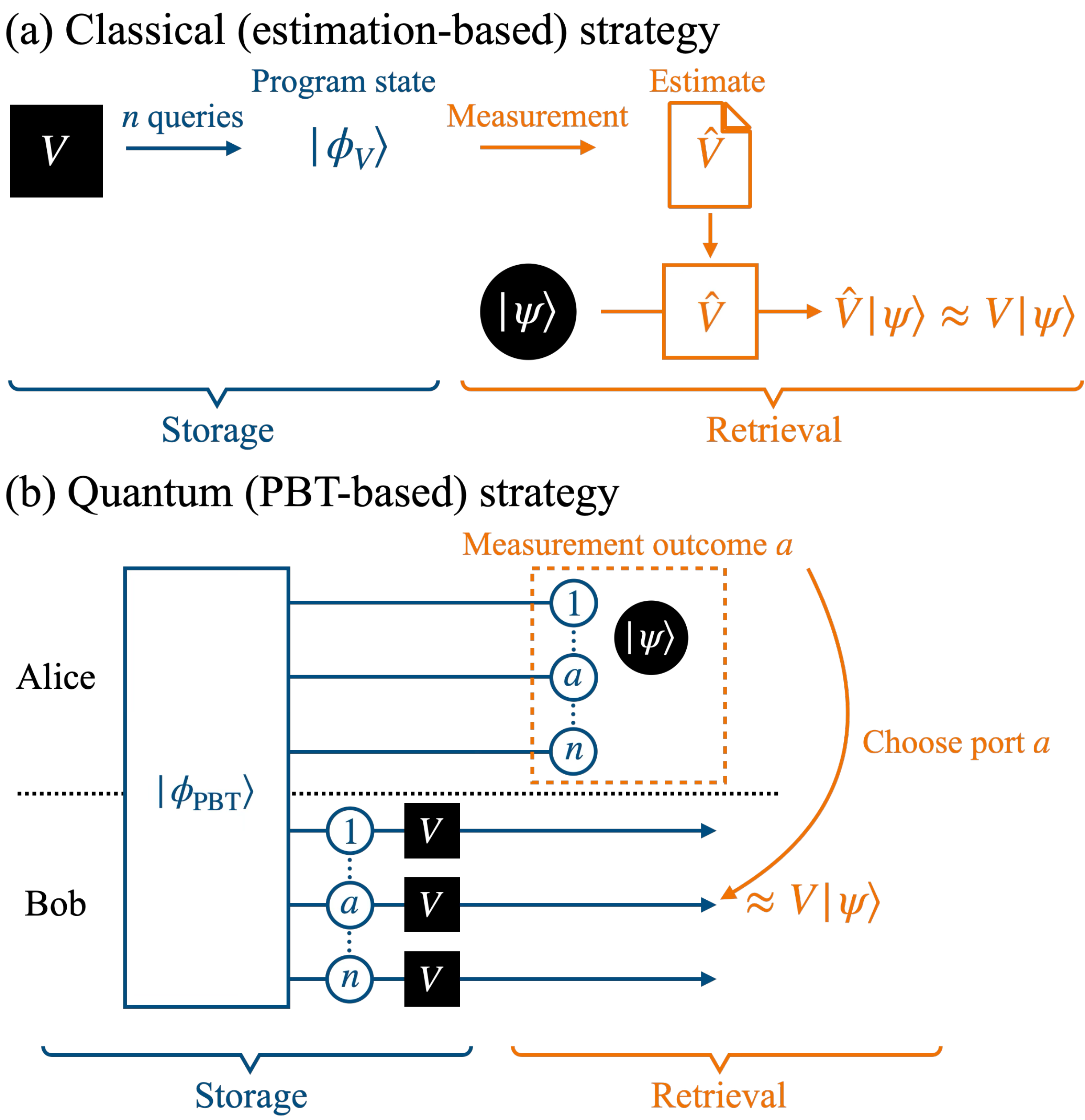}
    \caption{(a) Classical (estimation-based) strategy for dSAR of isometry channels $\mcV(\cdot)\coloneqq V\cdot V^\dagger$.
    One first prepares a quantum state $\ket{\phi_V}$ with multiple queries to $\mcV$, and measures $\ket{\phi_V}$ to obtain the estimate $\hat{V}$ corresponding to the isometry channel $\hat{\mcV}(\cdot)\coloneqq \hat{V}\cdot \hat{V}^\dagger$ as the measurement outcome.
    The quantum state $\ket{\phi_V}$ can be used as the program state, and the isometry channel can be retrieved by applying the estimated isometry channel $\hat{\mcV}$ on an input quantum state $\ket{\psi}$.\\
    (b) Quantum (PBT-based) strategy for dSAR of isometry channels.
    The sender (Alice) and the receiver (Bob) share a $2n$-qudit entangled state $\phi_\mathrm{PBT}$.
    PBT is the task to send an unknown quantum state $\ket{\psi}$, where Alice and Bob share a $2n$-qudit entangled state $\ket{\phi_\mathrm{PBT}}$.
    Alice applies a joint measurement on $\ket{\psi}$ and her share of $\ket{\phi_\mathrm{PBT}}$, sends the measurement outcome $a$ to Bob, and Bob chooses the port $a$ on his share of $\ket{\phi_\mathrm{PBT}}$ to obtain an input quantum state $\ket{\psi}$.
    The quantum state $(\1_d^{\otimes n} \otimes V^{\otimes n})\ket{\phi_\mathrm{PBT}}$ can be used as a program state for dSAR of an isometry channel $\mcV$.
    }
    \label{fig:comparison_classical_and_quantum}
\end{figure}

A natural strategy for SAR is the \emph{classical} strategy, which is proposed for the deterministic SAR (dSAR) of unitary channels, based on unitary estimation~\cite{acin2001optimal}, and also called the \emph{estimation-based} strategy.
It is classical in the sense that the strategy involves extracting the matrix representation of the unitary channel, which can be stored in a classical register.
Unitary estimation is the task to estimate an unknown unitary channel $\mcU(\cdot)\coloneqq U\cdot U^\dagger$ corresponding to $U\in \SU(d)$.
To do this, one first prepares a quantum state $\ket{\phi_U}$ using multiple queries to $\mcU$ and measures $\ket{\phi_U}$ to obtain the estimate $\hat{U} \in \SU(d)$ corresponding to the unitary channel $\hat{\mcU}(\cdot)\coloneqq \hat{U}\cdot\hat{U}^\dagger$.
The quantum state $\ket{\phi_U}$ can be used as the program state for dSAR and the retrieval of $\mcU$ is done by applying the estimated unitary channel $\hat{\mcU}$ [see Fig.~\ref{fig:comparison_classical_and_quantum}~(a)].

\footnotetext[1]{This Letter uses the big-O notation $O(\cdot)$, $\Omega(\cdot)$ and $\Theta(\cdot)$, defined as follows~\cite{arora2009computational}:
\begin{align}
    f(x) = O(g(x)) &\Leftrightarrow \limsup_{x\to \infty}{\abs{f(x) \over g(x)}} <\infty,\\
    f(x) = \Omega(g(x)) &\Leftrightarrow g(x) = O(f(x)),\\
    f(x) = \Theta(g(x)) &\Leftrightarrow f(x)=O(g(x)) \text{ and } f(x) = \Omega(g(x)).
\end{align}
Intuitively, $O(\cdot)$ represents an upper bound, $\Omega(\cdot)$ represents a lower bound, and $\Theta(\cdot)$ represents a tight bound.}

Beyond the classical strategy, researchers consider a quantum strategy for SAR, which is any strategy that is allowed in the quantum circuit model.
Besides the classical strategy, it includes the strategy based on port-based teleportation (PBT), called the \emph{PBT-based} strategy.
This strategy directly retrieves the stored channel without extracting its classical description.
The PBT-based strategy is proposed for the probabilistic SAR (pSAR) and dSAR of a general quantum channel $\Lambda$, where the pSAR retrieves $\Lambda$ exactly with a certain success probability, and the dSAR retrieves $\Lambda$ with a certain approximation error.
PBT is a variant of quantum teleportation, which uses a $2n$-qudit entangled state $\ket{\phi_\mathrm{PBT}}$ between the sender (Alice) and the receiver (Bob), and Bob chooses a port based on the measurement outcome~\cite{ishizaka2008asymptotic, ishizaka2009quantum}.
The quantum state $(\1_{\mcL(\CC^d)}^{\otimes n} \otimes \Lambda^{\otimes n})(\ketbra{\phi_\mathrm{PBT}})$ can be used as a program state, and the teleportation protocol retrieves the quantum channel $\Lambda$ [see Fig.~\ref{fig:comparison_classical_and_quantum}~(b)].

The investigation of SAR of unitary channels in the classical and quantum strategy provides insights into the advantage of quantum memory in learning tasks~\cite{huang2021information,aharonov2022quantum, huang2022quantum}.
The task of SAR of a quantum channel $\Lambda$ can be considered as a ``generative'' quantum learning, which generates the quantum state $\Lambda(\rho)$ for an input state $\rho$ using a trained model given as the program state $\phi_\Lambda$.
Note that SAR is also called ``quantum learning''~\cite{bisio2010optimal, mo2019quantum}.
While the optimal success probability of pSAR of unitary channels is shown to be achieved by the PBT-based strategy~\cite{studzinski2018port, sedlak2019optimal}, the optimal approximation error for dSAR of unitary channel is achieved by the classical strategy~\cite{bisio2010optimal, yang2020optimal} (see also Tab.~\ref{tab:summary}\footnotemark[1]).
The construction in Ref.~\cite{yang2020optimal} is based on the unitary estimation protocol achieving the Heisenberg limit (HL)~\cite{kahn2007fast, yang2020optimal,haah2023query, yoshida2024one, yoshida2025asymptotically}, which is made possible by accessing the input and output of the unitary.
Thus, there is \emph{no} quantum advantage in dSAR of unitary channels, i.e., the quantum strategy does not provide an advantage in the approximation error over the classical strategy.
Since the asymptotically optimal unitary estimation is done without quantum memory~\cite{haah2023query}, quantum advantage does not exist even if we restrict the classical strategy to that without quantum memory.

Despite the progress on SAR of unitary channels, less is known for SAR beyond unitary channels.
One of the important classes of quantum channels is the set of \emph{isometry channels}, which is defined by $\mcV(\cdot)\coloneqq V\cdot V^\dagger$ for $V:\CC^d\to\CC^D$ satisfying $V^\dagger V = \1_d$, where $\1_d$ is the identity operator on $\CC^d$.
The isometry channel represents encoding of quantum information onto a higher-dimensional space, used in various quantum information processing tasks such as quantum error correction and quantum communication~\cite{nielsen2010quantum}.
It can be interpreted as a unitary channel whose input space is restricted to a certain subspace, which appears in various quantum algorithms, e.g., Grover's algorithm~\cite{grover1996fast} and the Harrow-Hassidim-Lloyd algorithm~\cite{harrow2009quantum}.
It also represents a general quantum channel via the Stinespring dilation~\cite{wilde2013quantum, watrous2018theory}, and the recent progress on random purification channel and random dilation superchannel provides a way to process general quantum channels via the dilation isometry channel~\cite{chen2024local,tang2025conjugate,pelecanos2025mixed,mele2025optimal,girardi2025random,walter2025random,mele2025random,chen2025quantum,girardi2025random2,yoshida2025random}.
Due to its ubiquitous use in quantum information processing, SAR of isometry channels finds various applications, e.g., storing unitary operations applied on a subspace and general quantum channels via the Stinespring dilation.
However, less is known about an isometry channel compared to a unitary channel, e.g., the optimal estimation of an isometry channel is not known.
Since isometry channel can be considered as the unitary channel where the input space is restricted to a subspace, it is not trivial whether the optimal estimation error obeys the HL (true for unitary estimation~\cite{yang2020optimal, haah2023query, yoshida2024one, yoshida2025asymptotically}) or the standard quantum limit (SQL; true for state estimation~\cite{bruss1999optimal}).

This work investigates dSAR and estimation of isometry channels.
We define the task \emph{isometry estimation} as a generalization of unitary estimation and compare the classical (estimation-based) strategy and the quantum (PBT-based) strategy for dSAR of isometry channels.
We show that the isometry estimation obeys the SQL and completely determine the leading term, which leads to showing the inefficiency of the estimation-based strategy in terms of the query complexity.
The query complexity of the PBT-based strategy is shown to be optimal, which shows a quadratic advantage over the classical strategy.
We also show the universal programming of general quantum channels as an application, whose program cost is smaller than the protocol shown in Ref.~\cite{gschwendtner2021programmability}.

\begin{table}
    \caption{Comparison of dSAR of unitary and isometry channels.
    The optimal query complexity for dSAR of a unitary channel is achieved by the classical strategy.
    The classical strategy for dSAR of isometry channels provides a sub-optimal query complexity (Thm.~\ref{thm:SQL_for_isometry}), while the quantum strategy provides the optimal one (Cor.~\ref{cor:optimal_SAR}).}
    \begin{ruledtabular}
    \begin{tabular}{c|c|c}
        \makebox[\dimexpr0.17\linewidth-2\tabcolsep-\arrayrulewidth/2\relax][c]{} &
        \makebox[\dimexpr0.48\linewidth-2\tabcolsep-\arrayrulewidth\relax][c]{Classical strategy} &
        \makebox[\dimexpr\linewidth-0.17\linewidth-0.48\linewidth-2\tabcolsep-\arrayrulewidth/2\relax][c]{Quantum strategy}\\\hline
        Unitary & $n = {\Theta(d^2) \over \sqrt{\epsilon}}$~\cite{kahn2007fast, bisio2010optimal, yang2020optimal, haah2023query, yoshida2024one, yoshida2025asymptotically} & $n = {\Theta(d^2) \over \sqrt{\epsilon}}$~\cite{yoshida2024one}\\\hline
        \multirow{2}{*}{Isometry} & $n = {d(D-d) \over \epsilon} + O(\epsilon^{-2/3})$ & $n = {\Theta(d^2)\over \sqrt{\epsilon}}$\\
        & \textbf{[Thm.~\ref{thm:SQL_for_isometry}]} & \textbf{[Cor.~\ref{cor:optimal_SAR}]}
    \end{tabular}
    \end{ruledtabular}
    \label{tab:summary}
\end{table}

{\it Definition of the tasks}.---
For a set of quantum channels $\mbS$, dSAR of $\mbS$ is the task to prepare a quantum state $\phi_\Lambda \in \mcL(\mcP)$ called the program state by $n$ queries of $\Lambda\in \mbS$, and retrieve a quantum channel $\Lambda\in\mbS$.
The number of queries is called the \emph{query complexity} of the protocol.
The size of the program state is called the \emph{program cost}, defined by $c_P\coloneqq \log \dim \mcP$.
The retrieval is done by applying a quantum channel $\Phi$, which is independent of $\Lambda$, on an input quantum state $\rho$ and the quantum state $\phi_\Lambda$.
The retrieved channel $\mcR_\Lambda$ is given by $\mcR_\Lambda(\rho)\coloneqq \Phi(\rho\otimes \phi_\Lambda)$, which approximates the original channel $\Lambda$.
The approximation error $\epsilon$ of the retrieved channel is called the \emph{retrieval error}, which is given by the worst-case diamond-norm error:
\begin{align}
    \epsilon\coloneqq {1\over 2} \sup_{\Lambda\in\mbS} \|\mcR_\Lambda-\Lambda\|_\diamond.
\end{align}
In this work, we consider three classes of the set $\mbS$:
\begin{itemize}
    \item Unitary channels: $\mbS = \mbS_{\mathrm{Unitary}}^{(d)}\coloneqq \{\mcU(\cdot)\coloneqq U\cdot U^\dagger \mid U\in\SU(d)\}$,
    \item Isometry channels: $\mbS = \mbS_{\mathrm{Isometry}}^{(d,D)} \coloneqq \{\mcV(\cdot)\coloneqq V\cdot V^\dagger \mid V\in \isometry{d}{D}\}$, where $\isometry{d}{D}\coloneqq \{V: \CC^d\to \CC^D \mid V^\dagger V = \1_d\}$.
    \item General quantum channels, i.e., completely positive and trace preserving (CPTP) maps: $\mbS = \mbS_{\mathrm{CPTP}}^{(d, D)}\coloneqq \{\Lambda:\mcL(\CC^d)\to \mcL(\CC^D) \mid \Lambda \text{ is a CPTP map}\}$, where $\mcL(\mcX)$ represents the set of linear operators on a Hilbert space $\mcX$.
\end{itemize}

Unitary estimation is the task defined as follows.
An unknown unitary operator $U\in\SU(d)$ is drawn from the Haar measure of $\SU(d)$, which represents a completely random distribution~\cite{mele2024introduction}.
The task is to estimate the corresponding unitary channel $\mcU(\cdot)\coloneqq U\cdot U^\dagger$ with $n$ queries to $\mcU$.
One can first prepare a quantum state $\phi_U$ using $n$ queries to $\mcU$, and measure $\phi_U$ to obtain the estimate $\hat{U}$ as the measurement outcome with the probability distribution denoted by $p(\hat{U} | U) \dd \hat{U}$.
The accuracy of the estimation is evaluated by the \emph{estimation fidelity}, which is given by the average-case channel fidelity:
\begin{align}
    F_\mathrm{est}\coloneqq \int \dd U \int \dd \hat{U} p(\hat{U}|U) F_\mathrm{ch}(U, \hat{U}),
\end{align}
where $\dd U$ and $\dd \hat{U}$ are the Haar measure of $\SU(d)$ and $F_\mathrm{ch}(U,\hat{U})$ is the channel fidelity~\cite{raginsky2001fidelity} between unitary channels $\mcU, \hat{\mcU}$ defined by $F_\mathrm{ch}(U,\hat{U})\coloneqq {1\over d^2} \abs{\Tr(U^\dagger \hat{U})}^2$.
Note that we can also consider the worst-case fidelity given by $\inf_{U\in \SU(d)} \int \dd \hat{U} p(\hat{U}|U) F_\mathrm{ch}(U, \hat{U})$, but this equal to the average-case one in the covariant protocol, which achieves the optimal fidelity~\cite{holevo2011probabilistic,chiribella2005optimal}.

We extend the task of unitary estimation to \emph{isometry estimation}, where an unknown isometry operator $V\in\isometry{d}{D}$ is drawn from the Haar measure of $\isometry{d}{D}$~\cite{zyczkowski2000truncations, kukulski2021generating}, and we evaluate the estimation fidelity by using the channel fidelity between a true isometry channel $\mcV(\cdot)\coloneqq V\cdot V^\dagger$ and an estimated isometry channel $\hat{\mcV}(\cdot)\coloneqq \hat{V}\cdot \hat{V}^\dagger$ defined by $F_\mathrm{ch}(V,\hat{V})\coloneqq {1\over d^2} \abs{\Tr(V^\dagger \hat{V})}^2$.
The task of isometry estimation covers unitary estimation and state estimation as the special cases ($D=d$ for unitary estimation and $d=1$ for state estimation).
We denote the optimal fidelity of isometry estimation by $F_\mathrm{est}(n,d,D)$.
The optimal retrieved error in the estimation-based strategy for dSAR of $\mathbb{S}_\mathrm{Isometry}^{(d,D)}$ is given by $\epsilon = 1-F_\mathrm{est}(n,d,D)$ [see the Supplemental Material (SM)~\cite{supple} for the details].
Similarly to the unitary estimation, we can also define the worst-case fidelity, which is equivalent to the average-case one.

Deterministic port-based teleportation (dPBT) is defined as follows.
The task of dPBT is to send an unknown quantum state $\rho\in\mcL(\mcA_0)$ from the sender (Alice) to the receiver (Bob) using a shared $2n$-qubit entangled state $\phi_\mathrm{PBT}\in \mcL(\mcA^n\otimes \mcB^n)$ between Alice and Bob, where $\mcA^n=\bigotimes_{a=1}^{n}\mcA_a, \mcB^n=\bigotimes_{a=1}^{n}\mcB_a$, $\mcA_0=\mcA_a=\mcB_a=\CC^d$, and $\mcB_a$ is called a port.
Alice applies a positive operator-valued measure (POVM) measurement $\{\Pi_a\}_{a=1}^{n}$ on the unknown quantum state $\rho$ and her share of $\phi_\mathrm{PBT}$, send the measurement outcome $a$ to Bob, and Bob chooses the port $a$ on his share of $\phi_\mathrm{PBT}$ [see Fig.~\ref{fig:comparison_classical_and_quantum}~(b)].
The quantum state Bob obtains is given by
\begin{align}
\label{eq:teleportation_channel}
    \Phi(\rho)\coloneqq \sum_{a=1}^{n} \Tr_{\mcA_0 \mcA^n \overline{\mcB_a}}[(\Pi_a\otimes \1_{\mcB^n})(\rho\otimes \phi_\mathrm{PBT})],
\end{align}
where $\overline{\mcB_a}\coloneqq \bigotimes_{i\neq a} \mcB_i$ and $\1_{\mcB^n}$ is the identity operator on $\mcB^n$.
We define the \emph{teleportation error} $\delta_\mathrm{PBT}\coloneqq {1\over 2}\|\Phi-\1_{\mcL(\CC^d)}\|_\diamond$,
where $\|\cdot\|_\diamond$ is the diamond norm~\cite{kitaev1997quantum, watrous2005notes} and $\1_{\mcL(\CC^d)}$ is the $d$-dimensional identity channel.
We denote the optimal teleportation error by $\delta_\mathrm{PBT}(n,d)$.
The optimal retrieved error in the PBT-based strategy for dSAR of $\mathbb{S}_\mathrm{Isometry}^{(d,D)}$ is given by $\epsilon = \delta_{\mathrm{PBT}}(n,d)$ (see the SM~\cite{supple} for the details).

\begin{figure}
    \centering
    \includegraphics[width=\linewidth]{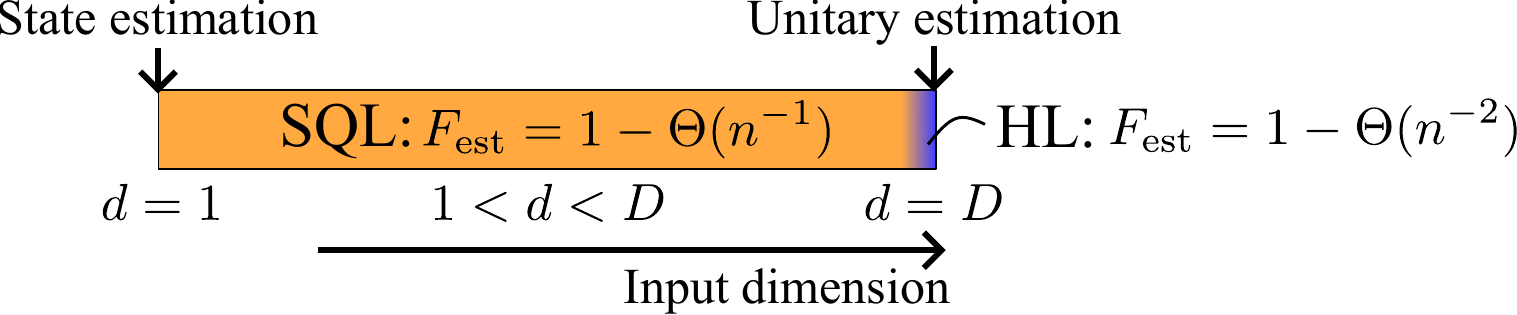}
    \caption{Quantum state estimation obeys the SQL, where the estimation fidelity $F_\mathrm{est}$ is given by $F_\mathrm{est} = 1- \Theta(n^{-1})$ for a query complexity $n$, while unitary estimation achieves the HL, where $F_\mathrm{est} = 1- \Theta(n^{-2})$ holds.
    This work shows that estimation of isometry $V\in\isometry{d}{D}$ obeys the SQL for all $d<D$ (see Thm.~\ref{thm:SQL_for_isometry}).}
    \label{fig:phase_diagram}
\end{figure}

{\it Standard quantum limit for the isometry estimation}.---
We derive the optimal fidelity of isometry estimation as shown in the following Theorem\footnotemark[1].

\begin{theorem}
\label{thm:SQL_for_isometry}
The optimal fidelity of isometry estimation is given by
    \begin{align}
        F_\mathrm{est}(n,d,D) = 1-{d(D-d) \over n}+O(n^{-4/3}),
        \label{eq:isometry_estimation_fidelity}
    \end{align}
which is achieved by a parallel protocol.
\end{theorem}
\begin{proof}[Proof sketch]
The SQL $F_\mathrm{est}(n,d,D) \leq 1-\Theta(n^{-1})$ can be shown via the SQL of the parameter estimation~\cite{helstrom1969quantum, holevo2011probabilistic}.
We show the SQL of the quantum fisher information, which lower bounds the estimation error of the parameter estimation, by the ``Hamiltonian-not-in-Kraus-span'' (HNKS) condition shown in Ref.~\cite{zhou2021asymptotic}.
Then, by using the van Trees inequality~\cite{van2001detection}, we show the SQL of the isometry estimation (see the End Matter for the details).
Note that the HNKS condition is usually used to investigate the effect of noise in the parameter estimation, but we utilize it for characterizing estimation of noiseless channels.
The exact determination of the leading order term in Eq.~\eqref{eq:isometry_estimation_fidelity} is shown by evaluating the estimation fidelity using representation-theoretic arguments shown in Ref.~\cite{chiribella2005optimal} (see the SM~\cite{supple} for the details).
\end{proof}

This result is compatible with the previously known result for the state estimation \cite{bruss1999optimal}: $F_\mathrm{est}(n,1,d) = 1-{d-1\over d+n}$,
and the unitary estimation~\cite{kahn2007fast, yang2020optimal, haah2023query, yoshida2024one, yoshida2025asymptotically}: $F_\mathrm{est}(n,d,d) = 1-{\Theta(d^4)\over n^2} + O(n^{-3})$.
In particular, as long as $D>d$ holds, the fidelity of isometry estimation obeys the SQL $F_\mathrm{est} = 1-\Theta(n^{-1})$ similarly to the state estimation [see Fig.~\ref{fig:phase_diagram}~(a)].
This is in contrast with the unitary estimation corresponding to the case of $D=d$, which obeys the HL $F_\mathrm{est} = 1- \Theta(n^{-2})$.
This theorem also shows the exact coefficient of the leading term in $n$, which is unknown for the case of unitary channels except for $d=2, 3$~\cite{bagan2004entanglement, yoshida2025asymptotically}.

Due to Thm.~\ref{thm:SQL_for_isometry}, the estimation-based strategy for the dSAR of isometry channels provides the retrieval error given by $\epsilon = 1-F_\mathrm{est}(n,d,D) = {d(D-d)\over n} + O(n^{-4/3})$, i.e., $n = {d(D-d)\over \epsilon} + O(\epsilon^{-2/3})$ (see Tab.~\ref{tab:summary}).
We show that the PBT-based strategy provides improved query complexity and program cost in the next section.

{\it Universal programming of CPTP maps}.---
Reference~\cite{yoshida2024one} constructs a covariant dPBT protocol achieving the asymptotically optimal teleportation error
\begin{align}
\label{eq:optimal_teleportation_error}
    \delta_\mathrm{PBT}(n,d) = {\Theta(d^4)\over n^2}+O(n^{-3}).
\end{align}
Using this PBT protocol, we can construct the asymptotically optimal protocol for SAR of isometry channels with the program cost given as follows (see the SM~\cite{supple} for the proof).

\begin{corollary}
\label{cor:optimal_SAR}
    The optimal query complexity for dSAR of $\mbS_{\mathrm{Isometry}}^{(d,D)}$ and $\mbS_{\mathrm{CPTP}}^{(d,D)}$ with the retrieval error $\epsilon$ are given by
    \begin{align}
    \label{eq:optimal_SAR}
        n = {\Theta(d^2)\over \sqrt{\epsilon}},
    \end{align}
    which is achieved by the PBT-based strategy.
    The program cost of this protocol for the isometry channels is given by
    \begin{align}
        \label{eq:program_cost_isometry_channel}
        c_P \leq {Dd - 1 \over 2} \log \Theta(\epsilon^{-1}).
    \end{align}
\end{corollary}

As an application of Cor.~\ref{cor:optimal_SAR}, we show a universal programming protocol of CPTP maps, based on the dSAR of isometry channels.
\begin{theorem}
\label{thm:universal_programming_cptp_construction}
    There exists a universal programmable processor of $\mathbb{S}_{\mathrm{CPTP}}^{(d,D)}$ with the program cost given by
    \begin{align}
        c_P \leq {Dd^2-1 \over 2} \log \Theta(\epsilon^{-1}).
    \end{align}
\end{theorem}
\begin{proof}
    By Carath\'{e}odory's theorem, a quantum channel $\Lambda$ with input dimension $d$ can be written as a convex combination of extremal quantum channels $\{\Lambda_i\}$, whose Kraus rank is upper bounded by $d$ \cite{choi1975completely, gschwendtner2021programmability}, as
    $\Lambda = \sum_i p_i \Lambda_i$,
    where $\{p_i\}$ is a probability distribution, i.e., $p_i\geq 0$ and $\sum_i p_i = 1$ hold.
    Let $V_i: \CC^d \to \CC^D\otimes \mcH_\mathrm{aux}$ for $\mcH_\mathrm{aux}=\CC^d$ be the Stinespring dilation \cite{stinespring1955positive} of the quantum channel $\Lambda_i$, i.e.,
    $\Lambda_i(\cdot) = \Tr_\mathrm{aux}[V_i\cdot V_i^\dagger]$.
    As shown in Cor.~\ref{cor:optimal_SAR}, isometry channel $V_i \in \isometry{d_\mathrm{in}}{d_\mathrm{out}}$ for $d_\mathrm{in} = d, d_\mathrm{out} = Dd$ can be stored in the program state $\ket{\phi_{V_i}}$ with the approximation error $\epsilon$ and the program cost 
    \begin{align}
        c_P &\leq {d_\mathrm{in}d_\mathrm{out}-1 \over 2} \log \Theta(\epsilon^{-1}) \\
        &={Dd^2-1 \over 2} \log \Theta(\epsilon^{-1}),
    \end{align}
    which is achieved by the PBT-based strategy.
    Suppose $\mcR_{V_i}$ be the retrieved channel corresponding to the program state $\ket{\phi_{V_i}}$ satisfying $\|\mcR_{V_i}-\mcV_i\|_\diamond \leq \epsilon.$
    The original quantum channel $\Lambda$ can be stored in the program state defined by
    $\phi_{\Lambda}\coloneqq \sum_i p_i \ketbra{\phi_{V_i}}$.
    From the program state $\phi_{\Lambda}$, we can retrieve the quantum channel $\mcR_{\Lambda}\coloneqq \sum_i p_i \mcR_{V_i}$ satisfying
    \begin{align}
        \|\mcR_\Lambda-\Lambda\|_\diamond
        \leq \sum_i p_i \|\mcR_{V_i}-\mcV_i\|_\diamond
        \leq \epsilon,
    \end{align}
    where we use the concavity of the diamond norm.
    Thus, the program cost is given by $c_P \leq {Dd^2-1 \over 2} \log \Theta(\epsilon^{-1})$.
\end{proof}

This program cost is improved over that shown in Ref.~\cite{gschwendtner2021programmability}, which uses classical bits to store the Choi state of $V_i$ to achieve the program cost
\begin{align}
    c_P\leq 2Dd^2 \log \Theta(\epsilon^{-1}).
\end{align}
Our protocol achieves a $75\%$ program cost reduction by using dPBT to store the isometry channels $V_i$.
Note that we can also store the quantum channel $\Lambda$ with the program state $(\1_{\mcL(\CC^d)}^{\otimes n} \otimes \Lambda^{\otimes n})(\phi_\mathrm{PBT})$, but the program cost of this protocol is given by $2n = {\Theta(d^2/\sqrt{\epsilon})}$, which is exponentially worse than our protocol in terms of $\epsilon$ scaling.

{\it Conclusion}.---
This work introduces the task of isometry estimation and obtains the asymptotically optimal estimation fidelity to aim for the classical (estimation-based) strategy for dSAR of isometry channels.
We compare it with the quantum (PBT-based) strategy, and show the quadratic quantum advantage in terms of the query complexity.
The quantum strategy also offers the optimal query complexity for dSAR of CPTP maps.
We show that the obtained dSAR protocol of isometry channels can be used for universal programming of CPTP maps, which has a smaller program cost than that shown in Ref.~\cite{gschwendtner2021programmability}.

In this work, we investigate isometry estimation within the storage-and-retrieval (SAR) framework.
Beyond SAR, the developed isometry estimation protocol also finds applications in the transformation of isometry channels~\cite{yoshida2023universal, yoshida2025universal}.
In particular, the optimal estimation fidelity provides an approximation error of the measure-and-prepare strategy for such a transformation, which serves as a starting point for optimizing such transformation protocols.

The program cost~\eqref{eq:program_cost_isometry_channel} for isometry channels can be considered as a counter-example of a conjecture in Ref.~\cite{yang2020optimal}, which states that the optimal program cost of unitary operators with $\nu$ real parameters is given by
\begin{align}
    \label{eq:conjectured_program_cost_unitary}
    c_P^{(\mathrm{conj})} = {\nu \over 2} \log\Theta(\epsilon^{-1}).
\end{align}
We consider the extendibility of this conjecture to an isometry channel.
Though this conjecture does not hold for the case of state ($d=1$), it is not a trivial problem for the case of $D>d>1$.
If we believe that this conjecture holds for the case of isometry channels, since any isometry operator $V\in\isometry{d}{D}$ can be specified by $2Dd-d^2-1$ parameters\footnote{The number of real parameters to specify an isometry operator $V\in\isometry{d}{D}$ can be derived with two independent ways, which outputs the same number:
\begin{enumerate}
    \item An arbitrary $d\times D$ complex matrix can be represented by $2Dd$ real parameters. Isometry operator $V\in\isometry{d}{D}$ is defined by a $d\times D$ complex matrix satisfying $V^\dagger V = \1_d$, which is given by $d^2$ independent conditions on real parameters.
    Subtracting the number of constraints $d^2$ and the degree of freedom of the global phase $1$ from $2Dd$, we obtain $2Dd-d^2-1$.
    \item An isometry operator $V\in\isometry{d}{D}$ can be represented by $d$ orthonormal $D$-dimensional vectors $\{\ket{v_1}, \ldots, \ket{v_d}\} \subset \CC^D$.
    We associate real parameters to represent $\ket{v_i}$ recursively as follows.
    The vector $\ket{v_1}$ is a unit norm $D$-dimensional complex vector, which can be represented by $2D-2$ real parameters by ignoring the global phase.
    The vector $\ket{v_{i+1}}$ is a unit norm $D$-dimensional complex vector that is orthogonal with $\ket{v_1}, \ldots, \ket{v_i}$, which can be represented by $2(D-i)-1$ real parameters.
    In total, an isometry operator $V\in\isometry{d}{D}$ can be represented by $2D-2+\sum_{i=1}^{d-1} [2(D-i)-1] = 2Dd-d^2-1$ real parameters.
\end{enumerate}}, this conjecture leads to the conclusion that the program cost for an isometry operator is given by
\begin{align}
    \label{eq:conjectured_program_cost_isometry}
    c_P^{(\mathrm{conj})} = {2Dd-d^2-1 \over 2} \log\Theta(\epsilon^{-1}),
\end{align}
which is strictly larger than Eq.~\eqref{eq:program_cost_isometry_channel} obtained by the PBT-based strategy (see the SM~\cite{supple}) as long as $D>d$ holds.
Therefore, our work falsifies the conjecture for the case of $D>d$.
We conjecture that Eq.~\eqref{eq:conjectured_program_cost_unitary} holds if we restrict the dSAR protocol to the estimation-based strategy, which is a natural class of protocols that includes the optimal protocol for unitary channels (see the SM~\cite{supple}).
We leave it for future work to prove or falsify this conjecture.
We also leave it open to obtain the optimal programming cost for the isometry channels and the CPTP maps.

Another future work is to investigate the SAR of multiple copies of the input channel, which retrieves $\Lambda^{\otimes m}$ from $n$ queries to a quantum channel $\Lambda$ for $m\geq 2$.
In this work, we consider the SAR of a single copy of the input channel.
The classical strategy is straightforwardly extended to multiple copies since we can copy the estimator.
The quantum strategy can also be extended to multiple copies by using the multi port-based teleportation, which teleports $m$ qudit states simultaneously via $n$ ports~\cite{sergii2013generalized,mozrzymas2021optimal,kopszak2021multiport,studzinski2022efficient}.
Intuitively, the classical strategy is expected to be more competitive against the quantum strategy for multiple copies since we can use multiple copies of the estimator, but its performance is limited by no-cloning theorem of unitary channels~\cite{chiribella2008optimal,bisio2014optimal}.
We leave it a future work to compare the classical and quantum strategies for the SAR of multiple copies of the input channel.

{\it Acknowledgments.}---
We acknowledge fruitful discussions with Dmitry Grinko.
This work was supported by the MEXT
Quantum Leap Flagship Program (MEXT QLEAP) JPMXS0118069605,
JPMXS0120351339, Japan Society for the Promotion of Science (JSPS) KAKENHI Grants No. 23KJ0734, No.
21H03394 and No. 23K21643, JST CREST Grant Number JPMJCR25I5, FoPM, WINGS Program, the University of Tokyo, DAIKIN Fellowship Program, the University of Tokyo, and IBM Quantum.

\onecolumngrid

\begin{center}
\textbf{End Matter}
\end{center}

\twocolumngrid

This End Matter shows that the parameter estimation of a family of isometry operators obeys the SQL based on the ``Hamiltonian-not-in-Kraus-span'' (HNKS) condition shown in Ref.~\cite{zhou2021asymptotic} and the van Trees inequality~\cite{van2001detection}.
We then show the SQL for isometry estimation using the SQL for the parameter estimation.
We provide another proof based on representation-theoretic arguments shown in Ref.~\cite{chiribella2005optimal} in the SM~\cite{supple}, which provides the exact coefficient of the leading term in Eq.~\eqref{eq:isometry_estimation_fidelity}.

A single-parameter estimation of a quantum channel is the task to estimate a parameter $\theta$ of a given quantum channel $\Lambda_\theta$ from the set $\{\Lambda_\theta \mid \theta\in \Theta \subset \RR\}$.
Suppose $q(\theta)$ is a prior probability distribution of $\theta$, and $p(\hat{\theta}|\theta)$ is the probability distribution of the estimate $\hat{\theta}$.
We define the estimation error of $\theta$ by
\begin{align}
    \delta \theta\coloneqq \sqrt{\int_\Theta \dd \theta q(\theta) \int_\Theta \dd \hat{\theta} p(\hat{\theta}|\theta) (\hat{\theta}-\theta)^2}.
\end{align}
Then, the van Trees inequality~\cite{van2001detection} shows
\begin{align}
    \delta\theta^2
    &\geq {1\over I_{q}+\int_\Theta \dd \theta q(\theta) I_p(\theta)},
\end{align}
where $I_p(\theta)$ is the Fisher information defined by
\begin{align}
    I_p(\theta) \coloneqq \int_\Theta \dd \hat{\theta} {\dot{p}(\hat{\theta}|\theta)^2\over p(\hat{\theta}|\theta)},
\end{align}
$I_q$ is defined by
\begin{align}
    I_q \coloneqq \int_\Theta \dd \theta {\dot{q}(\theta)^2 \over q(\theta)},
\end{align}
and $\dot{x}$ represents the differential of $x$ with respect to $\theta$.
We define the quantum Fisher information $I_n(\Lambda_\theta)$ of $\Lambda_\theta$ by the maximum value of $I_p(\theta)$ along all the estimator $\hat{\theta}$ implementable with $n$ queries of $\Lambda_\theta$.
Then, the van Trees inequality shows
\begin{align}
    \delta\theta^2 \geq {1\over I_{q}+\int_\Theta \dd \theta q(\theta) I_n(\Lambda_\theta)}.
\end{align}

Reference~\cite{zhou2021asymptotic} shows that the HNKS condition determines whether the QFI obeys the SQL $I_n(\Lambda_\theta) = \Theta(n)$ or the HL $I_n(\Lambda_\theta) = \Theta(n^2)$.
For a parametrized isometry channel $\Lambda_\theta(\cdot) \coloneqq V_\theta \cdot V_\theta^\dagger$, the HNKS condition is described as follows:
We define the Hamiltonian $H$ by
\begin{align}
    H_\theta \coloneqq i V_\theta^\dagger \dot{V}_\theta.
\end{align}
Then, the HNKS condition is given by
\begin{align}
    H_\theta \propto \1_d &\Leftrightarrow I_n(V_\theta) = \Theta(n),\\
    H_\theta \not\propto \1_d &\Leftrightarrow I_n(V_\theta) = \Theta(n^2).
\end{align}

Suppose $D=d+1$, and we define a family of the isometry operators $\{V_\theta\}_{\theta\in [0,\pi]} \subset \isometry{d}{d+1}$ by
\begin{align}
\label{eq:def_V_theta}
    V_\theta \coloneqq
    \left(
    \begin{matrix}
        1 & 0 & \cdots & 0 & 0\\
        0 & 1 & \cdots & 0 & 0\\
        \vdots & \vdots & \ddots & \vdots & \vdots\\
        0 & 0 & \cdots & 1 & 0\\
        0 & 0 & \cdots & 0 & \cos\theta\\
        0 & 0 & \cdots & 0 & \sin\theta
    \end{matrix}
    \right).
\end{align}
This isometry operator can be embedded into a $(d+1)$-dimensional unitary operator $U_\theta$ defined by
\begin{align}
\label{eq:def_U_theta}
    U_\theta \coloneqq
    \left(
    \begin{matrix}
        1 & 0 & \cdots & 0 & 0 & 0\\
        0 & 1 & \cdots & 0 & 0 & 0\\
        \vdots & \vdots & \ddots & \vdots & \vdots & \vdots\\
        0 & 0 & \cdots & 1 & 0 & 0\\
        0 & 0 & \cdots & 0 & \cos\theta & -\sin\theta\\
        0 & 0 & \cdots & 0 & \sin\theta & \cos\theta
    \end{matrix}
    \right).
\end{align}
The isometry operator $V_\theta$ can be considered as a unitary operator $U_\theta$ where we can only input a quantum state from a $d$-dimensional subspace of $\CC^{d+1}$.
For the isometry operator $V_\theta$ defined in Eq.~\eqref{eq:def_V_theta}, we obtain $H_\theta = 0$, i.e., $I_n(V_\theta) = \Theta(n)$ holds.
When $\theta$ is drawn from the uniform distribution of $[0, \pi]$, i.e., $q(\theta) = 1/\pi$, $I_q$ is given by $I_q = 0$.
Then, the van Trees inequality shows that
\begin{align}
    \delta\theta^2\geq \Omega(n^{-1}),
\end{align}
which shows the SQL of the Bayesian parameter estimation of isometry channels.
On the other hand, for the unitary operator $U_\theta$ defined in Eq.~\eqref{eq:def_U_theta}, we obtain $H_\theta = i\ketbra{d}{d-1}-i\ketbra{d-1}{d}$, i.e., $I_n(U_\theta) = \Theta(n^2)$ holds.

We investigate the relationship between the estimation error of $\theta$ and the channel fidelity to show the SQL of isometry estimation.
As shown in the SM~\cite{supple}, the optimal estimation fidelity of isometry estimation is achieved by a parallel covariant protocol, which satisfies
\begin{align}
    p(\hat{V}|V) = p(W \hat{V} U | WVU) \quad \forall U\in \SU(d), W\in \SU(D),
\end{align}
where $p(\hat{V}|V)$ is the probability distribution of the estimator $\hat{V}$ when the input isometry channel is given by $V$.
In this case, estimation fidelity for a given $V\in \isometry{d}{D}$ defined by
\begin{align}
    F(V) \coloneqq \int_{\isometry{d}{D}} \dd \hat{V} p(\hat{V}|V) F_\mathrm{ch}(V, \hat{V})
\end{align}
is shown to be equal to the average-case estimation fidelity.
Then, by applying the optimal isometry estimation protocol for an isometry operator $V_\theta$ defined in Eq.~\eqref{eq:def_V_theta}, we obtain
\begin{align}
    F(V_\theta) = F_\mathrm{est}(n,d,D).
\end{align}
We construct the estimator $\hat{\theta}$ from the estimator $\hat{V}$ by
\begin{align}
\label{eq:def_hat_theta}
    \hat{\theta}\coloneqq \arg\sup_{\hat{\theta}\in [0,\pi)} F_\mathrm{ch}(\hat{V}, V_{\hat{\theta}}).
\end{align}
The channel fidelity of isometry operators $V_1, V_2 \in \isometry{d}{D}$ is given by
\begin{align}
    \sqrt{1- F_\mathrm{ch}(V_1, V_2)} = {1\over 2d}\|\dketbra{V_1}-\dketbra{V_2}\|_1
\end{align}
using the $1$-norm $\|\cdot\|_1$.
Using the triangle equality for the $1$-norm, we obtain
\begin{align}
    &\sqrt{1- F_\mathrm{ch}(V_{\hat{\theta}}, V_{\theta})}\nonumber\\
    &\leq \sqrt{1-F_\mathrm{ch}(\hat{V}, V_{\theta})}+\sqrt{1-F_\mathrm{ch}(\hat{V}, V_{\hat{\theta}})}\\
    &\leq 2\sqrt{1-F_\mathrm{ch}(\hat{V}, V_{\theta})},
\end{align}
where we use the definition~\eqref{eq:def_hat_theta} of $\hat{\theta}$.
Thus, we obtain
\begin{align}
    1-F_\mathrm{ch}(\hat{V}, V_{\theta})
    &\geq {1\over 4}\left[1- F_\mathrm{ch}(V_{\hat{\theta}}, V_{\theta})\right]\\
    \label{eq:calculate_F_ch}
    &= {1\over 4}\left[1- \left(1-{2\over d}\sin^2 {\hat{\theta}-\theta \over 2}\right)^2\right]\\
    &= {1\over d} \sin^2 {\hat{\theta}-\theta \over 2} - {1\over d^2} \sin^4 {\hat{\theta}-\theta \over 2}\\
    \label{eq:sin^4_less_than_sin^2}
    &\geq {d-1\over d^2} \sin^2 {\hat{\theta}-\theta \over 2}\\
    \label{eq:sin_less_than_linear}
    &\geq {d-1\over 4\pi^2 d^2} (\hat{\theta}-\theta)^2, 
\end{align}
where we use the definition~\eqref{eq:def_V_theta} of $V_\theta$ in Eq.~\eqref{eq:calculate_F_ch}, and the inequalities shown below in Eqs.~\eqref{eq:sin^4_less_than_sin^2} and \eqref{eq:sin_less_than_linear}:
\begin{align}
    \sin^2 x \geq \sin^4 x, \quad \sin x\geq {x\over \pi} \quad \forall x\in \left[0, {\pi\over 2}\right].
\end{align}
Therefore, we obtain
\begin{align}
    &1-F_\mathrm{est}(n,d,D)\nonumber\\
    &= \int_\Theta \dd \theta q(\theta) \int_{\isometry{d}{D}} \dd \hat{V} p(\hat{V}|V_\theta) [1-F_\mathrm{ch}(\hat{V}, V_{\theta})]\\
    &\geq {d-1\over 4\pi^2 d^2} \int_\Theta \dd \theta q(\theta) \int_{\isometry{d}{D}} \dd \hat{V} p(\hat{V}|V_\theta) (\hat{\theta}-\theta)^2\\
    &= {d-1\over 4\pi^2 d^2}\delta\theta^2\\
    &\geq \Theta(n^{-1}),
\end{align}
i.e., $F_\mathrm{est}(n,d,D)\leq 1-\Omega(n^{-1})$ holds.

\clearpage

\setcounter{definition}{0}
\setcounter{table}{0}
\renewcommand{\thedefinition}{S\arabic{definition}}
\renewcommand{\thetable}{S\arabic{table}}

\onecolumngrid
\appendix

\tableofcontents

\section{Summary of notations}
We summarize the mathematical notations used in Appendix (see Tab.~\ref{tab:notations}).

\begin{table}[h]
    \centering
    \caption{Summary of mathematical notations}
    \begin{tabular}{|c|c|}\hline
        $\mcL(\mcX)$ & The set of linear operators on a Hilbert space $\mcX$.\\
        $\1_\mcX$ & The identity operator on $\mcX$.\\
        $\1_d$ & Abbreviation of $\1_{\CC^d}$.\\
        $\dim \mcX$ & The dimension of $\mcX$.\\
        $\abs{\mathbb{X}}$ & The cardinality of a set $\mathbb{X}$.\\
        $\SU(d)$ & The special unitary group.\\
        $\isometry{d}{D}$ & The set of isometry operators $\isometry{d}{D}\coloneqq \{V:\CC^d\to\CC^D \mid V^\dagger V = \1_d\}$.\\
        $\mfS_n$ & The symmetric group.\\
        $\young{d}{n}$ & The set of Young diagrams [see Eq.~\eqref{eq:def_Ydn}].\\
        $\mathrm{STab}(\alpha)$ & The set of standard tableaux with frame $\alpha$ [see Eq.~\eqref{eq:def_STab}].\\
        $\alpha\pm_d\square$ & See Eqs.~\eqref{eq:def_alpha+square} and \eqref{eq:def_alpha-square} for the definitions.\\
        $\mcU_\alpha$ & The irreducible representation space of $\SU(d)$ corresponding to a Young tableau $\alpha$ [see Eq.~\eqref{eq:decomp_space}].\\
        $\mcS_\alpha$ & The irreducible representation space of $\mfS_n$ corresponding to a Young tableau $\alpha$ [see Eq.~\eqref{eq:decomp_space}].\\
        $d_\alpha^{(d)}$ & The dimension of $\mcU_\alpha$ [see Eq.~\eqref{eq:d_alpha^d}].\\
        $m_\alpha$ & The dimension of $\mcS_\alpha$ [see Eq.~\eqref{eq:m_mu_STab}].\\
        $J_{\Lambda}$ & Choi matrix of a quantum channel $\Lambda$ [see Eq.~\eqref{eq:def_Choi_matrix}].\\
        $A \ast B$ & The link product of $A$ and $B$ [see Eq.~\eqref{eq:def_link_product}].\\\hline
    \end{tabular}
    \label{tab:notations}
\end{table}

\section{Review on the Young diagrams and the Schur-Weyl duality}

This section reviews notations and properties of Young diagrams and Schur-Weyl duality which are necessary for the proof.
We suggest the standard textbooks, e.g. Refs.~\cite{fulton1997young, georgi2000lie, ceccherini2010representation}, for more detailed reviews.
We also show Lemma~\ref{lem:isometry_irrep_product} used for the proof of asymptotically optimal isometry estimation in Appendix~\ref{appendix_sec:proof_asymptotic_fidelity_isometry_estimation}.

\subsection{Definition of the Young diagrams}
We define the set $\young{d}{n}$ by
\begin{align}
\label{eq:def_Ydn}
    \young{d}{n}\coloneqq \left\{\alpha = (\alpha_1, \ldots, \alpha_d) \in \mathbb{Z}^d \; \middle| \; \alpha_1\geq \cdots \geq \alpha_d \geq 0, \sum_{i=1}^{d}\alpha_i = n\right\},
\end{align}
where $\mathbb{Z}$ is the set of integers.
An element $\alpha\in \young{d}{n}$ is called a Young diagram.
For a Young diagram, we define the sets
\begin{align}
\label{eq:def_alpha+square}
    \alpha+_d \square &\coloneqq \{\alpha+e_i \mid i\in \{1,\ldots,d\}\} \cap \young{d}{n+1},\\
\label{eq:def_alpha-square}
    \alpha-_d \square &\coloneqq \{\alpha-e_i \mid i\in\{1,\ldots,d\}\}\cap \young{d}{n-1},
\end{align}
where $e_i$ is defined by $e_i \coloneqq (\delta_{ij})_{j=1}^{d}$ and $\delta_{ij}$ is Kronecker's delta defined by $\delta_{ii}=1$ and $\delta_{ij}=0$ for $i\neq j$.
We define a standard tableau by a sequence of Young diagrams $(\alpha_1, \ldots, \alpha_n)$ satisfying
\begin{align}
    \alpha_1 = \square, \quad \alpha_{i+1}\in \alpha_i +_d \square \quad \forall i\in \{1, \ldots, n-1\}.
\end{align}
We call $\alpha_n$ by the frame of a standard tableau $(\alpha_1, \ldots, \alpha_n)$, and define the set of standard tableaux with frame $\alpha$ by
\begin{align}
\label{eq:def_STab}
    \mathrm{STab}(\alpha) \coloneqq \{(\alpha_1, \ldots, \alpha_n) \mid \alpha_1=\square, \quad \alpha_{i+1}\in \alpha_i +_d \square \quad \forall i\in \{1, \ldots, n-1\}, \quad \alpha_n = \alpha \}.
\end{align}

\subsection{Schur-Weyl duality}

We consider the following representations on $(\CC^d)^{\otimes {n} }$ of the special unitary group $\SU(d)$ and the symmetric group $\mfS_n$:
\begin{align}
\label{eq:representation_unitary}
    &\SU(d)\ni U \mapsto U^{\otimes n} \in \mcL(\CC^d)^{\otimes { n} },\\
    &\mfS_{n} \ni \sigma \mapsto P_\sigma \in \mcL(\CC^d)^{\otimes {n} },
\end{align}
where $P_\sigma$ is a permutation operator defined as
\begin{align}
\label{eq:def_permutation_operator}
    P_\sigma\ket{i_1\cdots i_{n}} \coloneqq \ket{i_{\sigma^{-1}(1)}\cdots i_{\sigma^{-1}(n)}}
\end{align}
for the computational basis $\{\ket{i}\}_{i=1}^{d}$ of $\CC^d$.
Then, these representations are decomposed simultaneously as follows\footnote{To be more strict, Eq.~\eqref{eq:decomp_space} should be regarded as an isomorphism between the representation spaces of $\SU(d)\times \mfS_n$.
Using the isomorphism $V_\mathrm{Sch}: (\CC^d)^{\otimes n} \to \bigoplus_{\alpha\in\young{d}{{n}}} \mcU_{\alpha} \otimes \mcS_{\alpha}$, Eqs.~\eqref{eq:def_U_mu} and \eqref{eq:def_sigma_mu} should be written as
\begin{align}
    V_\mathrm{Sch} U^{\otimes n} V_\mathrm{Sch}^\dagger &= \bigoplus_{\alpha \in \young{d}{{ n}}} U_{\alpha} \otimes \1_{\mcS_{\alpha}},\\
    V_\mathrm{Sch} P_\sigma V_\mathrm{Sch}^\dagger &= \bigoplus_{\alpha \in\young{d}{{ n}}} \1_{\mcU_{\alpha}} \otimes \sigma_{\alpha}.
\end{align}
For simplicity, in the rest of this paper, we use the symbol `$=$' for the isomorphism between the spaces and omit $V_\mathrm{Sch}$.}:
\begin{align}
    (\CC^d)^{\otimes {n}} &= \bigoplus_{\alpha\in\young{d}{{n}}} \mcU_{\alpha} \otimes \mcS_{\alpha},\label{eq:decomp_space}\\
    U^{\otimes {n}} &= \bigoplus_{\alpha \in \young{d}{{ n}}} U_{\alpha} \otimes \1_{\mcS_{\alpha}},\label{eq:def_U_mu}\\
    P_\sigma &= \bigoplus_{\alpha \in\young{d}{{ n}}} \1_{\mcU_{\alpha}} \otimes \sigma_{\alpha},\label{eq:def_sigma_mu}
\end{align}
where $\alpha$ runs in the set $\young{d}{n}$ defined in Eq.~\eqref{eq:def_Ydn}, $\SU(d)\ni U\mapsto U_{\alpha}\in \mcL(\mcU_{\alpha})$ is an irreducible representation of $\SU(d)$, and $\mfS_{n}\ni \sigma \mapsto \sigma_{\alpha}\in \mcL(\mcS_{\alpha})$ is an irreducible representation of $\mfS_n$.
This relation shows that any operator commuting with $U^{\otimes n}$ for all $U\in\SU(d)$ can be written as a linear combination of $\{V_\sigma\}_{\sigma\in\mfS_n}$, which is called the Schur-Weyl duality.
The dimension of $\mcU_\alpha^{(d)}$ is given by~\cite{itzykson1966unitary}
\begin{align}
\label{eq:d_alpha^d}
    d_\alpha^{(d)}\coloneqq \dim \mcU_\alpha = {\prod_{1\leq i<j\leq d}(\alpha_i-\alpha_j-i+j) \over \prod_{k=1}^{d-1} k!},
\end{align}
and the dimension of $\mcS_\alpha$ is denoted by $m_\alpha$, which will be given later in Eq.~\eqref{eq:m_mu_STab}\footnote{ The dimension of $\mcS_\alpha$ is denoted by $m_\alpha$ since $\mcS_\alpha$ can be regarded as a multiplicity space for irreducible representation $\mcU_\alpha$ [see also Eq.~\eqref{eq:S_mu_multi}].}.
The tensor product representation $U_\alpha\otimes U$ satisfies
\begin{align}
    U_\alpha\otimes U = \bigoplus_{\mu\in\alpha+_d \square} U_\mu\otimes \ketbra{\alpha}_{\mathrm{multi}},
\end{align}
where $\{\ket{\alpha}\}_{\alpha\in\young{d}{n}}$ is an orthonormal basis in a multiplicity space, which shows the decomposition of the space $\mcU_\alpha \otimes \CC^d$ as
\begin{align}
\label{eq:decomp_tensor_product_U_alpha}
    \mcU_\alpha \otimes \CC^d = \bigoplus_{\mu\in\alpha+_d\square} \mcU_\mu\otimes \mathrm{span}\{\ket{\alpha}_{\mathrm{multi}}\}.
\end{align}
Since
\begin{align}
    \bigoplus_{\mu\in\young{d}{n+1}} U_\mu \otimes \1_{\mcS_\mu}
    &= U^{\otimes n}\otimes U\\
    &= \bigoplus_{\alpha\in\young{d}{n}} U_\alpha\otimes U\otimes \1_{\mcS_\alpha}\\
    &= \bigoplus_{\mu\in\young{d}{n+1}} U_\mu \otimes \bigoplus_{\alpha\in\mu-_d \square}\ketbra{\alpha}_{\mathrm{multi}}\otimes \1_{\mcS_\alpha}.
\end{align}
holds, we obtain
\begin{align}
\label{eq:decomposition_S_mu}
    \mcS_\mu = \bigoplus_{\alpha\in\mu-_d\square}\mcS_\alpha\otimes \mathrm{span}\{\ket{\alpha}_{\mathrm{multi}}\}.
\end{align}
Applying this equation recursively for $n$, we obtain
\begin{align}
\label{eq:S_mu_multi}
    \mcS_\mu = \bigoplus_{(\mu_1,\ldots,\mu_{n+1})\in \mathrm{STab}(\mu)} \mathrm{span}(\ket{\mu_1}\otimes \cdots \otimes \ket{\mu_{n+1}}),
\end{align}
where we omit the subscript `multi' for simplicity.
Therefore, the dimension of $\mcS_\mu$ is given by
\begin{align}
\label{eq:m_mu_STab}
    m_\mu \coloneqq \dim \mcS_\mu = \abs{\mathrm{STab}(\mu)}.
\end{align}
We define the Young-Yamanouchi basis $\{\ket{s_\mu}\}_{s_\mu\in\mathrm{STab}(\mu)}$ of $\mcS_\mu$ by
\begin{align}
\label{eq:yy_basis}
    \ket{s_\mu} \coloneqq \ket{\mu_1}\otimes \cdots \otimes \ket{\mu_{n+1}}
\end{align}
for $s_\mu = (\mu_1,\ldots, \mu_{n+1})$.

\subsection{Schur-Weyl duality applied for isometry channels}
\label{appendix_sec:shur_weyl_isometry}
As shown in Refs.~\cite{yoshida2023universal, yoshida2025universal}, the $n$-fold isometry operator $V^{\otimes n}$ for $V\in\isometry{d}{D}$ can be decomposed as
\begin{align}
\label{eq:isometry_sw_duality}
    V^{\otimes n} = \bigoplus_{\alpha\in\young{d}{n}} V_\alpha \otimes \1_{\mcS_\mu},
\end{align}
where $V_\alpha:\mcU_\alpha^{(d)}\to \mcU_\alpha^{(D)}$ is an isometry operator.
The isometry operator $V_\alpha$ has the following property:
\begin{lemma}
\label{lem:isometry_irrep_product}
    For any $\alpha\in\young{d}{n}$ and $V\in\isometry{d}{D}$,
    \begin{align}
        V_\alpha \otimes V = \bigoplus_{\mu\in\alpha+_d\square} V_\mu \otimes \ketbra{\alpha}_{\mathrm{multi}}
    \end{align}
    holds.
\end{lemma}
\begin{proof}
    Since
    \begin{align}
        \mcU_\alpha^{(d)} \otimes \CC^d &= \bigoplus_{\mu\in\alpha+_d\square} \mcU_\mu^{(d)}\otimes \mathrm{span}\{\ket{\alpha}_\mathrm{multi}\},\\
        \mcU_\alpha^{(D)} \otimes \CC^D &= \bigoplus_{\mu\in\alpha+_D\square} \mcU_\mu^{(D)}\otimes \mathrm{span}\{\ket{\alpha}_\mathrm{multi}\}
    \end{align}
    hold, any linear operator $O: \mcU_\alpha^{(d)} \otimes \CC^d \to \mcU_\alpha^{(D)} \otimes \CC^D$ can be decomposed into
    \begin{align}
        O = \bigoplus_{\mu\in \alpha+_d\square, \nu\in\alpha+_D\square} O_{\mu\to \nu} \otimes \ketbra{\alpha}_{\mathrm{multi}}
    \end{align}
    using linear operators $O_{\mu\to \nu}: \mcU_\mu^{(d)} \to \mcU_\nu^{(D)}$.
    Thus, $V_\alpha\otimes V$ can be decomposed into
    \begin{align}
    \label{V_alpha_times_V_decomposition}
        V_\alpha\otimes V = \bigoplus_{\mu\in \alpha+_d\square, \nu\in\alpha+_D\square} V^{\alpha}_{\mu\to \nu} \otimes \ketbra{\alpha}_{\mathrm{multi}},
    \end{align}
    using linear operators $V^{\alpha}_{\mu\to \nu}: \mcU_\mu^{(d)} \to \mcU_\nu^{(D)}$.
    On the other hand, by using Eq.~\eqref{eq:isometry_sw_duality} for $n$ and $n+1$ and the decomposition of the space $\mcS_\mu$ in Eq.~\eqref{eq:decomposition_S_mu}, we obtain
    \begin{align}
        V^{\otimes n+1}
        &= V^{\otimes n} \otimes V\\
        &= \bigoplus_{\alpha\in\young{d}{n}} V_\alpha \otimes V \otimes \1_{\mcS_\mu}\\
        &= \bigoplus_{\alpha\in\young{d}{n}} \bigoplus_{\mu\in \alpha+_d\square, \nu\in\alpha+_D\square} V^{\alpha}_{\mu\to \nu} \otimes \ketbra{\alpha}_{\mathrm{multi}} \otimes \1_{\mcS_\alpha},\\
        V^{\otimes n+1} &= \bigoplus_{\mu\in\young{d}{n+1}} V_\mu \otimes \1_{\mcS_\mu}\\
        &= \bigoplus_{\mu\in\young{d}{n+1}} V_\mu \otimes \bigoplus_{\alpha\in\mu-_d\square} \ketbra{\alpha}_{\mathrm{multi}}\otimes \1_{\mcS_\alpha}\\
        &= \bigoplus_{\substack{(\alpha, \mu)\in \young{d}{n}\times \young{d}{n+1}\\ \mathrm{s.t.}\; \mu\in\alpha+_d \square}} V_\mu \otimes \ketbra{\alpha}_{\mathrm{multi}}\otimes \1_{\mcS_\alpha}\\
        &= \bigoplus_{\alpha\in\young{d}{n}} \bigoplus_{\mu\in\alpha+_d\square} V_\mu \otimes  \ketbra{\alpha}_{\mathrm{multi}}\otimes \1_{\mcS_\alpha}.
    \end{align}
    Thus, we obtain
    \begin{align}
        \bigoplus_{\alpha\in\young{d}{n}} \bigoplus_{\mu\in \alpha+_d\square, \nu\in\alpha+_D\square} V^{\alpha}_{\mu\to \nu} \otimes \ketbra{\alpha}_{\mathrm{multi}} \otimes \1_{\mcS_\alpha}
        = \bigoplus_{\alpha\in\young{d}{n}} \bigoplus_{\mu\in\alpha+_d\square} V_\mu \otimes  \ketbra{\alpha}_{\mathrm{multi}}\otimes \1_{\mcS_\alpha},
    \end{align}
    i.e.,
    \begin{align}
        V^{\alpha}_{\mu\to \nu} = \delta_{\mu\nu} V_\mu
    \end{align}
    holds.
    Substituting this expression into \eqref{V_alpha_times_V_decomposition}, we obtain
    \begin{align}
        V_\alpha \otimes V = \bigoplus_{\mu\in\alpha+_d\square} V_\mu \otimes \ketbra{\alpha}_{\mathrm{multi}}.
    \end{align}
\end{proof}

\section{Review on quantum testers}

This section reviews notations and properties of quantum testers, based on the Choi representation.
We suggest Ref.~\cite{taranto2025higher} for more detailed reviews.

\subsection{Choi representation}

We consider a quantum channel $\Lambda: \mcL(\mcI) \to \mcL(\mcO)$, where $\mcI$ and $\mcO$ are the Hilbert spaces corresponding to the input and output systems.
The Choi matrix $J_{\Lambda} \in \mcL(\mcI\otimes \mcO)$ is defined by
\begin{align}
\label{eq:def_Choi_matrix}
    J_{\Lambda}\coloneqq \sum_{i,j} \ketbra{i}{j}_{\mcI} \otimes \Lambda(\ketbra{i}{j})_{\mcO},
\end{align}
where $\{\ket{i}\}_i$ is the computational basis of $\mcI$, and the subscripts $\mcI$ and $\mcO$ represent the Hilbert spaces where each term is defined.
The Choi matrix of the unitary channel $\mcU(\cdot) = U(\cdot) U^\dagger$ is given by $J_{\mcU} = \dketbra{U}$, where $\dket{U}$ is the dual ket defined by
\begin{align}
\label{eq:def_dual_vector}
    \dket{U}\coloneqq \sum_{i} \ket{i} \otimes U\ket{i}.
\end{align}
The complete positivity and the trace-preserving property of $\Lambda$ are represented in the Choi matrix by
\begin{align}
    J_\Lambda \geq 0,\quad \Tr_{\mcO} J_{\Lambda} = \1_{\mcI}.
\end{align}
In the Choi representation, the composition of a quantum channel $\Lambda$ with a quantum state $\rho$ and that of quantum channels $\Lambda_1, \Lambda_2$ are represented in a unified way using a link product $\ast$ as
\begin{align}
    \Lambda(\rho) &= J_{\Lambda} \ast \rho,\\
    J_{\Lambda_2\circ \Lambda_1} &= J_{\Lambda_1} \ast J_{\Lambda_2},
\end{align}
where the link product $\ast$ for $A\in\mcL(\mcX\otimes \mcY)$ and $B\in \mcL(\mcY\otimes \mcZ)$ is defined as~\cite{chiribella2008quantum}
\begin{align}
\label{eq:def_link_product}
    A \ast B \coloneqq \Tr_\mcY [(A^{\sfT_\mcY}\otimes \1_{\mcZ})(\1_{\mcX} \otimes B)],
\end{align}
and $A^{\sfT_\mcY}$ is the partial transpose of $A$ over the subsystem $\mcY$.

\subsection{Quantum testers}

A quantum tester is a multi-linear transformation from multiple quantum channels to a probability distribution.
The set of quantum testers contains the set of protocols allowed in the quantum circuit framework~\cite{chiribella2008quantum, wechs2021quantum} as a subset.
It also contains protocols beyond the quantum circuit framework, called the indefinite causal order protocols~\cite{hardy2007towards, oreshkov2012quantum, chiribella2013quantum}.

We consider quantum channels $\Lambda_i: \mcL(\mcI_i) \to \mcL(\mcO_i)$ for $i\in\{1, \cdots, n\}$, and define a multi-linear transformation $\mcT_a$ from the quantum channels $\Lambda_1,\ldots, \Lambda_n$ to a probability distribution:
\begin{align}
    p_a = \mcT_a[\Lambda_1, \ldots, \Lambda_n],
\end{align}
where $\{p_a\}_a$ is a probability distribution.
The action of $\mcT_a$ can be represented by a matrix $T_a$ as
\begin{align}
    p_a = T_a\ast (J_{\Lambda_1}\otimes \cdots \otimes J_{\Lambda_n}),
\end{align}
and $\{T_a\}_a$ is called a quantum tester.
The quantum tester $\{T_a\}_a$ should satisfy the following two properties:
\begin{itemize}
    \item Completely CP preserving: For any auxiliary Hilbert spaces $\mcA_i, \mcA'_i$ and any completely positive (CP) maps $\Lambda_i: \mcL(\mcI_i\otimes \mcA_i) \to \mcL(\mcO_i\otimes \mcA'_i)$ for $i\in\{1,\ldots,n\}$,
    \begin{align}
        (\1\otimes T_a)\ast (J_{\Lambda_1}\otimes \cdots \otimes J_{\Lambda_n}) \geq 0
    \end{align}
    holds.
    \item TP preserving: For any trace preserving (TP) maps $\Lambda_i: \mcL(\mcI_i) \to \mcL(\mcO_i)$ for $i\in\{1,\ldots,n\}$,
    \begin{align}
        \sum_a T_a\ast (J_{\Lambda_1}\otimes \cdots \otimes J_{\Lambda_n}) = 1
    \end{align}
    holds.
\end{itemize}
The completely CP preserving property is equivalent to
\begin{align}
\label{eq:tester_positivity}
    T_a\geq 0.
\end{align}
The TP preserving property is characterized as affine conditions on $\sum_a T_a$ (see, e.g., Ref.~\cite{oreshkov2012quantum} for the complete characterization).

\section{The optimal retrieval error of the estimation-based and PBT-based strategies for dSAR of isometry channels}

This section shows the following Lemmas on the optimal retrieval error of the estimation-based and PBT-based strategies for dSAR of isometry channels.

\begin{lemma}
\label{lem:optimal_estimation-based}
    The optimal retrieval error of the estimation-based strategy for dSAR of $\mathbb{S}_\mathrm{Isometry}^{(d,D)}$ is given by
    \begin{align}
    \label{eq:estimation-based_retrieval_error}
        \epsilon = 1-F_\mathrm{est}(n,d,D).
    \end{align}
\end{lemma}

\begin{lemma}
\label{lem:optimal_PBT-based}
    The optimal retrieval error of the PBT-based strategy for dSAR of $\mathbb{S}_\mathrm{Isometry}^{(d,D)}$ is given by
    \begin{align}
        \epsilon = \delta_\mathrm{PBT}(n,d).
    \end{align}
\end{lemma}

\begin{proof}[Proof of Lem.~\ref{lem:optimal_estimation-based}]
    We extend the proof shown in Ref.~\cite{yang2020optimal} for dSAR of unitary channels to the case of isometry channels.
    
    (Achievability)
    As shown in Lem.~\ref {lem:isometry_estimation_matrix} in Appendix~\ref{appendix_sec:proof_asymptotic_fidelity_isometry_estimation}, the optimal estimation fidelity of isometry estimation is achieved by a parallel covariant protocol, which satisfies
    \begin{align}
    \label{eq:covariance_probability_distribution}
        p(\hat{V}|V) = p(W \hat{V} U | WVU) \quad \forall U\in \SU(d), W\in \SU(D),
    \end{align}
    where $p(\hat{V}|V)$ is the probability distribution of the estimator $\hat{V}$ when the input isometry channel is given by $V$.
    The retrieved channel is given by
    \begin{align}
        \mcR_V(\rho) = \int\dd \hat{V} p(\hat{V}|V) \hat{\mcV}(\rho),
    \end{align}
    where $\hat{\mcV}(\cdot)\coloneqq \hat{V}\cdot\hat{V}^\dagger$.
    By using Eq.~\eqref{eq:covariance_probability_distribution}, we obtain
    \begin{align}
        \mcR_{WVU}
        &= \int\dd \hat{V} p(\hat{V}|WVU) \hat{\mcV}\\
        \label{eq:rename}
        &= \int\dd \hat{V} p(W\hat{V}U|WVU) \mcW \circ \hat{\mcV} \circ \mcU\\
        &= \int\dd \hat{V} p(\hat{V}|V) \mcW \circ \hat{\mcV} \circ \mcU\\
        &= \mcW\circ \mcR_V\circ \mcU,
    \end{align}
    where we rename $W^\dagger \hat{V} U^\dagger$ by $\hat{V}$ in Eq.~\eqref{eq:rename}.
    By taking $W$ to be $W = VU^\dagger V^\dagger + W'$ for any unitary operator $W'$ on $(\Im V)^\perp$, where $(\Im V)^\perp$ is the orthogonal complement of the image $\Im V$ of $V$, $WVU = V$ holds and we obtain
    \begin{align}
        \mcW\circ \mcR_V\circ \mcU = \mcR_V,
    \end{align}
    i.e.,
    \begin{align}
    \label{eq:R_V_commutation}
        [J_{\mcR_V}, U^\sfT \otimes (VU^\dagger V^\dagger + W')] = 0.
    \end{align}
    Due to Schur's lemma, we obtain the decomposition of $J_{\mcR_V}$:
    \begin{align}
        J_{\mcR_V} = J_1 + p \1_d\otimes {\Pi_{(\Im V)^\perp} \over D-d},
    \end{align}
    where the support of $J_1$ is in $\CC^d\otimes \Im V$, $p\geq 0$ and $\Pi_{(\Im V)^\perp}$ is the orthogonal projector onto $(\Im V)^\perp$.
    Since the orthogonal projector $\Pi_{\Im V}$ onto $\Im V$ is given by $\Pi_{\Im V} = V V^\dagger$, the operator $J_1$ is given by
    \begin{align}
        J_1
        &= (\1_d\otimes \Pi_{\Im V}) J_{\mcR_V} (\1_d\otimes \Pi_{\Im V})\\
        &= (\1_d\otimes V V^\dagger) J_{\mcR_V} (\1_d\otimes VV^\dagger)\\
        &= (\1_{\mcL(\CC^d)} \otimes \mcV)\circ (\1_{\mcL(\CC^d)} \otimes \mcV^\dagger)(J_{\mcR_V}).
    \end{align}
    Due to Eq.~\eqref{eq:R_V_commutation}, the operator $(\1_{\mcL(\CC^d)} \otimes \mcV^\dagger)(J_{\mcR_V})$ satisfies
    \begin{align}
        [(\1_{\mcL(\CC^d)} \otimes \mcV^\dagger)(J_{\mcR_V}), U\otimes U^*] =  0 \quad \forall U\in\SU(d).
    \end{align}
    Due to Schur's lemma, we obtain the decomposition of $(\1_{\mcL(\CC^d)} \otimes \mcV^\dagger)(J_{\mcR_V})$:
    \begin{align}
        (\1_{\mcL(\CC^d)} \otimes \mcV^\dagger)(J_{\mcR_V}) = q \dketbra{\1} + r\1_d\otimes {\1_d\over d},
    \end{align}
    where $q, r\geq 0$.
    Therefore, we obtain
    \begin{align}
        J_{\mcR_V} = q\dketbra{V} + r \1_d \otimes {\Pi_{\Im V}\over d} + p \1_d \otimes {\Pi_{(\Im V)^\perp} \over D-d},
    \end{align}
    i.e.,
    \begin{align}
        \mcR_V = q \mcV + r \mcT_1 + p \mcT_2,
    \end{align}
    where $\mcT_1$ and $\mcT_2$ are trace-and-replace channels defined by
    \begin{align}
        \mcT_1(\cdot) &\coloneqq {\Pi_{\Im V}\over d} \Tr(\cdot),\\
        \mcT_2(\cdot) &\coloneqq {\Pi_{(\Im V)^\perp}\over D-d} \Tr(\cdot),
    \end{align}
    and $p, q, r$ satisfies $p+q+r=1$.
    Defining the depolarizing channel $\mcD_\eta$ for $\eta\in [0,1]$ by
    \begin{align}
        \mcD_\eta(\cdot)\coloneqq (1-\eta)\1_{\mcL(\CC^d)}(\cdot) + \eta {\1_d\over d} \Tr(\cdot),
    \end{align}
    we obtain
    \begin{align}
        \mcR_V - \mcV = (q+r) \mcV\circ (\mcD_\eta-\1_{\mcL(\CC^d)}) + p \mcT_2 - p \mcV
    \end{align}
    for $\eta = {r\over q+r}$.
    Then, we obtain
    \begin{align}
        {1\over 2}\|\mcR_V - \mcV\|_\diamond
        &\leq (q+r) {\|\mcV\circ (\mcD_\eta-\1_{\mcL(\CC^d)})\|_\diamond \over 2} + {p\over 2} \|\mcT_2\|_\diamond + {p\over 2} \|\mcV\|_\diamond\\
        &= (q+r) {d^2-1\over d^2} \eta  + p \\
        &= {d^2-1\over d^2} r + p\\
        &= {d^2-1\over d^2} r + 1-q-r\\
        &= 1-q-{r\over d^2},
    \end{align}
    where we use the concavity of the diamond norm, the invariance of the diamond norm under any isometry channel, and the fact that~\cite{matsumoto2012input}
    \begin{align}
        {1\over 2}\|\mcD_\eta - \1_{\mcL(\CC^d)}\|_\diamond = {d^2-1\over d^2} \eta.
    \end{align}
    On the other hand, the estimation fidelity is given by
    \begin{align}
        F_\mathrm{est}
        &= \int \dd \hat{V} p(\hat{V}|V) F_\mathrm{ch}(\hat{V}, V)\\
        &= F_\mathrm{ch}(\mcR_V, \mcV)\\
        &= q + {r\over d^2}.
    \end{align}
    Therefore, we obtain
    \begin{align}
        {1\over 2}\|\mcR_V - \mcV\|_\diamond\leq 1-F_\mathrm{est}.
    \end{align}

    (Optimality) The diamond norm satisfies
    \begin{align}
        {1\over 2}\|\mcR_V - \mcV\|_\diamond \geq 1-F_\mathrm{ch}(\mcR_V, \mcV),
    \end{align}
    which follows from the Fuchs-van de Graaf inequality~\cite{fuchs2002cryptographic}.
    The right-hand side is evaluated by
    \begin{align}
        F_\mathrm{ch}(\mcR_V, \mcV) = \int \dd \hat{V} p(\hat{V}|V) F_\mathrm{ch}(\hat{V}, V) = F_\mathrm{est}.
    \end{align}
    Therefore, we obtain the optimality of Eq.~\eqref{eq:estimation-based_retrieval_error}.
\end{proof}

\begin{proof}[Proof of Lem.~\ref{lem:optimal_PBT-based}]
    The retrieved channel of the PBT-based strategy for dSAR of an isometry channel $\mcV(\cdot)\coloneqq V\cdot V^\dagger$ is given by $\mcV \circ \Phi$, where $\Phi$ is the teleportation channel defined in Eq.~\eqref{eq:teleportation_channel} in the main text.
    Therefore, the retrieval error is given by
    \begin{align}
        \epsilon
        &= {1\over 2}\|\mcV\circ \Phi - \mcV\|_\diamond = {1\over 2}\|\Phi-\1_{\mcL(\CC^d)}\|_\diamond = \delta_{\mathrm{PBT}},
    \end{align}
    where we use the invariance of the diamond norm under any isometry channel.
\end{proof}

\section{Proof of Thm.~1 (Asymptotic fidelity of optimal isometry estimation)}
\label{appendix_sec:proof_asymptotic_fidelity_isometry_estimation}

We use the following two Lemmas on the optimal fidelity of isometry estimation for the proof of Thm.~1 (see Appendices \ref{appendix_sec:proof_isometry_estimation_matrix} and \ref{appendix_sec:proof_isometry_estimation_fidelity} for the proofs):
\begin{lemma}
\label{lem:isometry_estimation_matrix}
    The optimal fidelity $F_\mathrm{est}(n,d,D)$ of isometry estimation is given by the maximal eigenvalue of the $\abs{\young{d}{n}}\times \abs{\young{d}{n}}$ matrix $M_\mathrm{est}(n,d,D)$ given by
    \begin{align}
        (M_\mathrm{est}(n,d,D))_{\alpha\beta} \coloneqq {1\over d^2}\sum_{i,j=1}^{d} \delta_{\alpha+e_i, \beta+e_j} f(\alpha_i-i)f(\beta_j-j) \quad \forall \alpha, \beta\in\young{d}{n},
    \end{align}
    where $f: [-d,\infty) \to \mathbb{R}$ is defined by $f(x)\coloneqq \sqrt{x+d+1 \over x+D+1}$.
    The optimal fidelity is achieved by a parallel protocol.
\end{lemma}

\begin{lemma}
\label{lem:isometry_estimation_fidelity}
For any function $g:\NN\to\NN$ satisfying $g(n)\leq {2\over 3(d-1)}({n\over d}+d-2)$, we can implement an isometry estimation protocol with the fidelity given by
\begin{align}
\label{eq:estimation_fidelity}
    F_\mathrm{est}\geq 1-{\pi^2 (d-1)^2 \over d^2 g(n)^2}-{D-d \over {n\over d}-{d-1 \over 2} g(n)+D-d}.
\end{align}
The isometry estimation protocol constructed in this Lemma satisfies
\begin{align}
\label{eq:estimation_fidelity_converse}
    F_\mathrm{est}\leq 1-{\pi^2 (d-1)^2\over d^2 g(n)^2} - {d(D-d) \over n} + O(g(n) n^{-2}, g(n)^{-3}).
\end{align}
\end{lemma}

\begin{proof}[Proof of Thm.~1]
    By putting $g(n) = \Theta(n^{2/3})$ in Lem.~\ref{lem:isometry_estimation_fidelity}, we obtain
    \begin{align}
        F_\mathrm{est}(n,d,D) \geq 1-{d(D-d) \over n} +O(n^{-4/3}).\label{eq:lower_bound}
    \end{align}
    We show a matching upper bound on $F_\mathrm{est}(n,d,D)$ to complete the proof.

    The optimal isometry estimation fidelity $F_\mathrm{est}(n,d,D)$ is given as the maximal eigenvalue of the matrix $M_\mathrm{est}(n,d,D)$ given in Lem.~\ref{lem:isometry_estimation_matrix}.
    Since $(M_\mathrm{est}(n,d,D))_{\alpha\beta}\geq 0$ holds for all $\alpha, \beta \in\young{d}{n}$, due to the Perron-Frobenius theorem~\cite{horn2012matrix}, the maximal eigenvalue is bounded as
    \begin{align}
        F_\mathrm{est}(n,d,D)\leq \max_{\alpha} \sum_{\beta} (M_\mathrm{est}(n,d,D))_{\alpha\beta}.
    \end{align}
    Therefore, we obtain
    \begin{align}
        F_\mathrm{est}(n,d,D)
        &\leq {1\over d^2}\max_{\alpha} \sum_{i,j=1}^{d} f(\alpha_i-i)f(\alpha_j-j-1+\delta_{ij})\\
        \label{eq:f_increasing_used}
        &\leq {1\over d^2}\max_{\alpha} \sum_{i,j=1}^{d} f(\alpha_i-i)f(\alpha_j-j)\\
        &\leq \left[{1\over d}\max_{\alpha} \sum_{i=1}^{d} f(\alpha_i-i)\right]^2\\
        \label{eq:jensen_used}
        &\leq \left[f\left({n\over d}-{d+1 \over 2}\right)\right]^2\\
        &= 1-{D-d\over {n\over d}-{d+1\over 2}+D+1}\\
        &= 1-{d(D-d) \over n} +O(n^{-2}),\label{eq:upper_bound}
    \end{align}
    where Eq.~\eqref{eq:f_increasing_used} uses the property that the function $f$ is monotonically increasing, and Eq.~\eqref{eq:jensen_used} uses Jensen's inequality \cite{degroot2012probability} for the concave function $f$ and ${1\over d}\sum_i(\alpha_i-i) ={n\over d}-{d+1 \over 2}$.
\end{proof}

\subsection{Proof of Lem.~\ref{lem:isometry_estimation_matrix} (Parallel covariant form of optimal isometry channel)}
\label{appendix_sec:proof_isometry_estimation_matrix}
We show that a parallel covariant protocol can obtain the optimal fidelity of isometry estimation for a given number of queries.
We prove this statement in a constructive way, similarly to the arguments shown in Refs.~\cite{chiribella2005optimal, chiribella2009optimal, chiribella2008memory, bavaresco2022unitary, yoshida2024one}.
We show the construction of a parallel covariant protocol achieving the same fidelity as that of any estimation protocol.

Suppose a quantum tester $\{T_{\hat{V}}\dd \hat{V}\}_{\hat{V}}$ implements an isometry estimation protocol.
We show that a parallel covariant protocol can achieve the same fidelity of isometry estimation as that of the quantum tester $\{T_{\hat{V}}\dd \hat{V}\}_{\hat{V}}$.
The probability distribution of the estimator $\hat{V}$ for a given input isometry operator $V$ is given by
\begin{align}
    p(\hat{V}|V) = T_{\hat{V}} \ast \dketbra{V}^{\otimes n}_{\mcI^n \mcO^n}
\end{align}
and the estimation fidelity given by
\begin{align}
    F = \int_{\isometry{d}{D}} \dd V \int_{\isometry{d}{D}} \dd \hat{V} p(\hat{V}|V) F_\mathrm{ch}(\hat{V}, V).
\end{align}
Defining the $\SU(d)\times \SU(D)$-twirled operator $T'_{\hat{V}}$ by
\begin{align}
    T'_{\hat{V}} \coloneqq \int_{\SU(d)} \dd U \int_{\SU(D)} \dd W (U^{\otimes n}_{\mcI^n} \otimes W^{\otimes n}_{\mcO^n})T_{W^\sfT \hat{V} U}(U^{\otimes n}_{\mcI^n} \otimes W^{\otimes n}_{\mcO^n})^\dagger,
    \label{eq:def_Tprime}
\end{align}
the set of operators $\{T'_{\hat{V}} \dd \hat{V}\}_{\hat{V}}$ forms a quantum tester, and the corresponding probability distribution is given by
\begin{align}
    p'(\hat{V}|V)
    &\coloneqq T'_{\hat{V}}\ast \dketbra{V}^{\otimes n}_{\mcI^n \mcO^n}\\
    &= \int_{\SU(d)} \dd U \int_{\SU(D)} \dd W (U^{\otimes n}_{\mcI^n} \otimes W^{\otimes n}_{\mcO^n})T_{W^\sfT \hat{V} U}(U^{\otimes n}_{\mcI^n} \otimes W^{\otimes n}_{\mcO^n})^\dagger \ast \dketbra{V}^{\otimes n}_{\mcI^n \mcO^n}\\
    &= \int_{\SU(d)} \dd U \int_{\SU(D)} \dd W T_{W^\sfT \hat{V} U} \ast \dketbra{W^\sfT VU}^{\otimes n}_{\mcI^n \mcO^n}\\
    &= \int_{\SU(d)} \dd U \int_{\SU(D)} \dd W p(W^\sfT \hat{V} U|W^\sfT VU).
\end{align}
This tester achieves the same fidelity since
\begin{align}
    F'
    &\coloneqq \int_{\isometry{d}{D}} \dd V \int_{\isometry{d}{D}} \dd \hat{V} p'(\hat{V}|V) F_\mathrm{ch}(\hat{V}, V)\\
    &= \int_{\isometry{d}{D}} \dd V \int_{\isometry{d}{D}} \dd \hat{V} \int_{\SU(d)} \dd U \int_{\SU(D)} \dd W p(W^\sfT \hat{V} U|W^\sfT VU) F_\mathrm{ch}(\hat{V}, V)\\
    &= \int_{\isometry{d}{D}} \dd V \int_{\isometry{d}{D}} \dd \hat{V} \int_{\SU(d)} \dd U \int_{\SU(D)} \dd W p(\hat{V}|V) F_\mathrm{ch}(W^* \hat{V} U^{-1}, W^* VU^{-1})\label{eq:apply_haar_measure_invariance}\\
    &= \int_{\isometry{d}{D}} \dd V \int_{\isometry{d}{D}} \dd \hat{V} \int_{\SU(d)} \dd U \int_{\SU(D)} \dd W p(\hat{V}|V) F_\mathrm{ch}(\hat{V}, V)\label{eq:apply_channel_fidelity_invariance}\\
    &= \int_{\isometry{d}{D}} \dd V \int_{\isometry{d}{D}} \dd \hat{V} p(\hat{V}|V) F_\mathrm{ch}(\hat{V}, V)\\
    &= F
\end{align}
holds, where we use the invariance of the Haar measure given by $\dd V = \dd (W^\sfT V U)$ and $\dd \hat{V} = \dd (W^\sfT \hat{V} U)$ in \eqref{eq:apply_haar_measure_invariance}, and the invariance of the channel fidelity given by $F_\mathrm{ch}(W^* \hat{V} U^{-1}, W^* VU^{-1}) = F_\mathrm{ch}(\hat{V}, V)$ in \eqref{eq:apply_channel_fidelity_invariance}.
Defining $T'$ by
\begin{align}
    T'\coloneqq \int_{\isometry{d}{D}}\dd \hat{V} T'_{\hat{V}},
\end{align}
the operator $T'$ satisfies the following $\SU(d)\times \SU(D)$ symmetry:
\begin{align}
    [T', U^{\otimes n}_{\mcI^n} \otimes W^{\otimes n}_{\mcO^n}] = 0 \quad \forall U\in\SU(d), W\in\SU(D).
\end{align}
Then, $T'$ can be represented in the Schur basis as
\begin{align}
    T' = \bigoplus_{\mu\in\young{d}{n}, \nu\in\young{D}{n}}T'_{\mu\nu} \otimes (\1_{\mcU_\mu^{(d)}})_{\mcI^n} \otimes (\1_{\mcU_\nu^{(D)}})_{\mcO^n},
\end{align}
using $T'_{\mu\nu}\in\mcL(\mcS_\mu\otimes \mcS_\nu)$.
Defining $\tilde{T}'\in\mcL(\mcI^n\otimes \mcA^n)$ for $\mcA^n\coloneqq \bigotimes_{i=1}^{n}\mcA_i$ and $\mcA_i=\CC^d$ by
\begin{align}
    \tilde{T}'\coloneqq \bigoplus_{\mu\in\young{d}{n}, \nu\in\young{d}{n}}T'_{\mu\nu} \otimes (\1_{\mcU_\mu^{(d)}})_{\mcI^n} \otimes (\1_{\mcU_\nu^{(d)}})_{\mcA^n},
\end{align}
$T'$ and $\tilde{T}'$ satisfy the following condition:
\begin{align}
    (\1_{\mcI^n} \otimes V^{\otimes n}_{\mcA^n\to \mcO^n})\tilde{T}' = T' (\1_{\mcI^n} \otimes V^{\otimes n}_{\mcA^n\to \mcO^n}).
\end{align}
Similarly, $\sqrt{T'}^\sfT$ and $\sqrt{\tilde{T}'}^\sfT$ satisfy
\begin{align}
    (\1_{\mcI^n} \otimes V^{\otimes n}_{\mcA^n\to \mcO^n})\sqrt{\tilde{T}'}^\sfT = \sqrt{T'}^\sfT (\1_{\mcI^n} \otimes V^{\otimes n}_{\mcA^n\to \mcO^n}).\label{eq:intertwining}
\end{align}
Using the operators $T'$, $\tilde{T}'$ and $T'_{\hat{V}}$, we define a quantum state $\ket{\phi_\mathrm{est}} \in \mcI^n \otimes \mcA^n$ and a POVM $\{M_{\hat{V}} \dd \hat{V}\}_{\hat{V}} \subset \mcL(\mcI^n \otimes \mcO^n)$ by
\begin{align}
    \ket{\phi_\mathrm{est}}&\coloneqq \sqrt{\tilde{T}'}^\sfT \dket{\1}^{\otimes n}_{\mcI^n \mcA^n},\\
    M_{\hat{V}} &\coloneqq (T^{\prime -1/2} T'_{\hat{V}} T^{\prime -1/2})^\sfT.
\end{align}
The normalization condition of $\ket{\phi_\mathrm{est}}$ can be checked as follows:
\begin{align}
    \braket{\phi_\mathrm{est}}
    &= \bra{\phi_\mathrm{est}} (\1\otimes V^{\otimes n}_{\mcA^n\to \mcO^n})^\dagger (\1\otimes V^{\otimes n}_{\mcA^n\to \mcO^n})\ket{\phi_\mathrm{est}}\\
    &= T'\ast \dketbra{V}^{\otimes n}_{\mcI^n \mcO^n}\\
    &= \int_{\isometry{d}{D}} \dd \hat{V} \int_{\SU(d)} \dd U \int_{\SU(D)} \dd W (U^{\otimes n}_{\mcI^n} \otimes W^{\otimes n}_{\mcO^n})T_{W^\sfT \hat{V} U}(U^{\otimes n}_{\mcI^n} \otimes W^{\otimes n}_{\mcO^n})^\dagger \ast \dketbra{V}^{\otimes n}_{\mcI^n \mcO^n}\\
    &= \int_{\isometry{d}{D}} \dd \hat{V} \int_{\SU(d)} \dd U \int_{\SU(D)} \dd W T_{W^\sfT \hat{V} U} \ast \dketbra{W^\sfT VU}^{\otimes n}_{\mcI^n \mcO^n}\\
    &= \int_{\isometry{d}{D}} \dd \hat{V} \int_{\SU(d)} \dd U \int_{\SU(D)} \dd W p(W^\sfT \hat{V} U|W^\sfT V U)\\
    &= 1.
\end{align}
The positivity of $M_{\hat{V}}$ follows from Eq.~\eqref{eq:tester_positivity}, and
\begin{align}
    \int_{\isometry{d}{D}} \dd \hat{V} M_{\hat{V}}
    &= \int_{\isometry{d}{D}} \dd \hat{V} (T^{\prime -1/2} T'_{\hat{V}} T^{\prime -1/2})^\sfT\\
    &= (T^{\prime -1/2} T' T^{\prime -1/2})^\sfT\\
    &= \1_{\mcI^n\mcO^n}
\end{align}
holds.
Then, the probability distribution $p'(\hat{V}|V)$ can be reproduced by the parallel protocol as
\begin{align}
    p''(\hat{V}|V)
    &\coloneqq \Tr[M_{\hat{V}}  (\1_{\mcI^n} \otimes V^{\otimes n}_{\mcA^n\to \mcO^n}) \ketbra{\phi_{\mathrm{est}}} (\1_{\mcI^n} \otimes V^{\otimes n}_{\mcA^n\to \mcO^n})^\dagger]\\
    &= \Tr[M_{\hat{V}}  (\1_{\mcI^n} \otimes V^{\otimes n}_{\mcA^n\to \mcO^n})\sqrt{\tilde{T}'}^\sfT \dketbra{\1}^{\otimes n}_{\mcI^n \mcA^n}\sqrt{\tilde{T}'}^\sfT (\1_{\mcI^n} \otimes V^{\otimes n}_{\mcA^n\to \mcO^n})^\dagger]\\
    &=\Tr[\sqrt{T'}^\sfT M_{\hat{V}} \sqrt{T'}^\sfT (\1_{\mcI^n} \otimes V^{\otimes n}_{\mcA^n\to \mcO^n}) \dketbra{\1}^{\otimes n}_{\mcI^n \mcA^n} (\1_{\mcI^n} \otimes V^{\otimes n}_{\mcA^n\to \mcO^n})^\dagger]\label{eq:apply_intertwining_and_cyclic}\\
    &= \Tr[T_{\hat{V}}^{\prime \sfT} \dketbra{V}^{\otimes n}_{\mcI^n \mcO^n}]\\
    &= T'_{\hat{V}} \ast \dketbra{V}^{\otimes n}_{\mcI^n \mcO^n}\\
    &= p'(\hat{V}|V),
\end{align}
where we use \eqref{eq:intertwining} and the cyclic property of the trace in \eqref{eq:apply_intertwining_and_cyclic}.
As shown in Appendix.~C of Ref.~\cite{yoshida2024one}, the quantum state $\ket{\phi_\mathrm{est}}$ can be given as
\begin{align}
\label{eq:def_phi_est}
    \ket{\phi_\mathrm{est}} = \bigoplus_{\alpha\in\young{d}{n}} {v_\alpha \over \sqrt{d_\alpha^{(d)}}} \dket{S_\alpha},
\end{align}
where $v_\alpha \in \RR$ satisfies $\sum_{\alpha\in\young{d}{n}}v_\alpha^2 = 1$ and $v_\alpha\geq 0$, $d_\alpha^{(d)}$ is given in Eq.~\eqref{eq:d_alpha^d}, and $\dket{S_\alpha}$ is given by
\begin{align}
    \dket{S_\alpha} &\coloneqq \dket{\1_{\mcU_\alpha^{(d)}}}_{\mcU_{\alpha,1} \mcU_{\alpha,2}} \otimes \ket{\mathrm{arb}_\alpha}_{\mcS_{\alpha,1} \mcS_{\alpha,2}},\\
\label{eq:def_1_alpha}
    \dket{\1_{\mcU_\alpha^{(d)}}}_{\mcU_{\alpha,1} \mcU_{\alpha,2}} &\coloneqq \sum_{s=1}^{d_\alpha^{(d)}} \ket{\alpha, s}_{\mcU_{\alpha,1}} \otimes \ket{\alpha, s}_{\mcU_{\alpha,2}},
\end{align}
using a normalized vector $\ket{\mathrm{arb}_\alpha}$.
Defining the quantum state $\ket{\phi_V}$ by \begin{align}
    \ket{\phi_V}
    &\coloneqq (\1_{\mcI^n} \otimes V^{\otimes n}_{\mcA^n \to \mcO^n})\ket{\phi_\mathrm{est}}\\
    &= \bigoplus_{\alpha\in\young{d}{n}}{v_\alpha\over \sqrt{d_\alpha^{(d)}}} (\1_\alpha \otimes V_\alpha) \dket{\1_{\mcU_\alpha^{(d)}}} \otimes \ket{\mathrm{arb}_\alpha}_{\mcS_{\alpha,1} \mcS_{\alpha,2}},
\end{align}
the quantum state $\ket{\phi_V}$ can be compressed into the space $\mcH\coloneqq \bigoplus_{\alpha\in\young{d}{n}} \mcU_{\alpha}^{(d)} \otimes \mcU_{\alpha}^{(D)}$.
Formally, defining an embedding isometry operator $E:\mcH\to \mcI^n\otimes \mcO^n$ by
\begin{align}
    E(\ket{\alpha, s}\otimes \ket{\alpha, t}) = \ket{\alpha, s}_{\mcU_{\alpha,1}}\otimes \ket{\alpha, t}_{\mcU_\alpha,2} \otimes \ket{\mathrm{arb}_\alpha}_{\mcS_{\alpha,1} \mcS_{\alpha,2}},
\end{align}
$\ket{\phi_V}$ can be compressed into the quantum state $\ket{\phi'_V}\in\mcH$ defined by
\begin{align}
    \ket{\phi'_V}
    &\coloneqq E^\dagger \ket{\phi_V}\\
    &= \bigoplus_{\alpha\in\young{d}{n}}{v_\alpha\over \sqrt{d_\alpha^{(d)}}} (\1_\alpha \otimes V_\alpha) \dket{\1_{\mcU_\alpha^{(d)}}}.
    \label{eq:phi_V_prime}
\end{align}
The original quantum state $\ket{\phi_V}$ can be retrieved from the compressed state $\ket{\phi'_V}$ by $\ket{\phi_V} = E\ket{\phi'_V}$.
Therefore, instead of considering the POVM $\{M_{\hat{V}} \dd \hat{V}\}_{\hat{V}}$ on $\mcI^n\otimes \mcO^n$, we can consider a POVM $\{M'_{\hat{V}} \dd \hat{V}\}_{\hat{V}}$ on $\mcH$ defined by
\begin{align}
    M'_{\hat{V}}\coloneqq E^\dagger M_{\hat{V}} E
\end{align}
satisfying
\begin{align}
    \Tr(M_{\hat{V}} \ketbra{\phi_V}) = \Tr(M'_{\hat{V}} \ketbra{\phi'_V}).
\end{align}
By definition \eqref{eq:def_Tprime}, $T'_{\hat{V}}$ satisfies
\begin{align}
    T'_{W^\sfT \hat{V} U} = (U^{\otimes n}_{\mcI^n} \otimes W^{\otimes n}_{\mcO^n})^\dagger T'_{\hat{V}} (U^{\otimes n}_{\mcI^n} \otimes W^{\otimes n}_{\mcO^n}),
\end{align}
thus
\begin{align}
    M_{W^\sfT \hat{V} U} = (U^{\otimes n}_{\mcI^n} \otimes W^{\otimes n}_{\mcO^n})^\sfT M_{\hat{V}} (U^{\otimes n}_{\mcI^n} \otimes W^{\otimes n}_{\mcO^n})^*,
\end{align}
which reads
\begin{align}
    M'_{W^\sfT \hat{V} U} = \bigoplus_{\alpha\in\young{d}{n}}(U_\alpha \otimes W_\alpha)^\sfT M'_{\hat{V}} \bigoplus_{\alpha\in\young{d}{n}}(U_\alpha \otimes W_\alpha)^*.
\end{align}

Finally, we obtain the optimized POVM for the resource state $\ket{\phi_\mathrm{est}}$.
The fidelity is given by
\begin{align}
    F
    &= \int \dd V \int \dd \hat{V} \Tr[M'_{\hat{V}} \ketbra{\phi'_V}] F_\mathrm{ch}(\hat{V}, V)\\
    &= {1\over d^2}\int \dd V \int \dd \hat{V} \Tr[M'_{\hat{V}} \otimes \dketbra{\hat{V}} \cdot \ketbra{\phi'_V}\otimes \dketbra{V}]\\
    &= {1\over d^2}\int \dd V \int \dd W \Tr[M'_{WV_0} \otimes \dketbra{WV_0} \cdot \ketbra{\phi'_V}\otimes \dketbra{V}]\\
    &= {1\over d^2}\int \dd V \int \dd W \Tr[\bigoplus_{\alpha\in\young{d}{n}}(\1_\alpha\otimes W_\alpha)M'_{V_0} \bigoplus_{\alpha\in\young{d}{n}}(\1_\alpha\otimes W_\alpha)^\dagger \otimes (\1\otimes W)\dketbra{V_0}(\1\otimes W)^\dagger \cdot \ketbra{\phi'_V}\otimes \dketbra{V}]\\
    &= {1\over d^2}\int \dd V \int \dd W \Tr[M'_{V_0} \otimes \dketbra{V_0} \cdot \ketbra{\phi'_{W^\dagger V}}\otimes \dketbra{W^\dagger V}]\\
    &= {1\over d^2}\Tr[M'_{V_0} \otimes \dketbra{V_0} \cdot \Pi],
\end{align}
where $V_0\in\isometry{d}{D}$ is a fixed isometry operator, $\dd W$ is the Haar measure on $\SU(D)$ and $\Pi$ is defined by
\begin{align}
    \Pi
    &\coloneqq \int \dd V \ketbra{\phi'_V} \otimes \dketbra{V}\\
    &= \int \dd V \bigoplus_{\alpha, \beta\in\young{d}{n}}{v_\alpha v_\beta\over \sqrt{d_\alpha^{(d)}d_\beta^{(d)}}} \dketbra{V_\alpha}{V_\beta} \otimes \dketbra{V}\\
    &= \bigoplus_{\alpha, \beta\in\young{d}{n}, \mu\in\alpha+\square\cap \beta+\square} {v_\alpha v_\beta \over \sqrt{d_\alpha^{(d)} d_\beta^{(d)}}} {\1_{\mcU_\mu^{(d)}} \otimes \1_{\mcU_\mu^{(D)}} \otimes \ketbra{\alpha\alpha}{\beta\beta}_{\mathrm{multi}} \over d_\mu^{(D)}}.
\end{align}
We decompose $M'_{V_0}$ as
\begin{align}
    M'_{V_0} &= \sum_{i=1}^{r} \ketbra{\eta^i},\\
    \ket{\eta^i} &= \bigoplus_{\alpha\in\young{d}{n}} \sqrt{d_\alpha^{(D)}}\dket{\eta_\alpha^i}
\end{align}
using linear maps $\eta_\alpha^i: \mcU_\alpha^{(d)} \to \mcU_\alpha^{(D)}$.
Then, since $\{M'_{V} \dd V\}$ forms the POVM,
\begin{align}
    \1
    &= \int \dd V M'_{V}\\
    &= \int \dd W M'_{WV_0}\\
    &= \int \dd W \bigoplus_{\alpha, \beta\in\young{d}{n}} (\1_\alpha\otimes W_\alpha) M'_{V_0} (\1_\beta\otimes W_\beta)^\dagger\\
    &= \int \dd W \bigoplus_{\alpha, \beta\in\young{d}{n}} \sqrt{d_\alpha^{(D)}d_\beta^{(D)}} (\1_\alpha\otimes W_\alpha) \sum_i \dketbra{\eta_\alpha^i}{\eta_\beta^i} (\1_\beta\otimes W_\beta)^\dagger\\
    &= \bigoplus_{\alpha\in\young{d}{n}} \1_{\mcU_{\alpha}^{(D)}} \otimes \sum_i \eta_\alpha^{i \sfT} \eta_\alpha^{i *},
\end{align}
i.e.,
\begin{align}
    \sum_{i=1}^{r} \eta_\alpha^{i \dagger} \eta_\alpha^{i} = \1_{\mcU_\alpha^{(d)}} \quad \forall \alpha\in\young{d}{n}.
\end{align}
Defining the decomposition of $\eta_\alpha^i \otimes V_0$ by
\begin{align}
    \eta_\alpha^i \otimes V_0 &= \bigoplus_{\mu, \nu\in\alpha+\square} \eta_{\mu\to \nu}^{i, \alpha} \otimes \ketbra{\alpha}_{\mathrm{multi}},
\end{align}
using $\eta_{\mu\to \nu}^{i'}: \mcU_\mu^{(d)} \to \mcU_\nu^{(D)}$, we obtain
\begin{align}
    \bigoplus_{\mu\in\alpha+\square} \1_{\mcU_{\mu}^{(d)}} \otimes \ketbra{\alpha}_{\mathrm{multi}}
    &= \1_{\mcU_{\alpha}^{(d)}}\otimes \1_d\\
    &= \sum_{i=1}^{r} \eta_\alpha^{i \dagger} \eta_\alpha^{i} \otimes V_0^{\dagger} V_0\\
    &= \bigoplus_{\mu, \mu'\in\alpha+\square} \sum_{\nu\in\alpha+\square}\sum_{i=1}^{r} \eta_{\mu'\to \nu}^{i, \alpha \dagger} \eta_{\mu\to \nu}^{i, \alpha} \otimes \ketbra{\alpha}_{\mathrm{multi}},
\end{align}
i.e.,
\begin{align}
    \sum_{\nu\in\alpha+\square} \sum_{i=1}^{r} \eta_{\mu'\to \nu}^{i, \alpha \dagger} \eta_{\mu\to \nu}^{i, \alpha} = \delta_{\mu, \mu'} \1_{\mcU_\mu^{(d)}} \quad \forall \alpha\in\young{d}{n}, \mu, \mu'\in\alpha+\square.
\end{align}
In particular, we obtain
\begin{align}
    \sum_{i=1}^{r} \eta_{\mu\to \mu}^{i, \alpha \dagger} \eta_{\mu\to \mu}^{i, \alpha}
    &=\1_{\mcU_\mu^{(d)}} - \sum_{\nu\in\alpha+\square \setminus \{\mu\}} \sum_{i=1}^{r} \eta_{\mu\to \nu}^{i, \alpha \dagger} \eta_{\mu\to \nu}^{i, \alpha}\\
    &\leq \1_{\mcU_\mu^{(d)}} \quad \forall \alpha\in\young{d}{n}, \mu\in\alpha+\square.
    \label{eq:eta_inequality}
\end{align}
Therefore, the fidelity is further calculated by
\begin{align}
    F
    &= {1\over d^2}\Tr[M'_{V_0} \otimes \dketbra{V_0} \cdot \Pi]\\
    &= {1\over d^2} \Tr[\bigoplus_{\alpha,\beta\in\young{d}{n}} \sqrt{d_\alpha^{(D)} d_\beta^{(D)}} \sum_{i=1}^{r} \dketbra{\eta_\alpha^i \otimes V_0}{\eta_\beta^i \otimes V_0} \cdot \Pi]\\
    &= {1\over d^2} \Tr[\bigoplus_{\alpha,\beta\in\young{d}{n}} \sqrt{d_\alpha^{(D)} d_\beta^{(D)}} \bigoplus_{\mu, \nu\in\alpha+\square, \mu', \nu' \in \beta+\square}\sum_{i=1}^{r} \dketbra{\eta^{i,\alpha}_{\mu\to\nu}}{\eta^{i,\beta}_{\mu'\to\nu'}} \otimes \ketbra{\alpha\alpha}{\beta\beta}_{\mathrm{multi}} \cdot \Pi]\\
    &= {1\over d^2} \sum_{\alpha,\beta\in\young{d}{n}} \sum_{\mu\in\alpha+\square \cap\beta+\square} {v_\alpha v_\beta} \sqrt{d_\alpha^{(D)} d_\beta^{(D)} \over d_\alpha^{(d)} d_\beta^{(d)}}\sum_{i=1}^{r} {\Tr[\eta^{i,\alpha}_{\mu\to\mu} \eta^{i,\beta\dagger}_{\mu\to\mu}] \over d_\mu^{(D)}}\\
    \label{eq:cauchy_schwartz_used}
    &\leq {1\over d^2} \sum_{\alpha,\beta\in\young{d}{n}} \sum_{\mu\in\alpha+\square \cap\beta+\square} {v_\alpha v_\beta} \sqrt{d_\alpha^{(D)} d_\beta^{(D)} \over d_\alpha^{(d)} d_\beta^{(d)}}{\sqrt{\sum_{i=1}^{r} \Tr[\eta^{i,\alpha}_{\mu\to\mu} \eta^{i,\alpha\dagger}_{\mu\to\mu}]} \sqrt{\sum_{i=1}^{r} \Tr[\eta^{i,\beta}_{\mu\to\mu} \eta^{i,\beta\dagger}_{\mu\to\mu}]} \over d_\mu^{(D)}}\\
    \label{eq:eta_inequality_used}
    &\leq {1\over d^2} \sum_{\alpha,\beta\in\young{d}{n}} \sum_{\mu\in\alpha+\square \cap\beta+\square} {v_\alpha v_\beta} \sqrt{d_\alpha^{(D)} d_\beta^{(D)} \over d_\alpha^{(d)} d_\beta^{(d)}}{d_\mu^{(d)} \over d_\mu^{(D)}},
\end{align}
where we use the Cauchy-Schwartz inequality in Eq.~\eqref{eq:cauchy_schwartz_used} and Eq.~\eqref{eq:eta_inequality} in Eq.~\eqref{eq:eta_inequality_used}.
The equality holds when $r=1$ and $\eta^i_{\alpha} = V_{0,\alpha}$ (i.e., $M'_{V_0} = \bigoplus_{\alpha\in\young{d}{n}} d_\alpha^{(D)}\dketbra{V_{0,\alpha}}$) as shown below.
Using Lem.~\ref{lem:isometry_irrep_product} in Appendix~\ref{appendix_sec:shur_weyl_isometry}, we obtain
\begin{align}
    \eta^i_{\alpha} \otimes V_0
    &= V_{0,\alpha} \otimes V_0\\
    &= \bigoplus_{\mu\in\alpha+\square} V_{0, \mu} \otimes \ketbra{\alpha}_\mathrm{multi},
\end{align}
i.e.,
\begin{align}
    \eta^{i,\alpha}_{\mu\to\nu} = \delta_{\mu\nu} V_{0,\mu}
\end{align}
holds.
Therefore, the equality holds in the Cauchy-Schwartz inequality used in Eq.~\eqref{eq:cauchy_schwartz_used} and the inequality shown in Eq.~\eqref{eq:eta_inequality}.
Thus, the optimized POVM is given by
\begin{align}
\label{eq:def_M_hat_V_prime}
    M'_{\hat{V}}
    &= \bigoplus_{\alpha\in\young{d}{n}} d_\alpha^{(D)} \dketbra{\hat{V}},\\
    M_{\hat{V}}
    &= E M'_{\hat{V}} E^\dagger\\
    &= \bigoplus_{\alpha\in\young{d}{n}} d_\alpha^{(D)} (\1_d^{\otimes n} \otimes \hat{V}^{\otimes n})\dketbra{S_\alpha}(\1_d^{\otimes n} \otimes \hat{V}^{\otimes n})^\dagger.
    \label{eq:def_M_hat_V}
\end{align}
The fidelity for the optimized POVM is given by
\begin{align}
    F
    &= \vec{v}^\sfT M_{\mathrm{est}}(n,d,D) \vec{v},
\end{align}
where $M_\mathrm{est}(n,d,D)$ is a $\abs{\young{d}{n}} \times \abs{\young{d}{n}}$ real symmetric matrix defined by
\begin{align}
    (M_\mathrm{est}(n,d,D))_{\alpha\beta}
    &\coloneqq {1\over d^2}\sum_{\mu\in\alpha+\square \cap\beta+\square} \sqrt{d_\alpha^{(D)} d_\beta^{(D)} \over d_\alpha^{(d)} d_\beta^{(d)}}{d_\mu^{(d)} \over d_\mu^{(D)}}\\
    &= {1\over d^2}\sum_{i,j=1}^{d} \delta_{\alpha+e_i, \beta+e_j} \sqrt{d_\alpha^{(D)} d_{\alpha+e_i}^{(d)} \over d_\alpha^{(d)} d_{\alpha+e_i}^{(D)}}\sqrt{d_\beta^{(D)} d_{\beta+e_j}^{(d)} \over d_\beta^{(d)} d_{\beta+e_j}^{(D)}}\\
    &= {1\over d^2}\sum_{i,j=1}^{d} \delta_{\alpha+e_i, \beta+e_j} f(\alpha_i-i)f(\beta_j-j).
\end{align}
Thus, the optimal fidelity is given by
\begin{align}
\label{eq:maximization_M_est}
    \max_{\vec{v}: \abs{\vec{v}}=1, v_\alpha\geq 0} \vec{v}^\sfT M_{\mathrm{est}}(n,d,D) \vec{v}.
\end{align}
Since $(M_\mathrm{est}(n,d,D))_{\alpha\beta}\geq 0$ holds for all $\alpha, \beta\in \young{d}{n}$, due to the Perron-Frobenius theorem~\cite{horn2012matrix}, Eq.~\eqref{eq:maximization_M_est} is given by the maximal eigenvalue of $M_{\mathrm{est}}(n,d,D)$.

\subsection{Proof of Lem.~\ref{lem:isometry_estimation_fidelity} (Construction of the asymptotically optimal isometry estimation protocol)}
\label{appendix_sec:proof_isometry_estimation_fidelity}
We construct a parallel covariant isometry estimation protocol achieving the fidelity shown in Lem.~\ref{lem:isometry_estimation_fidelity}, which is used in the proof of Thm.~1 to construct the asymptotically optimal isometry estimation protocol.
We use a similar strategy used in \cite{yang2020optimal} to construct the optimal universal programming of unitary channels.

We construct the vector $(v_\alpha)_{\alpha\in\young{d}{n}}$ in Lem.~\ref{lem:isometry_estimation_matrix} to show the protocol.
To this end, we define a function $g:\NN \to \NN$ satisfying $g(n)\leq {2\over 3(d-1)}({n\over d}+d-2)$ and set a parameter $N=g(n)$.
Using this parameter $N$, we define a set of Young diagrams $\Syoung$ by
\begin{align}
\label{eq:def_Syoung}
    \Syoung\coloneqq \left\{\alpha = (\alpha_1, \ldots, \alpha_d)\in \young{d}{n}\;\middle|\;\exists \tilde{\alpha}\in [N-1]^{d-1} \; \mathrm{s.t.} \;\alpha_i = A_i + \tilde{\alpha}_i \quad \forall i\in\{1,\ldots,d-1\}\right\},
\end{align}
where $[N-1]$ denotes $[N-1] = \{0, \ldots, N-1\}$, $A_i$ is defined by
\begin{align}
\label{eq:def_Ai}
    A_i &= q+(d-i)N+\delta_{i\leq r},
\end{align}
using $q\in\NN$ and $r\in\{0,\ldots,d-1\}$ satisfying
\begin{align}
\label{eq:def_qr}
    n - {d(d-1)\over 2}N = dq+r.
\end{align}
Since
\begin{align}
    A_i &\geq A_{i+1}+N \quad \forall i\in \{1, \ldots, d-2\},\\
    A_{d-1}&\geq \max_{\tilde{\alpha}} \alpha_d +N = n-\sum_{i=1}^{d-1} A_i +N,\\
    \min_{\tilde{\alpha}}\alpha_d &= n-\sum_{i=1}^{d-1}A_i - (d-1)(N-1) \geq 0
\end{align}
hold, any element $\alpha\in \Syoung$ is uniquely specified by $\tilde{\alpha}\in [N-1]^{d-1}$ and $\alpha\in \Syoung$ satisfies
\begin{align}
    \alpha_i > \alpha_{i+1} \quad \forall i\in\{1,\ldots,d-1\}.
\end{align}
We notice that for $\alpha, \beta\in\Syoung$,
\begin{align}
    (M_\mathrm{est}(n,d,D))_{\alpha \beta} &= {1\over d^2}\sum_{i,j=1}^{d} \delta_{\alpha+e_i, \beta+e_j} f(\alpha_i-i)f(\beta_j-j)\\
    &=
    \begin{cases}
        {1\over d^2} f(\alpha_i-i)f(\beta_j-j) & (\exists i, j \;\mathrm{s.t.}\;\tilde{\alpha}-\tilde{\beta} = e_i - e_j)\\
        {1\over d^2} f(\alpha_d - d) f(\beta_i-i) & (\exists i \;\mathrm{s.t.}\; \tilde{\alpha}-\tilde{\beta} = e_i)\\
        {1\over d^2} f(\alpha_i - i) f(\beta_d-d) & (\exists i \;\mathrm{s.t.}\; \tilde{\alpha}-\tilde{\beta} = -e_i)\\
        {1\over d^2} \sum_{i=1}^{d} [f(\alpha_i-i)]^2 & (\tilde{\alpha} = \tilde{\beta})\\
         0 & (\mathrm{otherwise})
    \end{cases}
\end{align}
holds.
Since the function $f$ defined in Lem.~\ref{lem:isometry_estimation_matrix} is monotonically increasing and $q-d\leq \alpha_i-i \leq q+(d-1) N$ holds for all $\alpha\in\mathbb{S}_\mathrm{young}$ and $i\in\{1, \ldots, d\}$,
\begin{align}
    f(q-d) \leq f(\alpha_i-i) \leq f(q+(d-1)N) \quad \forall \alpha\in\Syoung, i \in \{1, \ldots, d\}
\end{align}
holds, and we obtain
\begin{align}
    \begin{cases}
        {1\over d^2} [f(q-d)]^2 \leq (M_\mathrm{est}(n,d,D))_{\alpha \beta} \leq {1\over d^2} [f(q+(d-1)N)]^2 & (\exists i\neq j \;\mathrm{s.t.}\; \tilde{\alpha}-\tilde{\beta} = e_i - e_j)\\
        {1\over d^2} [f(q-d)]^2 \leq (M_\mathrm{est}(n,d,D))_{\alpha \beta} \leq {1\over d^2} [f(q+(d-1)N)]^2 & (\exists i \;\mathrm{s.t.}\; \tilde{\alpha}-\tilde{\beta} = \pm e_i)\\
        {1\over d} [f(q-d)]^2 \leq (M_\mathrm{est}(n,d,D))_{\alpha \beta} \leq {1\over d} [f(q+(d-1)N)]^2 & (\tilde{\alpha} = \tilde{\beta})\\
        (M_\mathrm{est}(n,d,D))_{\alpha \beta} = 0 & (\mathrm{otherwise})
    \end{cases}.
\end{align}
Defining the probability distribution $(g_k)_{k=0}^{N-1}$ by
\begin{align}
    g_k \coloneqq {2\over N}\sin^2\left(\pi(2k+1)\over 2N\right),
\end{align}
we set the vector $(v_\alpha)_{\alpha\in\young{d}{n}}$ to be
\begin{align}
    v_\alpha
    &=
    \begin{cases}
        G_{\tilde{\alpha}} & (\alpha\in\Syoung)\\
        0 & (\mathrm{otherwise})
    \end{cases},
    \label{eq:v_alpha_definition}\\
    G_{\tilde{\alpha}}&\coloneqq
    \begin{cases}
        \prod_{i=1}^{d-1} \sqrt{g_{\tilde{\alpha}_i}} & (\tilde{\alpha} \in [N-1]^{d-1})\\
        0 & (\mathrm{otherwise})
    \end{cases},
\end{align}
which satisfies the normalization condition
\begin{align}
    \sum_{\alpha\in\young{d}{n}}v_\alpha^2 &= \sum_{\tilde{\alpha}\in [N-1]^{d-1}} \prod_{i=1}^{d-1}g_{\tilde{\alpha}_i}\\
    &= \left(\sum_{k=0}^{N-1} g_k\right)^{d-1}\\
    &= 1.
\end{align}
Defining $\epsilon_g$ by
\begin{align}
    \epsilon_g
    &\coloneqq 1-\sum_{k=0}^{N-2} \sqrt{g_k g_{k+1}}\\
    &= 1-{2\over N}\sum_{k=0}^{N-2} \sin(\pi(2k+1)\over 2N) \sin(\pi(2k+3)\over 2N)\\
    &= 1-{1\over N}\sum_{k=0}^{N-2} \left[\cos(\pi\over N) - \cos(\pi(2k+2) \over N)\right]\\
    &= {N-1\over N}\left[1-\cos(\pi\over N)\right],
\end{align}
we have
\begin{align}
    {\pi^2\over 2N^2} - O(N^{-3}) \leq \epsilon_g \leq {\pi^2\over 2 N^2}
\end{align}
and the fidelity of isometry estimation corresponding to the vector $(v_\alpha)_{\alpha\in\young{d}{n}}$ defined in \eqref{eq:v_alpha_definition} is evaluated as
\begin{align}
    F_\mathrm{est}
    &= \sum_{\alpha, \beta\in\Syoung} v_\alpha (M_\mathrm{est}(n,d,D))_{\alpha\beta} v_\beta\\
    &\geq {1\over d^2}[f(q-d)]^2 \left[\sum_{\tilde{\alpha}\in [N-1]^{d-1}} \sum_{i\neq j} G_{\tilde{\alpha}}G_{\tilde{\alpha}-e_i+e_j} + \sum_{\tilde{\alpha}\in [N-1]^{d-1}} \sum_{i=1}^{d-1} \sum_{\pm} G_{\tilde{\alpha}}G_{\tilde{\alpha}\pm e_i} + d \sum_{\tilde{\alpha}\in [N-1]^{d-1}} G_{\tilde{\alpha}}^2\right]\\
    &= {1\over d^2}[f(q-d)]^2 \left[\sum_{i\neq j} \sum_{\tilde{\alpha}_i=1}^{N-1} \sqrt{g_{\tilde{\alpha}_i} g_{\tilde{\alpha}_i-1}} \sum_{\tilde{\alpha}_j=0}^{N-2} \sqrt{g_{\tilde{\alpha}_j} g_{\tilde{\alpha}_j+1}}
    +\sum_{i=1}^{d-1} \sum_{\tilde{\alpha}_i=0}^{N-2}\left(\sqrt{g_{\tilde{\alpha}_i}g_{\tilde{\alpha}_i+1}} + \sum_{\tilde{\alpha}_i=1}^{N-1}\sqrt{g_{\tilde{\alpha}_i}g_{\tilde{\alpha}_i-1}}\right) + d\right]\\
    &= {1\over d^2}[f(q-d)]^2 \left[(d-1)(d-2)(1-\epsilon_g)^2 + 2(d-1)(1-\epsilon_g) + d\right]\\
    & = [f(q-d)]^2 \left[1-{2(d-1)^2 \over d^2} \epsilon_g + {(d-1)(d-2) \over d^2} \epsilon_g^2\right]\\
    &\geq \left[1-{D-d \over q+1+D-d}\right]\left[1-{2(d-1)^2\over d^2}\epsilon_g\right]\\
    &\geq 1-{2(d-1)^2\over d^2}\epsilon_g - {D-d \over q+1+D-d}\\
    \label{eq:derivation_isometry_estimation_fidelity}
    &\geq 1-{\pi^2(d-1)^2\over d^2 N^2}-{D-d\over {n\over d}-{d-1\over 2}N+D-d}\\
    &= 1-{\pi^2(d-1)^2\over d^2 g(n)^2}-{D-d\over {n\over d}-{d-1\over 2}g(n)+D-d}.
\end{align}
Similarly, we show a converse bound given by
\begin{align}
    F_\mathrm{est}
    &\leq [f(q+(d-1)N)]^2 \left[1-{2(d-1)^2 \over d^2} \epsilon_g + {(d-1)(d-2) \over d^2} \epsilon_g^2\right]\\
    &\leq \left[1-{D-d\over q+(d-1)N+D+1}\right] \left[1-{\pi^2 (d-1)^2\over d^2 N^2} +O(N^{-3})\right]\\
    &\leq 1-{\pi^2 (d-1)^2\over d^2 N^2} - {d(D-d) \over n} + O(Nn^{-2}, N^{-3})\\
    & = 1-{\pi^2 (d-1)^2\over d^2 g(n)^2} - {d(D-d) \over n} + O(g(n) n^{-2}, g(n)^{-3}).
\end{align}

\section{Proof of Cor.~2 (Asymptotic optimality of the PBT-based strategy for dSAR of isometry channels and CPTP maps)}

\begin{proof}[Proof of the optimality of Eq.~\eqref{eq:optimal_SAR} in the main text]
We show the optimality of Eq.~\eqref{eq:optimal_SAR} in the main text.
The optimal retrieval error of $\mbS_\mathrm{Isometry}^{(d,D)}$ is lower bounded by that of $\mbS_\mathrm{Unitary}^{(d)}$ since $\mbS_\mathrm{Unitary}^{(d)}$ can be regarded as a subset of $\mbS_\mathrm{Isometry}^{(d,D)}$ with a natural embedding.
Since the optimal storage and retrieval of $\mbS_\mathrm{Unitary}^{(d)}$ is achieved by the estimation-based strategy~\cite{bisio2010optimal}, the optimal retrieval error is lower bounded by
\begin{align}
    \epsilon \geq 1-F_\mathrm{est}(n,d) = {\Theta(d^4)\over n^2} + O(n^{-3}).
\end{align}
\end{proof}

We then show the program cost of the dSAR of isometry channels via the dPBT protocol shown in Eq.~(10) in the main text.
To construct the dPBT protocol, we utilize the equivalence between the unitary estimation and the dPBT~\cite{yoshida2024one}.
Reference~\cite{yoshida2024one} constructs the dPBT protocol from the parallel covariant unitary estimation protocol, which corresponds to the parallel covariant isometry estimation protocol shown in Sec.~\ref{appendix_sec:proof_isometry_estimation_matrix} for the case of $D=d$.
Suppose the resource state~\eqref{eq:def_phi_est} combined with the POVM~\eqref{eq:def_M_hat_V_prime} for $D=d$ implements unitary estimation with the estimation fidelity $F_\mathrm{est} = 1-\epsilon$.
Then, Ref.~\cite{yoshida2024one} shows a dPBT protocol achieving the teleportation error $\delta\leq \epsilon$ with the resource state
\begin{align}
    \ket{\phi_\mathrm{PBT}}\coloneqq \bigoplus_{\mu\in\young{d}{n+1}} {w_\mu \over \sqrt{d_\mu^{(d)} m_\mu}} \dket{\1_{\mcU_{\mu}^{(d)}}} \otimes \dket{\1_{\mcS_\mu}},
\end{align}
where $d_\mu^{(d)}$ and $m_\mu$ are given in Eqs.~\eqref{eq:d_alpha^d} and \eqref{eq:m_mu_STab}, $w_\mu$ is defined by
\begin{align}
    w_\mu = {\sum_{\alpha\in\mu-_d \square} v_\alpha \over \sqrt{\sum_{\mu\in\young{d}{n+1}} \left(\sum_{\alpha\in\mu-_d \square} v_\alpha\right)^2}},
\end{align}
$\dket{\1_{\mcU_{\mu}^{(d)}}}$ is defined in Eq.~\eqref{eq:def_1_alpha}, and $\dket{\1_{\mcS_\mu}}$ is defined by
\begin{align}
    \dket{\1_{\mcS_\mu}} = \sum_{s_\mu\in\mathrm{STab}(\mu)} \ket{s_\mu} \otimes \ket{s_\mu}
\end{align}
using the Young-Yamanouchi basis $\{\ket{s_\mu}\}_{s_\mu\in\mathrm{STab}(\mu)}$ of $\mcS_\mu$ defined in Eq.~\eqref{eq:yy_basis}.
Defining $\mbS_\mathrm{Young}$ by
\begin{align}
\label{eq:def_Syoung_valpha}
    \mbS_\mathrm{Young}\coloneqq \{\alpha\mid v_\alpha\neq 0\},
\end{align}
$w_\mu$ can be nonzero for $\mu\in \Syoung+_d\square$,
where $\mbS_\mathrm{Young}+_d\square$ is defined by
\begin{align}
    \mbS_\mathrm{Young}+_d\square\coloneqq \bigcup_{\alpha\in\Syoung}(\alpha+_d \square).
\end{align}
The resource state $\ket{\phi_\mathrm{PBT}}$ can be used for storage and retrieval of isometry channel $V$, where the program state is given by
\begin{align}
    (\1_d^{\otimes n+1} \otimes V^{\otimes n+1})\ket{\phi_\mathrm{PBT}} = \bigoplus_{\mu\in\young{d}{n+1}} {w_\mu \over \sqrt{d_\mu^{(d)} m_\mu}} \dket{V_\mu} \otimes \dket{\1_{\mcS_\mu}},
\end{align}
where $\dket{V_\mu}\in \mcU_\mu^{(d)}\otimes \mcU_\mu^{(D)}$ is defined by
\begin{align}
    \dket{V_\mu}\coloneqq (\1_{\mcU_\mu^{(d)}} \otimes V_\mu) \dket{\1_{\mcU_\mu^{(d)}}}.
\end{align}
This program state can be stored in a Hilbert space given by
\begin{align}
    \mcP = \bigoplus_{w_\mu\neq 0} \mcU_\mu^{(d)}\otimes \mcU_\mu^{(D)} \subset \bigoplus_{\mu\in \mbS_\mathrm{Young}+_d\square} \mcU_\mu^{(d)}\otimes \mcU_\mu^{(D)}.
\end{align}
Therefore, the program cost is given by
\begin{align}
    c'_P\leq \log\left[\sum_{\mu\in \mbS_\mathrm{Young}+_d\square} d_\mu^{(d)}d_\mu^{(D)}\right].
\end{align}
This shows the following Lemma:

\begin{lemma}
\label{lem:unitary_estimation_to_isometry_learning}
    Suppose there exists a parallel covariant unitary estimation protocol with the resource state~\eqref{eq:def_phi_est} achieving the estimation fidelity $F_\mathrm{est}=1-\epsilon$.
    Then, there exists a universal programmable processor of $\mbS_\mathrm{Isometry}^{(d,D)}$ achieving the retrieval error $\epsilon$ with the program cost given by
    \begin{align}
    \label{eq:program_cost_PBT_based}
        c'_P\leq \log\left[\sum_{\mu\in \mbS_\mathrm{Young}+_d\square} d_\mu^{(d)}d_\mu^{(D)}\right],
    \end{align}
    where $\mbS_\mathrm{Young}$ is defined in Eq.~\eqref{eq:def_Syoung_valpha}.
\end{lemma}

We apply Lem.~\ref{lem:unitary_estimation_to_isometry_learning} for the unitary estimation protocol shown in Ref.~\cite{yang2020optimal} to prove Eq.~(10) in the main text.

\begin{proof}[Proof of Eq.~(10) in the main text]
    Reference~\cite{yang2020optimal} constructs the unitary estimation protocol similar to the one shown in Appendix~\ref{appendix_sec:proof_isometry_estimation_fidelity}.
    By setting $N=\left\lfloor {2\over 3(d-1)}({n\over d}+d-2)\right\rfloor = \Theta(n)$, we can construct the unitary estimation protocol with the estimation fidelity
    \begin{align}
        F_\mathrm{est}
        &\geq 1-{2(d-1)^2\pi^2\over d^2 N^2}\\
        &= 1- {\Theta(d^4) \over n^2}+O(n^{-3}).
    \end{align}
    This evaluation is shown by setting $D=d$ in Eq.~\eqref{eq:derivation_isometry_estimation_fidelity}.
    We evaluate the corresponding program cost~\eqref{eq:program_cost_PBT_based} for the universal programming of the isometry channel as follows:
    \begin{align}
        c'_P
        &\leq \log(\abs{\Syoung+_d\square}) + \max_{\mu\in\Syoung+_d\square} \log\left[d_\mu^{(d)}d_\mu^{(D)}\right]\\
        &\leq \log(\abs{\Syoung+_d\square}) + \max_{\alpha\in\Syoung}\left\{ \log\left[\sum_{\mu\in\alpha+_d\square}d_\mu^{(d)}\right]+\log\left[\sum_{\mu\in\alpha+_D\square}d_\mu^{(D)}\right]\right\}.
    \end{align}
    The cardinality of $\Syoung+_d\square$ is given by
    \begin{align}
        \abs{\Syoung+_d\square}
        &\leq \sum_{\alpha\in\Syoung} \abs{\alpha+_d\square}\\
        &\leq d\abs{\Syoung},
    \end{align}
    and $\sum_{\mu\in\alpha+_d\square}d_\mu^{(d)}$ and $\sum_{\mu\in\alpha+_D\square}d_\mu^{(D)}$ are given by [see Eq.~\eqref{eq:decomp_tensor_product_U_alpha}]
    \begin{align}
        \sum_{\mu\in\alpha+_d\square}d_\mu^{(d)} = d d_\alpha^{(d)}, \quad \sum_{\mu\in\alpha+_D\square}d_\mu^{(D)} = Dd_\alpha^{(D)}.
    \end{align}
    Thus, the program cost $c'_P$ is further evaluated as
    \begin{align}
        c'_P \leq (d-1)\log N + 2\log d + \log D + \max_{\alpha\in\Syoung}\log d_\alpha^{(d)}+ \max_{\alpha\in\Syoung}\log d_\alpha^{(D)}.
    \end{align}
    The values $\max_{\alpha\in\Syoung} \log d_\alpha^{(d)}$ and $\max_{\alpha\in\Syoung}\log d_\alpha^{(D)}$ are evaluated as follows [see Eqs.~\eqref{eq:def_Syoung}--\eqref{eq:def_qr}, \eqref{eq:d_alpha^d} and \eqref{eq:d_alpha^D}]:
    \begin{align}
        \log d_\alpha^{(d)}&= \log\left[{\prod_{1\leq i<j\leq d} (\alpha_i-\alpha_j-i+j) \over \prod_{k=1}^{d-1}k!}\right]\\
        &\leq \sum_{1\leq i<j\leq d} \log(\alpha_i-\alpha_j-i+j)\\
        &= \sum_{1\leq i<j< d} \log(\alpha_i-\alpha_j-i+j) + \sum_{1\leq i< d} \log(\alpha_i-\alpha_d-i+d)\\
        &= \sum_{1\leq i<j< d} \log(A_i-A_j+\tilde{\alpha}_i-\tilde{\alpha}_j-i+j) + \sum_{1\leq i< d} \log(A_i+\tilde{\alpha}_i-n+\sum_{i=1}^{d-1}(A_i+\tilde{\alpha}_i)-i+d)\\
        &\leq \sum_{1\leq i<j< d} \log((2j-2i+1)N-1) + \sum_{1\leq i< d} \log((2d-i)N+1-i)\\
        &\leq {d(d-1)\over 2}\log n + O(1).
    \end{align}
    \begin{align}
        \log d_\alpha^{(D)}
        &= \log d_\alpha^{(d)} + \sum_{1\leq i\leq d, d+1\leq j\leq D}\log(\alpha_i-\alpha_j-i+j) + \sum_{d+1\leq i<j\leq D}\log(\alpha_i-\alpha_j-i+j)-\sum_{k=d}^{D-1}\log(k!)\\
        &= \log d_\alpha^{(d)} + \sum_{1\leq i\leq d, d+1\leq j\leq D}\log(\alpha_i-i+j) + \sum_{d+1\leq i<j\leq D}\log(-i+j)-\sum_{k=d}^{D-1}\log(k!)\\
        &\leq \log d_\alpha^{(d)} + d(D-d) \max_{\alpha\in\Syoung}\log(\alpha_1-1+D)+O(1)\\
        &\leq \log d_\alpha^{(d)} + d(D-d) \log(q+(d-1)N+1)+O(1)\\
        &\leq \log d_\alpha^{(d)} + d(D-d) \log n+O(1).
    \end{align}
    The corresponding program cost is given by
    \begin{align}
        c_P
        & = \log \left[\sum_{\alpha\in\Syoung} (d_\alpha^{(d)})^2\right]\\
        &\leq \log \left[\abs{\Syoung}\right]+2\max_{\alpha\in\Syoung} \log\left[d_\alpha^{(d)}\right],
    \end{align}
    where $\abs{\Syoung} = N^{d-1}$ is the cardinality of the set $\Syoung$ and 
    $\max_{\alpha\in\Syoung} \log\left[d_\alpha^{(d)}\right]$ is evaluated as follows:
    \begin{align}
        \log\left[d_\alpha^{(d)}\right]&= \log\left[{\prod_{1\leq i<j\leq d} (\alpha_i-\alpha_j-i+j) \over \prod_{k=1}^{d-1}k!}\right]\\
        &\leq \sum_{1\leq i<j\leq d} \log(\alpha_i-\alpha_j-i+j)\\
        &= \sum_{1\leq i<j< d} \log(\alpha_i-\alpha_j-i+j) + \sum_{1\leq i< d} \log(\alpha_i-\alpha_d-i+d)\\
        &= \sum_{1\leq i<j< d} \log(A_i-A_j+\tilde{\alpha}_i-\tilde{\alpha}_j-i+j) + \sum_{1\leq i< d} \log(A_i+\tilde{\alpha}_i-n+\sum_{i=1}^{d-1}(A_i+\tilde{\alpha}_i)-i+d)\\
        &\leq \sum_{1\leq i<j< d} \log((2j-2i+1)N-1) + \sum_{1\leq i< d} \log((2d-i)N+1-i)\\
        &\leq {d(d-1)\over 2}\log n + O(1).
    \end{align}
    Thus, the program cost is given by
    \begin{align}
        c'_P
        &\leq (d-1)\log n + d(d-1)\log n + d(D-d)\log n + O(1)\\
        &= (Dd-1) \log n + O(1).
    \end{align}
    Therefore, to achieve the retrieval error $\epsilon$, we can put $n = \sqrt{\Theta(d^4)\over \epsilon}$ to obtain
    \begin{align}
        c'_P &\leq {Dd-1\over 2}\log\Theta(\epsilon^{-1}).
    \end{align}
\end{proof}

\section{Program cost of the estimation-based universal programming of isometry channels}
\label{appendix_sec:isometry_programming_estimation}

This section shows the program cost of the universal programming of isometry channels via the isometry estimation, as shown in the following corollary.
\begin{corollary}
\label{cor:estimation_based_program_cost}
    The program cost of the universal programming of isometry channels with the estimation-based strategy is given by
    \begin{align}
    \label{eq:program_cost_estimation_based}
        c_P \leq {2Dd-d^2-1\over 2}\log\Theta(\epsilon^{-1}).
    \end{align}
\end{corollary}
\begin{proof}
We show the following Lemma:
\begin{lemma}
\label{lem:program_cost_estimation_based_isometry}
    The universal programming via isometry estimation shown in Lem.~\ref{lem:isometry_estimation_fidelity} has the program cost
    \begin{align}
        c_P = h(t) \log\Theta(\epsilon^{-1})
    \end{align}
    for any $0\leq t\leq 1$ satisfying $g(n)=\Theta(n^t)$, where $h(t)$ is defined by
    \begin{align}
    \label{eq:h_alpha_def}
        h(t)\coloneqq
        \begin{cases}
            {t(d^2-1) + d(D-d)\over 2t} & (0\leq t\leq {1\over 2})\\
            t(d^2-1)+d(D-d) & ({1\over 2}\leq t\leq 1)
        \end{cases}.
    \end{align}
\end{lemma}
Lemma~\ref{lem:program_cost_estimation_based_isometry} leads to Cor.~\ref{cor:estimation_based_program_cost} since the minimum value of $h(t)$ in Eq.~\eqref{eq:h_alpha_def} is given by
    \begin{align}
        \min_{0\leq r\leq 1} h(t)
        &=h\left({1\over 2}\right)= {2Dd-d^2-1\over 2}.
    \end{align}
\end{proof}
From Cor.~\ref{cor:estimation_based_program_cost} and the fact that the isometry channels have $2Dd-d^2-1$ parameters, we conjecture the following statement:
\begin{conjecture}
    The optimal program cost of estimation-based universal programming of an isometry channel with $\nu$ parameters is given by
    \begin{align}
        c_P^{(\mathrm{est})} = {\nu \over 2} \log\Theta(\epsilon^{-1}).
    \end{align}
\end{conjecture}
\begin{proof}[Proof of Lem.~\ref{lem:program_cost_estimation_based_isometry}]
    The program cost of the protocol shown in the proof of Lem.~\ref{lem:isometry_estimation_fidelity} is given by
    \begin{align}
        c_P = \log d_P = \log \sum_{\alpha\in\Syoung} d_\alpha^{(d)}d_\alpha^{(D)}.
    \end{align}
    The irreducible representation dimension $d_\alpha^{(d)}$ is given by Eq.~\eqref{eq:d_alpha^d}.
    Since for all $i<j$, $\alpha\in\Syoung$ satisfies [see Eqs.~\eqref{eq:def_Syoung}--\eqref{eq:def_qr}]
    \begin{align}
        \alpha_i-\alpha_j
        &\leq \alpha_1-\alpha_d\\
        &= A_1+\tilde{\alpha}_1 - \left[n-\sum_{i=1}^{d-1} (A_i+\tilde{\alpha}_i)\right]\\
        &= A_1-q+\tilde{\alpha}_1 + \sum_{i=1}^{d-1}\tilde{\alpha}_i\\
        &\leq (d-1)g(n)+1+d[g(n)-1]\\
        &\leq \Theta(g(n)),
    \end{align}
    we obtain
    \begin{align}
        d_\alpha^{(d)} \leq \Theta(g(n)^{d(d-1)\over 2}).
    \end{align}
    Similarly, since
    \begin{align}
        d_\alpha^{(D)} = {\prod_{1\leq i<j\leq D}(\alpha_i-\alpha_j-i+j) \over \prod_{k=1}^{D-1} k!}
    \end{align}
    holds, and $\alpha_{d+1} = \cdots = \alpha_D = 0$ holds for $\alpha\in\Syoung$, we obtain
    \begin{align}
        d_\alpha^{(D)}
        &= {\prod_{1\leq i<j\leq d}(\alpha_i-\alpha_j-i+j) \prod_{1\leq i\leq d, d+1\leq j\leq D}(\alpha_i-i+j) \prod_{d+1\leq i<j\leq D}(-i+j) \over \prod_{k=1}^{d-1} k! \prod_{k=d}^{D-1}k!}\\
        &= {\prod_{1\leq i<j\leq d}(\alpha_i-\alpha_j-i+j) \over \prod_{k=1}^{d-1} k!} {\prod_{d+1\leq i<j\leq D} (-i+j) \over \prod_{k=d}^{D-1}k!} \prod_{1\leq i\leq d, d+1\leq j\leq D} (\alpha_i-i+j)\\
        \label{eq:d_alpha^D}
        &= d_\alpha^{(d)} \prod_{k=1}^{D-d-1} {k! \over (k+d-1)!} \prod_{1\leq i\leq d, d+1\leq j\leq D} (\alpha_i-i+j)\\
        &\leq d_\alpha^{(d)} \prod_{i=1}^{d}(\alpha_i-i+D)^{D-d}.
    \end{align}
    Since $\alpha\in\Syoung$ satisfies [see Eqs.~\eqref{eq:def_Syoung}--\eqref{eq:def_qr}]
    \begin{align}
        \alpha_i \leq \alpha_1\leq q+(d-1)g(n)+g(n) \leq {1\over d}\left[n - {d(d-1)\over 2}g(n)\right] + dg(n) \leq \Theta(n) \quad \forall i,
    \end{align}
    we obtain
    \begin{align}
        d_\alpha^{(D)}\leq \Theta(g(n)^{d(d-1) \over 2} n^{d(D-d)}).
    \end{align}
    Since the cardinality of $\Syoung$ is given by $\abs{\Syoung} = g(n)^{d-1}$, we obtain an upper bound on $c_P = \log \sum_{\alpha\in\Syoung} d_\alpha^{(d)} d_\alpha^{(D)}$ by
    \begin{align}
        c_P\leq (d^2-1)\log \Theta(g(n)) + d(D-d)\log \Theta(n).
    \end{align}
    By putting $g(n)=\Theta(n^t)$ for $0\leq t\leq 1$ in Eqs.~\eqref{eq:estimation_fidelity} and \eqref{eq:estimation_fidelity_converse},
    \begin{align}
    \label{eq:epsilon_scaling}
        \epsilon = 1-F_\mathrm{est} =  \begin{cases}
            \Theta(n^{-2t}) & (0\leq t \leq {1\over 2})\\
            \Theta(n^{-1}) & ({1\over 2}\leq t \leq 1)
        \end{cases}
    \end{align}
    holds, and we obtain
    \begin{align}
        c_P \leq h(t)\log\Theta(\epsilon^{-1}),
    \end{align}
    where $h(t)$ is defined in Eq.~\eqref{eq:h_alpha_def}.
    We prove the converse bound to complete the proof.
    To this end, we define a subset $\tilde{\mathbb{S}}_\mathrm{Young} \subset \Syoung$ by
    \begin{align}
        \tilde{\mathbb{S}}_\mathrm{Young}\coloneqq \left\{\alpha\in\Syoung \; \middle | \;{g(n)\over 2}\geq \tilde{\alpha}_1\geq \cdots \geq \tilde{\alpha}_{d-1}\right\},
    \end{align}
    where $\tilde{\alpha}_i$ is defined in Eq.~\eqref{eq:def_Syoung}.
    Since for all $i<j$, any $\alpha\in \tilde{\mathbb{S}}_\mathrm{Young}$ satisfies [see Eqs.~\eqref{eq:def_Syoung}--\eqref{eq:def_qr}]
    \begin{align}
        \alpha_i-\alpha_j
        &\geq A_i-A_j\\
        &\geq (j-i)g(n)\\
        &\geq \Theta(g(n))
    \end{align}
    we obtain
    \begin{align}
        d_\alpha^{(d)}&\geq \Theta(g(n)^{d(d-1)\over 2})
    \end{align}
    using Eq.~\eqref{eq:d_alpha^d}.
    Since for all $i$, any $\alpha\in \tilde{\mathbb{S}}_\mathrm{Young}$ satisfies [see Eqs.~\eqref{eq:def_Syoung}--\eqref{eq:def_qr}]
    \begin{align}
        \alpha_i
        &\geq \alpha_d\\
        &= n-\sum_{i=1}^{d-1} (A_i+\tilde{\alpha}_i)\\
        &= q-\sum_{i=1}^{d-1}\tilde{\alpha}_i\\
        &\geq {1\over d} \left[n-{d(d-1)\over 2} g(n)\right]-1-(d-1){g(n)\over 2}\\
        &\geq {1\over d}\left[n-d(d-1)g(n)\right]-1\\
        &\geq {n\over 3d}-1\\
        &\geq \Theta(n),
    \end{align}
    where we use $g(n)\leq {2\over 3(d-1)}({n\over d}+d-2)$.
    Therefore, we obtain
    \begin{align}
        d_\alpha^{(D)}&\geq \Theta(g(n)^{d(d-1) \over 2} n^{d(D-d)}),
    \end{align}
    using Eq.~\eqref{eq:d_alpha^D}.
    The cardinality of $\tilde{\mathbb{S}}_\mathrm{Young}$ satisfies
    \begin{align}
        \abs{\tilde{\mathbb{S}}_\mathrm{Young}} \geq {\abs{\Syoung} \over 2^d(d-1)!}.
    \end{align}
    Therefore, we obtain
    \begin{align}
        c_P\geq (d^2-1) \log \Theta(g(n)) + d(D-d)\log  \Theta(n).
    \end{align}
    Using Eq.~\eqref{eq:epsilon_scaling}, we obtain
    \begin{align}
        c_P\geq h(t)\log \Theta(\epsilon^{-1}).
    \end{align}
\end{proof}

\bibliography{main}

\begin{thebibliography}{110}%
\makeatletter
\providecommand \@ifxundefined [1]{%
 \@ifx{#1\undefined}
}%
\providecommand \@ifnum [1]{%
 \ifnum #1\expandafter \@firstoftwo
 \else \expandafter \@secondoftwo
 \fi
}%
\providecommand \@ifx [1]{%
 \ifx #1\expandafter \@firstoftwo
 \else \expandafter \@secondoftwo
 \fi
}%
\providecommand \natexlab [1]{#1}%
\providecommand \enquote  [1]{``#1''}%
\providecommand \bibnamefont  [1]{#1}%
\providecommand \bibfnamefont [1]{#1}%
\providecommand \citenamefont [1]{#1}%
\providecommand \href@noop [0]{\@secondoftwo}%
\providecommand \href [0]{\begingroup \@sanitize@url \@href}%
\providecommand \@href[1]{\@@startlink{#1}\@@href}%
\providecommand \@@href[1]{\endgroup#1\@@endlink}%
\providecommand \@sanitize@url [0]{\catcode `\\12\catcode `\$12\catcode
  `\&12\catcode `\#12\catcode `\^12\catcode `\_12\catcode `\%12\relax}%
\providecommand \@@startlink[1]{}%
\providecommand \@@endlink[0]{}%
\providecommand \url  [0]{\begingroup\@sanitize@url \@url }%
\providecommand \@url [1]{\endgroup\@href {#1}{\urlprefix }}%
\providecommand \urlprefix  [0]{URL }%
\providecommand \Eprint [0]{\href }%
\providecommand \doibase [0]{https://doi.org/}%
\providecommand \selectlanguage [0]{\@gobble}%
\providecommand \bibinfo  [0]{\@secondoftwo}%
\providecommand \bibfield  [0]{\@secondoftwo}%
\providecommand \translation [1]{[#1]}%
\providecommand \BibitemOpen [0]{}%
\providecommand \bibitemStop [0]{}%
\providecommand \bibitemNoStop [0]{.\EOS\space}%
\providecommand \EOS [0]{\spacefactor3000\relax}%
\providecommand \BibitemShut  [1]{\csname bibitem#1\endcsname}%
\let\auto@bib@innerbib\@empty
\bibitem [{\citenamefont {Nielsen}\ and\ \citenamefont
  {Chuang}(1997)}]{nielsen1997programmable}%
  \BibitemOpen
  \bibfield  {author} {\bibinfo {author} {\bibfnamefont {M.~A.}\ \bibnamefont
  {Nielsen}}\ and\ \bibinfo {author} {\bibfnamefont {I.~L.}\ \bibnamefont
  {Chuang}},\ }\bibfield  {title} {\bibinfo {title} {{Programmable Quantum Gate
  Arrays}},\ }\href {https://doi.org/10.1103/PhysRevLett.79.321} {\bibfield
  {journal} {\bibinfo  {journal} {Phys. Rev. Lett.}\ }\textbf {\bibinfo
  {volume} {79}},\ \bibinfo {pages} {321} (\bibinfo {year} {1997})},\ \Eprint
  {https://arxiv.org/abs/quant-ph/9703032} {arXiv:quant-ph/9703032}
  \BibitemShut {NoStop}%
\bibitem [{\citenamefont {Kim}\ \emph {et~al.}(2001)\citenamefont {Kim},
  \citenamefont {Cheong}, \citenamefont {Lee},\ and\ \citenamefont
  {Lee}}]{kim2001storing}%
  \BibitemOpen
  \bibfield  {author} {\bibinfo {author} {\bibfnamefont {J.}~\bibnamefont
  {Kim}}, \bibinfo {author} {\bibfnamefont {Y.}~\bibnamefont {Cheong}},
  \bibinfo {author} {\bibfnamefont {J.-S.}\ \bibnamefont {Lee}},\ and\ \bibinfo
  {author} {\bibfnamefont {S.}~\bibnamefont {Lee}},\ }\bibfield  {title}
  {\bibinfo {title} {Storing unitary operators in quantum states},\ }\href
  {https://doi.org/10.1103/PhysRevA.65.012302} {\bibfield  {journal} {\bibinfo
  {journal} {Phys. Rev. A}\ }\textbf {\bibinfo {volume} {65}},\ \bibinfo
  {pages} {012302} (\bibinfo {year} {2001})},\ \Eprint
  {https://arxiv.org/abs/quant-ph/0109097} {arXiv:quant-ph/0109097}
  \BibitemShut {NoStop}%
\bibitem [{\citenamefont {Vidal}\ \emph {et~al.}(2002)\citenamefont {Vidal},
  \citenamefont {Masanes},\ and\ \citenamefont {Cirac}}]{vidal2002storing}%
  \BibitemOpen
  \bibfield  {author} {\bibinfo {author} {\bibfnamefont {G.}~\bibnamefont
  {Vidal}}, \bibinfo {author} {\bibfnamefont {L.}~\bibnamefont {Masanes}},\
  and\ \bibinfo {author} {\bibfnamefont {J.~I.}\ \bibnamefont {Cirac}},\
  }\bibfield  {title} {\bibinfo {title} {{Storing Quantum Dynamics in Quantum
  States: A Stochastic Programmable Gate}},\ }\href
  {https://doi.org/10.1103/PhysRevLett.88.047905} {\bibfield  {journal}
  {\bibinfo  {journal} {Phys. Rev. Lett.}\ }\textbf {\bibinfo {volume} {88}},\
  \bibinfo {pages} {047905} (\bibinfo {year} {2002})},\ \Eprint
  {https://arxiv.org/abs/quant-ph/0102037} {arXiv:quant-ph/0102037}
  \BibitemShut {NoStop}%
\bibitem [{\citenamefont {Hillery}\ \emph
  {et~al.}(2002{\natexlab{a}})\citenamefont {Hillery}, \citenamefont
  {Bu\ifmmode~\check{z}\else \v{z}\fi{}ek},\ and\ \citenamefont
  {Ziman}}]{hillery2002probabilistic}%
  \BibitemOpen
  \bibfield  {author} {\bibinfo {author} {\bibfnamefont {M.}~\bibnamefont
  {Hillery}}, \bibinfo {author} {\bibfnamefont {V.}~\bibnamefont
  {Bu\ifmmode~\check{z}\else \v{z}\fi{}ek}},\ and\ \bibinfo {author}
  {\bibfnamefont {M.}~\bibnamefont {Ziman}},\ }\bibfield  {title} {\bibinfo
  {title} {Probabilistic implementation of universal quantum processors},\
  }\href {https://doi.org/10.1103/PhysRevA.65.022301} {\bibfield  {journal}
  {\bibinfo  {journal} {Phys. Rev. A}\ }\textbf {\bibinfo {volume} {65}},\
  \bibinfo {pages} {022301} (\bibinfo {year} {2002}{\natexlab{a}})},\ \Eprint
  {https://arxiv.org/abs/quant-ph/0106088} {arXiv:quant-ph/0106088}
  \BibitemShut {NoStop}%
\bibitem [{\citenamefont {Hillery}\ \emph
  {et~al.}(2002{\natexlab{b}})\citenamefont {Hillery}, \citenamefont {Ziman},\
  and\ \citenamefont {Bu\ifmmode~\check{z}\else
  \v{z}\fi{}ek}}]{hillery2002implementation}%
  \BibitemOpen
  \bibfield  {author} {\bibinfo {author} {\bibfnamefont {M.}~\bibnamefont
  {Hillery}}, \bibinfo {author} {\bibfnamefont {M.}~\bibnamefont {Ziman}},\
  and\ \bibinfo {author} {\bibfnamefont {V.}~\bibnamefont
  {Bu\ifmmode~\check{z}\else \v{z}\fi{}ek}},\ }\bibfield  {title} {\bibinfo
  {title} {Implementation of quantum maps by programmable quantum processors},\
  }\href {https://doi.org/10.1103/PhysRevA.66.042302} {\bibfield  {journal}
  {\bibinfo  {journal} {Phys. Rev. A}\ }\textbf {\bibinfo {volume} {66}},\
  \bibinfo {pages} {042302} (\bibinfo {year} {2002}{\natexlab{b}})}\BibitemShut
  {NoStop}%
\bibitem [{\citenamefont {Winter}(2002)}]{winter2002scalable}%
  \BibitemOpen
  \bibfield  {author} {\bibinfo {author} {\bibfnamefont {A.}~\bibnamefont
  {Winter}},\ }\bibfield  {title} {\bibinfo {title} {Scalable programmable
  quantum gates and a new aspect of the additivity problem for the classical
  capacity of quantum channels},\ }\href {https://doi.org/10.1063/1.1498489}
  {\bibfield  {journal} {\bibinfo  {journal} {Journal of Mathematical Physics}\
  }\textbf {\bibinfo {volume} {43}},\ \bibinfo {pages} {4341} (\bibinfo {year}
  {2002})},\ \Eprint {https://arxiv.org/abs/quant-ph/0108066}
  {arXiv:quant-ph/0108066} \BibitemShut {NoStop}%
\bibitem [{\citenamefont {Yu}\ \emph {et~al.}(2002)\citenamefont {Yu},
  \citenamefont {Feng},\ and\ \citenamefont {Zhan}}]{yu2002milti}%
  \BibitemOpen
  \bibfield  {author} {\bibinfo {author} {\bibfnamefont {Y.}~\bibnamefont
  {Yu}}, \bibinfo {author} {\bibfnamefont {J.}~\bibnamefont {Feng}},\ and\
  \bibinfo {author} {\bibfnamefont {M.}~\bibnamefont {Zhan}},\ }\bibfield
  {title} {\bibinfo {title} {Multi-output programmable quantum processor},\
  }\href {https://doi.org/10.1103/PhysRevA.66.052310} {\bibfield  {journal}
  {\bibinfo  {journal} {Phys. Rev. A}\ }\textbf {\bibinfo {volume} {66}},\
  \bibinfo {pages} {052310} (\bibinfo {year} {2002})},\ \Eprint
  {https://arxiv.org/abs/quant-ph/0209069} {arXiv:quant-ph/0209069}
  \BibitemShut {NoStop}%
\bibitem [{\citenamefont {Hillery}\ \emph {et~al.}(2004)\citenamefont
  {Hillery}, \citenamefont {Ziman},\ and\ \citenamefont
  {Bu\ifmmode~\check{z}\else \v{z}\fi{}ek}}]{hillery2004improving}%
  \BibitemOpen
  \bibfield  {author} {\bibinfo {author} {\bibfnamefont {M.}~\bibnamefont
  {Hillery}}, \bibinfo {author} {\bibfnamefont {M.}~\bibnamefont {Ziman}},\
  and\ \bibinfo {author} {\bibfnamefont {V.}~\bibnamefont
  {Bu\ifmmode~\check{z}\else \v{z}\fi{}ek}},\ }\bibfield  {title} {\bibinfo
  {title} {Improving the performance of probabilistic programmable quantum
  processors},\ }\href {https://doi.org/10.1103/PhysRevA.69.042311} {\bibfield
  {journal} {\bibinfo  {journal} {Phys. Rev. A}\ }\textbf {\bibinfo {volume}
  {69}},\ \bibinfo {pages} {042311} (\bibinfo {year} {2004})},\ \Eprint
  {https://arxiv.org/abs/quant-ph/0311170} {arXiv:quant-ph/0311170}
  \BibitemShut {NoStop}%
\bibitem [{\citenamefont {Brazier}\ \emph {et~al.}(2005)\citenamefont
  {Brazier}, \citenamefont {Bu\ifmmode~\check{z}\else \v{z}\fi{}ek},\ and\
  \citenamefont {Knight}}]{brazier2005probabilistic}%
  \BibitemOpen
  \bibfield  {author} {\bibinfo {author} {\bibfnamefont {A.}~\bibnamefont
  {Brazier}}, \bibinfo {author} {\bibfnamefont {V.}~\bibnamefont
  {Bu\ifmmode~\check{z}\else \v{z}\fi{}ek}},\ and\ \bibinfo {author}
  {\bibfnamefont {P.~L.}\ \bibnamefont {Knight}},\ }\bibfield  {title}
  {\bibinfo {title} {Probabilistic programmable quantum processors with
  multiple copies of program states},\ }\href
  {https://doi.org/10.1103/PhysRevA.71.032306} {\bibfield  {journal} {\bibinfo
  {journal} {Phys. Rev. A}\ }\textbf {\bibinfo {volume} {71}},\ \bibinfo
  {pages} {032306} (\bibinfo {year} {2005})},\ \Eprint
  {https://arxiv.org/abs/quant-ph/0505202} {arXiv:quant-ph/0505202}
  \BibitemShut {NoStop}%
\bibitem [{\citenamefont {Hillery}\ \emph {et~al.}(2006)\citenamefont
  {Hillery}, \citenamefont {Ziman},\ and\ \citenamefont
  {Bu\ifmmode~\check{z}\else \v{z}\fi{}ek}}]{hillery2006approximate}%
  \BibitemOpen
  \bibfield  {author} {\bibinfo {author} {\bibfnamefont {M.}~\bibnamefont
  {Hillery}}, \bibinfo {author} {\bibfnamefont {M.}~\bibnamefont {Ziman}},\
  and\ \bibinfo {author} {\bibfnamefont {V.}~\bibnamefont
  {Bu\ifmmode~\check{z}\else \v{z}\fi{}ek}},\ }\bibfield  {title} {\bibinfo
  {title} {Approximate programmable quantum processors},\ }\href
  {https://doi.org/10.1103/PhysRevA.73.022345} {\bibfield  {journal} {\bibinfo
  {journal} {Phys. Rev. A}\ }\textbf {\bibinfo {volume} {73}},\ \bibinfo
  {pages} {022345} (\bibinfo {year} {2006})},\ \Eprint
  {https://arxiv.org/abs/quant-ph/0510161} {arXiv:quant-ph/0510161}
  \BibitemShut {NoStop}%
\bibitem [{\citenamefont {Ishizaka}\ and\ \citenamefont
  {Hiroshima}(2008)}]{ishizaka2008asymptotic}%
  \BibitemOpen
  \bibfield  {author} {\bibinfo {author} {\bibfnamefont {S.}~\bibnamefont
  {Ishizaka}}\ and\ \bibinfo {author} {\bibfnamefont {T.}~\bibnamefont
  {Hiroshima}},\ }\bibfield  {title} {\bibinfo {title} {{Asymptotic
  Teleportation Scheme as a Universal Programmable Quantum Processor}},\ }\href
  {https://doi.org/10.1103/PhysRevLett.101.240501} {\bibfield  {journal}
  {\bibinfo  {journal} {Phys. Rev. Lett.}\ }\textbf {\bibinfo {volume} {101}},\
  \bibinfo {pages} {240501} (\bibinfo {year} {2008})},\ \Eprint
  {https://arxiv.org/abs/0807.4568} {arXiv:0807.4568} \BibitemShut {NoStop}%
\bibitem [{\citenamefont {Ishizaka}\ and\ \citenamefont
  {Hiroshima}(2009)}]{ishizaka2009quantum}%
  \BibitemOpen
  \bibfield  {author} {\bibinfo {author} {\bibfnamefont {S.}~\bibnamefont
  {Ishizaka}}\ and\ \bibinfo {author} {\bibfnamefont {T.}~\bibnamefont
  {Hiroshima}},\ }\bibfield  {title} {\bibinfo {title} {Quantum teleportation
  scheme by selecting one of multiple output ports},\ }\href
  {https://doi.org/10.1103/PhysRevA.79.042306} {\bibfield  {journal} {\bibinfo
  {journal} {Phys. Rev. A}\ }\textbf {\bibinfo {volume} {79}},\ \bibinfo
  {pages} {042306} (\bibinfo {year} {2009})},\ \Eprint
  {https://arxiv.org/abs/0901.2975} {arXiv:0901.2975} \BibitemShut {NoStop}%
\bibitem [{\citenamefont {Sedl\'{a}k}\ \emph {et~al.}(2019)\citenamefont
  {Sedl\'{a}k}, \citenamefont {Bisio},\ and\ \citenamefont
  {Ziman}}]{sedlak2019optimal}%
  \BibitemOpen
  \bibfield  {author} {\bibinfo {author} {\bibfnamefont {M.}~\bibnamefont
  {Sedl\'{a}k}}, \bibinfo {author} {\bibfnamefont {A.}~\bibnamefont {Bisio}},\
  and\ \bibinfo {author} {\bibfnamefont {M.}~\bibnamefont {Ziman}},\ }\bibfield
   {title} {\bibinfo {title} {{Optimal Probabilistic Storage and Retrieval of
  Unitary Channels}},\ }\href {https://doi.org/10.1103/PhysRevLett.122.170502}
  {\bibfield  {journal} {\bibinfo  {journal} {Physical Review Letters}\
  }\textbf {\bibinfo {volume} {122}},\ \bibinfo {pages} {170502} (\bibinfo
  {year} {2019})},\ \Eprint {https://arxiv.org/abs/1809.04552}
  {arXiv:1809.04552 [quant-ph]} \BibitemShut {NoStop}%
\bibitem [{\citenamefont {Kubicki}\ \emph {et~al.}(2019)\citenamefont
  {Kubicki}, \citenamefont {Palazuelos},\ and\ \citenamefont
  {P\'erez-Garc\'{\i}a}}]{kubicki2019resource}%
  \BibitemOpen
  \bibfield  {author} {\bibinfo {author} {\bibfnamefont {A.~M.}\ \bibnamefont
  {Kubicki}}, \bibinfo {author} {\bibfnamefont {C.}~\bibnamefont
  {Palazuelos}},\ and\ \bibinfo {author} {\bibfnamefont {D.}~\bibnamefont
  {P\'erez-Garc\'{\i}a}},\ }\bibfield  {title} {\bibinfo {title} {{Resource
  Quantification for the No-Programing Theorem}},\ }\href
  {https://doi.org/10.1103/PhysRevLett.122.080505} {\bibfield  {journal}
  {\bibinfo  {journal} {Phys. Rev. Lett.}\ }\textbf {\bibinfo {volume} {122}},\
  \bibinfo {pages} {080505} (\bibinfo {year} {2019})},\ \Eprint
  {https://arxiv.org/abs/1805.00756} {arXiv:1805.00756} \BibitemShut {NoStop}%
\bibitem [{\citenamefont {Sedl\'{a}k}\ and\ \citenamefont
  {Ziman}(2020)}]{sedlak2020probabilistic}%
  \BibitemOpen
  \bibfield  {author} {\bibinfo {author} {\bibfnamefont {M.}~\bibnamefont
  {Sedl\'{a}k}}\ and\ \bibinfo {author} {\bibfnamefont {M.}~\bibnamefont
  {Ziman}},\ }\bibfield  {title} {\bibinfo {title} {Probabilistic storage and
  retrieval of qubit phase gates},\ }\href
  {https://doi.org/10.1103/PhysRevA.102.032618} {\bibfield  {journal} {\bibinfo
   {journal} {Phys. Rev. A}\ }\textbf {\bibinfo {volume} {102}},\ \bibinfo
  {pages} {032618} (\bibinfo {year} {2020})},\ \Eprint
  {https://arxiv.org/abs/2008.09555} {arXiv:2008.09555} \BibitemShut {NoStop}%
\bibitem [{\citenamefont {Yang}\ \emph {et~al.}(2020)\citenamefont {Yang},
  \citenamefont {Renner},\ and\ \citenamefont {Chiribella}}]{yang2020optimal}%
  \BibitemOpen
  \bibfield  {author} {\bibinfo {author} {\bibfnamefont {Y.}~\bibnamefont
  {Yang}}, \bibinfo {author} {\bibfnamefont {R.}~\bibnamefont {Renner}},\ and\
  \bibinfo {author} {\bibfnamefont {G.}~\bibnamefont {Chiribella}},\ }\bibfield
   {title} {\bibinfo {title} {{Optimal Universal Programming of Unitary
  Gates}},\ }\href {https://doi.org/10.1103/PhysRevLett.125.210501} {\bibfield
  {journal} {\bibinfo  {journal} {Phys. Rev. Lett.}\ }\textbf {\bibinfo
  {volume} {125}},\ \bibinfo {pages} {210501} (\bibinfo {year} {2020})},\
  \Eprint {https://arxiv.org/abs/2007.10363} {arXiv:2007.10363} \BibitemShut
  {NoStop}%
\bibitem [{\citenamefont {Banchi}\ \emph {et~al.}(2020)\citenamefont {Banchi},
  \citenamefont {Pereira}, \citenamefont {Lloyd},\ and\ \citenamefont
  {Pirandola}}]{banchi2020convex}%
  \BibitemOpen
  \bibfield  {author} {\bibinfo {author} {\bibfnamefont {L.}~\bibnamefont
  {Banchi}}, \bibinfo {author} {\bibfnamefont {J.}~\bibnamefont {Pereira}},
  \bibinfo {author} {\bibfnamefont {S.}~\bibnamefont {Lloyd}},\ and\ \bibinfo
  {author} {\bibfnamefont {S.}~\bibnamefont {Pirandola}},\ }\bibfield  {title}
  {\bibinfo {title} {Convex optimization of programmable quantum computers},\
  }\href {https://doi.org/10.1038/s41534-020-0268-2} {\bibfield  {journal}
  {\bibinfo  {journal} {npj Quantum Information}\ }\textbf {\bibinfo {volume}
  {6}},\ \bibinfo {pages} {42} (\bibinfo {year} {2020})},\ \Eprint
  {https://arxiv.org/abs/1905.01316} {arXiv:1905.01316} \BibitemShut {NoStop}%
\bibitem [{\citenamefont {Gschwendtner}\ \emph {et~al.}(2021)\citenamefont
  {Gschwendtner}, \citenamefont {Bluhm},\ and\ \citenamefont
  {Winter}}]{gschwendtner2021programmability}%
  \BibitemOpen
  \bibfield  {author} {\bibinfo {author} {\bibfnamefont {M.}~\bibnamefont
  {Gschwendtner}}, \bibinfo {author} {\bibfnamefont {A.}~\bibnamefont
  {Bluhm}},\ and\ \bibinfo {author} {\bibfnamefont {A.}~\bibnamefont
  {Winter}},\ }\bibfield  {title} {\bibinfo {title} {Programmability of
  covariant quantum channels},\ }\href
  {https://doi.org/10.22331/q-2021-06-29-488} {\bibfield  {journal} {\bibinfo
  {journal} {Quantum}\ }\textbf {\bibinfo {volume} {5}},\ \bibinfo {pages}
  {488} (\bibinfo {year} {2021})},\ \Eprint {https://arxiv.org/abs/2012.00717}
  {arXiv:2012.00717} \BibitemShut {NoStop}%
\bibitem [{\citenamefont {Pavli\ifmmode~\check{c}\else \v{c}\fi{}ko}\ and\
  \citenamefont {Ziman}(2022)}]{pavlivcko2022robustness}%
  \BibitemOpen
  \bibfield  {author} {\bibinfo {author} {\bibfnamefont {J.}~\bibnamefont
  {Pavli\ifmmode~\check{c}\else \v{c}\fi{}ko}}\ and\ \bibinfo {author}
  {\bibfnamefont {M.}~\bibnamefont {Ziman}},\ }\bibfield  {title} {\bibinfo
  {title} {Robustness of optimal probabilistic storage and retrieval of unitary
  channels to noise},\ }\href {https://doi.org/10.1103/PhysRevA.106.052416}
  {\bibfield  {journal} {\bibinfo  {journal} {Phys. Rev. A}\ }\textbf {\bibinfo
  {volume} {106}},\ \bibinfo {pages} {052416} (\bibinfo {year} {2022})},\
  \Eprint {https://arxiv.org/abs/2211.07079} {arXiv:2211.07079} \BibitemShut
  {NoStop}%
\bibitem [{\citenamefont {Schoute}\ \emph {et~al.}(2024)\citenamefont
  {Schoute}, \citenamefont {Grinko}, \citenamefont {Subasi},\ and\
  \citenamefont {Volkoff}}]{schoute2024quantum}%
  \BibitemOpen
  \bibfield  {author} {\bibinfo {author} {\bibfnamefont {E.}~\bibnamefont
  {Schoute}}, \bibinfo {author} {\bibfnamefont {D.}~\bibnamefont {Grinko}},
  \bibinfo {author} {\bibfnamefont {Y.}~\bibnamefont {Subasi}},\ and\ \bibinfo
  {author} {\bibfnamefont {T.}~\bibnamefont {Volkoff}},\ }\bibfield  {title}
  {\bibinfo {title} {Quantum programmable reflections},\ }\Eprint
  {https://arxiv.org/abs/2411.03648} {arXiv:2411.03648}  (\bibinfo {year}
  {2024})\BibitemShut {NoStop}%
\bibitem [{\citenamefont {Du\ifmmode~\check{s}\else \v{s}\fi{}ek}\ and\
  \citenamefont {Bu\ifmmode~\check{z}\else
  \v{z}\fi{}ek}(2002)}]{duvsek2002quantum}%
  \BibitemOpen
  \bibfield  {author} {\bibinfo {author} {\bibfnamefont {M.}~\bibnamefont
  {Du\ifmmode~\check{s}\else \v{s}\fi{}ek}}\ and\ \bibinfo {author}
  {\bibfnamefont {V.}~\bibnamefont {Bu\ifmmode~\check{z}\else \v{z}\fi{}ek}},\
  }\bibfield  {title} {\bibinfo {title} {Quantum-controlled measurement device
  for quantum-state discrimination},\ }\href
  {https://doi.org/10.1103/PhysRevA.66.022112} {\bibfield  {journal} {\bibinfo
  {journal} {Phys. Rev. A}\ }\textbf {\bibinfo {volume} {66}},\ \bibinfo
  {pages} {022112} (\bibinfo {year} {2002})}\BibitemShut {NoStop}%
\bibitem [{\citenamefont {Fiur\'a\ifmmode~\check{s}\else \v{s}\fi{}ek}\ \emph
  {et~al.}(2002)\citenamefont {Fiur\'a\ifmmode~\check{s}\else \v{s}\fi{}ek},
  \citenamefont {Du\ifmmode~\check{s}\else \v{s}\fi{}ek},\ and\ \citenamefont
  {Filip}}]{fiuravsek2002universal}%
  \BibitemOpen
  \bibfield  {author} {\bibinfo {author} {\bibfnamefont {J.}~\bibnamefont
  {Fiur\'a\ifmmode~\check{s}\else \v{s}\fi{}ek}}, \bibinfo {author}
  {\bibfnamefont {M.}~\bibnamefont {Du\ifmmode~\check{s}\else \v{s}\fi{}ek}},\
  and\ \bibinfo {author} {\bibfnamefont {R.}~\bibnamefont {Filip}},\ }\bibfield
   {title} {\bibinfo {title} {{Universal Measurement Apparatus Controlled by
  Quantum Software}},\ }\href {https://doi.org/10.1103/PhysRevLett.89.190401}
  {\bibfield  {journal} {\bibinfo  {journal} {Phys. Rev. Lett.}\ }\textbf
  {\bibinfo {volume} {89}},\ \bibinfo {pages} {190401} (\bibinfo {year}
  {2002})},\ \Eprint {https://arxiv.org/abs/quant-ph/0202152}
  {arXiv:quant-ph/0202152} \BibitemShut {NoStop}%
\bibitem [{\citenamefont {Paz}\ and\ \citenamefont
  {Roncaglia}(2003)}]{paz2003quantum}%
  \BibitemOpen
  \bibfield  {author} {\bibinfo {author} {\bibfnamefont {J.~P.}\ \bibnamefont
  {Paz}}\ and\ \bibinfo {author} {\bibfnamefont {A.}~\bibnamefont
  {Roncaglia}},\ }\bibfield  {title} {\bibinfo {title} {Quantum gate arrays can
  be programmed to evaluate the expectation value of any operator},\ }\href
  {https://doi.org/10.1103/PhysRevA.68.052316} {\bibfield  {journal} {\bibinfo
  {journal} {Phys. Rev. A}\ }\textbf {\bibinfo {volume} {68}},\ \bibinfo
  {pages} {052316} (\bibinfo {year} {2003})},\ \Eprint
  {https://arxiv.org/abs/quant-ph/0306143} {arXiv:quant-ph/0306143}
  \BibitemShut {NoStop}%
\bibitem [{\citenamefont {Ro\ifmmode~\check{s}\else \v{s}\fi{}ko}\ \emph
  {et~al.}(2003)\citenamefont {Ro\ifmmode~\check{s}\else \v{s}\fi{}ko},
  \citenamefont {Bu\ifmmode~\check{z}\else \v{z}\fi{}ek}, \citenamefont
  {Chouha},\ and\ \citenamefont {Hillery}}]{rovsko2003generalized}%
  \BibitemOpen
  \bibfield  {author} {\bibinfo {author} {\bibfnamefont {M.}~\bibnamefont
  {Ro\ifmmode~\check{s}\else \v{s}\fi{}ko}}, \bibinfo {author} {\bibfnamefont
  {V.}~\bibnamefont {Bu\ifmmode~\check{z}\else \v{z}\fi{}ek}}, \bibinfo
  {author} {\bibfnamefont {P.~R.}\ \bibnamefont {Chouha}},\ and\ \bibinfo
  {author} {\bibfnamefont {M.}~\bibnamefont {Hillery}},\ }\bibfield  {title}
  {\bibinfo {title} {Generalized measurements via a programmable quantum
  processor},\ }\href {https://doi.org/10.1103/PhysRevA.68.062302} {\bibfield
  {journal} {\bibinfo  {journal} {Phys. Rev. A}\ }\textbf {\bibinfo {volume}
  {68}},\ \bibinfo {pages} {062302} (\bibinfo {year} {2003})},\ \Eprint
  {https://arxiv.org/abs/quant-ph/0311172} {arXiv:quant-ph/0311172}
  \BibitemShut {NoStop}%
\bibitem [{\citenamefont {Fiur\'a\ifmmode~\check{s}\else \v{s}\fi{}ek}\ and\
  \citenamefont {Du\ifmmode~\check{s}\else
  \v{s}\fi{}ek}(2004)}]{fiuravsek2004probabilistic}%
  \BibitemOpen
  \bibfield  {author} {\bibinfo {author} {\bibfnamefont {J.}~\bibnamefont
  {Fiur\'a\ifmmode~\check{s}\else \v{s}\fi{}ek}}\ and\ \bibinfo {author}
  {\bibfnamefont {M.}~\bibnamefont {Du\ifmmode~\check{s}\else \v{s}\fi{}ek}},\
  }\bibfield  {title} {\bibinfo {title} {Probabilistic quantum multimeters},\
  }\href {https://doi.org/10.1103/PhysRevA.69.032302} {\bibfield  {journal}
  {\bibinfo  {journal} {Phys. Rev. A}\ }\textbf {\bibinfo {volume} {69}},\
  \bibinfo {pages} {032302} (\bibinfo {year} {2004})},\ \Eprint
  {https://arxiv.org/abs/quant-ph/0308111} {arXiv:quant-ph/0308111}
  \BibitemShut {NoStop}%
\bibitem [{\citenamefont {D'Ariano}\ and\ \citenamefont
  {Perinotti}(2005)}]{dariano2005efficient}%
  \BibitemOpen
  \bibfield  {author} {\bibinfo {author} {\bibfnamefont {G.~M.}\ \bibnamefont
  {D'Ariano}}\ and\ \bibinfo {author} {\bibfnamefont {P.}~\bibnamefont
  {Perinotti}},\ }\bibfield  {title} {\bibinfo {title} {{Efficient Universal
  Programmable Quantum Measurements}},\ }\href
  {https://doi.org/10.1103/PhysRevLett.94.090401} {\bibfield  {journal}
  {\bibinfo  {journal} {Phys. Rev. Lett.}\ }\textbf {\bibinfo {volume} {94}},\
  \bibinfo {pages} {090401} (\bibinfo {year} {2005})},\ \Eprint
  {https://arxiv.org/abs/quant-ph/0410169} {arXiv:quant-ph/0410169}
  \BibitemShut {NoStop}%
\bibitem [{\citenamefont {Bergou}\ \emph {et~al.}(2006)\citenamefont {Bergou},
  \citenamefont {Bu\ifmmode~\check{z}\else \v{z}\fi{}ek}, \citenamefont
  {Feldman}, \citenamefont {Herzog},\ and\ \citenamefont
  {Hillery}}]{bergou2006programmable}%
  \BibitemOpen
  \bibfield  {author} {\bibinfo {author} {\bibfnamefont {J.~A.}\ \bibnamefont
  {Bergou}}, \bibinfo {author} {\bibfnamefont {V.}~\bibnamefont
  {Bu\ifmmode~\check{z}\else \v{z}\fi{}ek}}, \bibinfo {author} {\bibfnamefont
  {E.}~\bibnamefont {Feldman}}, \bibinfo {author} {\bibfnamefont
  {U.}~\bibnamefont {Herzog}},\ and\ \bibinfo {author} {\bibfnamefont
  {M.}~\bibnamefont {Hillery}},\ }\bibfield  {title} {\bibinfo {title}
  {Programmable quantum-state discriminators with simple programs},\ }\href
  {https://doi.org/10.1103/PhysRevA.73.062334} {\bibfield  {journal} {\bibinfo
  {journal} {Phys. Rev. A}\ }\textbf {\bibinfo {volume} {73}},\ \bibinfo
  {pages} {062334} (\bibinfo {year} {2006})}\BibitemShut {NoStop}%
\bibitem [{\citenamefont {Zhang}\ \emph {et~al.}(2006)\citenamefont {Zhang},
  \citenamefont {Ying},\ and\ \citenamefont {Qiao}}]{zhang2006universal}%
  \BibitemOpen
  \bibfield  {author} {\bibinfo {author} {\bibfnamefont {C.}~\bibnamefont
  {Zhang}}, \bibinfo {author} {\bibfnamefont {M.}~\bibnamefont {Ying}},\ and\
  \bibinfo {author} {\bibfnamefont {B.}~\bibnamefont {Qiao}},\ }\bibfield
  {title} {\bibinfo {title} {Universal programmable devices for unambiguous
  discrimination},\ }\href {https://doi.org/10.1103/PhysRevA.74.042308}
  {\bibfield  {journal} {\bibinfo  {journal} {Phys. Rev. A}\ }\textbf {\bibinfo
  {volume} {74}},\ \bibinfo {pages} {042308} (\bibinfo {year} {2006})},\
  \Eprint {https://arxiv.org/abs/quant-ph/0606189} {arXiv:quant-ph/0606189}
  \BibitemShut {NoStop}%
\bibitem [{\citenamefont {P\'erez-Garc\'{\i}a}(2006)}]{perez2006optimality}%
  \BibitemOpen
  \bibfield  {author} {\bibinfo {author} {\bibfnamefont {D.}~\bibnamefont
  {P\'erez-Garc\'{\i}a}},\ }\bibfield  {title} {\bibinfo {title} {Optimality of
  programmable quantum measurements},\ }\href
  {https://doi.org/10.1103/PhysRevA.73.052315} {\bibfield  {journal} {\bibinfo
  {journal} {Phys. Rev. A}\ }\textbf {\bibinfo {volume} {73}},\ \bibinfo
  {pages} {052315} (\bibinfo {year} {2006})},\ \Eprint
  {https://arxiv.org/abs/quant-ph/0602084} {arXiv:quant-ph/0602084}
  \BibitemShut {NoStop}%
\bibitem [{\citenamefont {He}\ and\ \citenamefont
  {Bergou}(2007)}]{he2007programmable}%
  \BibitemOpen
  \bibfield  {author} {\bibinfo {author} {\bibfnamefont {B.}~\bibnamefont
  {He}}\ and\ \bibinfo {author} {\bibfnamefont {J.~A.}\ \bibnamefont
  {Bergou}},\ }\bibfield  {title} {\bibinfo {title} {Programmable unknown
  quantum-state discriminators with multiple copies of program and data: A
  jordan-basis approach},\ }\href {https://doi.org/10.1103/PhysRevA.75.032316}
  {\bibfield  {journal} {\bibinfo  {journal} {Phys. Rev. A}\ }\textbf {\bibinfo
  {volume} {75}},\ \bibinfo {pages} {032316} (\bibinfo {year} {2007})},\
  \Eprint {https://arxiv.org/abs/quant-ph/0610226} {arXiv:quant-ph/0610226}
  \BibitemShut {NoStop}%
\bibitem [{\citenamefont {Sent\'{\i}s}\ \emph {et~al.}(2010)\citenamefont
  {Sent\'{\i}s}, \citenamefont {Bagan}, \citenamefont {Calsamiglia},\ and\
  \citenamefont {Mu\~noz Tapia}}]{sentis2010multicopy}%
  \BibitemOpen
  \bibfield  {author} {\bibinfo {author} {\bibfnamefont {G.}~\bibnamefont
  {Sent\'{\i}s}}, \bibinfo {author} {\bibfnamefont {E.}~\bibnamefont {Bagan}},
  \bibinfo {author} {\bibfnamefont {J.}~\bibnamefont {Calsamiglia}},\ and\
  \bibinfo {author} {\bibfnamefont {R.}~\bibnamefont {Mu\~noz Tapia}},\
  }\bibfield  {title} {\bibinfo {title} {Multicopy programmable discrimination
  of general qubit states},\ }\href
  {https://doi.org/10.1103/PhysRevA.82.042312} {\bibfield  {journal} {\bibinfo
  {journal} {Phys. Rev. A}\ }\textbf {\bibinfo {volume} {82}},\ \bibinfo
  {pages} {042312} (\bibinfo {year} {2010})}\BibitemShut {NoStop}%
\bibitem [{\citenamefont {Bisio}\ \emph {et~al.}(2011)\citenamefont {Bisio},
  \citenamefont {D'Ariano}, \citenamefont {Perinotti},\ and\ \citenamefont
  {Sedl\'{a}k}}]{bisio2011quantum}%
  \BibitemOpen
  \bibfield  {author} {\bibinfo {author} {\bibfnamefont {A.}~\bibnamefont
  {Bisio}}, \bibinfo {author} {\bibfnamefont {G.~M.}\ \bibnamefont {D'Ariano}},
  \bibinfo {author} {\bibfnamefont {P.}~\bibnamefont {Perinotti}},\ and\
  \bibinfo {author} {\bibfnamefont {M.}~\bibnamefont {Sedl\'{a}k}},\ }\bibfield
   {title} {\bibinfo {title} {Quantum learning algorithms for quantum
  measurements},\ }\href {https://doi.org/10.1016/j.physleta.2011.08.002}
  {\bibfield  {journal} {\bibinfo  {journal} {Physics Letters A}\ }\textbf
  {\bibinfo {volume} {375}},\ \bibinfo {pages} {3425} (\bibinfo {year}
  {2011})},\ \Eprint {https://arxiv.org/abs/1103.0480} {arXiv:1103.0480}
  \BibitemShut {NoStop}%
\bibitem [{\citenamefont {Zhou}\ \emph {et~al.}(2012)\citenamefont {Zhou},
  \citenamefont {Cui}, \citenamefont {Wu},\ and\ \citenamefont
  {Lon}}]{zhou2012multicopy}%
  \BibitemOpen
  \bibfield  {author} {\bibinfo {author} {\bibfnamefont {T.}~\bibnamefont
  {Zhou}}, \bibinfo {author} {\bibfnamefont {J.~X.}\ \bibnamefont {Cui}},
  \bibinfo {author} {\bibfnamefont {X.}~\bibnamefont {Wu}},\ and\ \bibinfo
  {author} {\bibfnamefont {G.~L.}\ \bibnamefont {Lon}},\ }\bibfield  {title}
  {\bibinfo {title} {Multicopy programmable discriminators between two unknown
  qubit states with group-theoretic approach},\ }\href
  {https://doi.org/10.26421/QIC12.11-12-9} {\bibfield  {journal} {\bibinfo
  {journal} {Quantum Information \& Computation}\ }\textbf {\bibinfo {volume}
  {12}},\ \bibinfo {pages} {1017} (\bibinfo {year} {2012})},\ \Eprint
  {https://arxiv.org/abs/1112.0931} {arXiv:1112.0931} \BibitemShut {NoStop}%
\bibitem [{\citenamefont {Zhou}(2014)}]{zhou2014success}%
  \BibitemOpen
  \bibfield  {author} {\bibinfo {author} {\bibfnamefont {T.}~\bibnamefont
  {Zhou}},\ }\bibfield  {title} {\bibinfo {title} {Success probabilities for
  universal unambiguous discriminators between unknown pure states},\ }\href
  {https://doi.org/10.1103/PhysRevA.89.014301} {\bibfield  {journal} {\bibinfo
  {journal} {Phys. Rev. A}\ }\textbf {\bibinfo {volume} {89}},\ \bibinfo
  {pages} {014301} (\bibinfo {year} {2014})},\ \Eprint
  {https://arxiv.org/abs/1308.0707} {arXiv:1308.0707} \BibitemShut {NoStop}%
\bibitem [{\citenamefont {Jafarizadeh}\ \emph {et~al.}(2017)\citenamefont
  {Jafarizadeh}, \citenamefont {Mahmoudi}, \citenamefont {Akhgar},\ and\
  \citenamefont {Faizi}}]{jafarizadeh2017designing}%
  \BibitemOpen
  \bibfield  {author} {\bibinfo {author} {\bibfnamefont {M.~A.}\ \bibnamefont
  {Jafarizadeh}}, \bibinfo {author} {\bibfnamefont {P.}~\bibnamefont
  {Mahmoudi}}, \bibinfo {author} {\bibfnamefont {D.}~\bibnamefont {Akhgar}},\
  and\ \bibinfo {author} {\bibfnamefont {E.}~\bibnamefont {Faizi}},\ }\bibfield
   {title} {\bibinfo {title} {Designing an optimal, universal, programmable,
  and unambiguous discriminator for $n$ unknown qubits},\ }\href
  {https://doi.org/10.1103/PhysRevA.96.052111} {\bibfield  {journal} {\bibinfo
  {journal} {Phys. Rev. A}\ }\textbf {\bibinfo {volume} {96}},\ \bibinfo
  {pages} {052111} (\bibinfo {year} {2017})}\BibitemShut {NoStop}%
\bibitem [{\citenamefont {Chabaud}\ \emph {et~al.}(2018)\citenamefont
  {Chabaud}, \citenamefont {Diamanti}, \citenamefont {Markham}, \citenamefont
  {Kashefi},\ and\ \citenamefont {Joux}}]{chabaud2018optimal}%
  \BibitemOpen
  \bibfield  {author} {\bibinfo {author} {\bibfnamefont {U.}~\bibnamefont
  {Chabaud}}, \bibinfo {author} {\bibfnamefont {E.}~\bibnamefont {Diamanti}},
  \bibinfo {author} {\bibfnamefont {D.}~\bibnamefont {Markham}}, \bibinfo
  {author} {\bibfnamefont {E.}~\bibnamefont {Kashefi}},\ and\ \bibinfo {author}
  {\bibfnamefont {A.}~\bibnamefont {Joux}},\ }\bibfield  {title} {\bibinfo
  {title} {Optimal quantum-programmable projective measurement with linear
  optics},\ }\href {https://doi.org/10.1103/PhysRevA.98.062318} {\bibfield
  {journal} {\bibinfo  {journal} {Phys. Rev. A}\ }\textbf {\bibinfo {volume}
  {98}},\ \bibinfo {pages} {062318} (\bibinfo {year} {2018})},\ \Eprint
  {https://arxiv.org/abs/1805.02546} {arXiv:1805.02546} \BibitemShut {NoStop}%
\bibitem [{\citenamefont {Lewandowska}\ \emph {et~al.}(2022)\citenamefont
  {Lewandowska}, \citenamefont {Kukulski}, \citenamefont {Pawela},\ and\
  \citenamefont {Pucha\l{}a}}]{lewandowska2022storage}%
  \BibitemOpen
  \bibfield  {author} {\bibinfo {author} {\bibfnamefont {P.}~\bibnamefont
  {Lewandowska}}, \bibinfo {author} {\bibfnamefont {R.}~\bibnamefont
  {Kukulski}}, \bibinfo {author} {\bibfnamefont {L.}~\bibnamefont {Pawela}},\
  and\ \bibinfo {author} {\bibfnamefont {Z.}~\bibnamefont {Pucha\l{}a}},\
  }\bibfield  {title} {\bibinfo {title} {{Storage and retrieval of von Neumann
  measurements}},\ }\href {https://doi.org/10.1103/PhysRevA.106.052423}
  {\bibfield  {journal} {\bibinfo  {journal} {Phys. Rev. A}\ }\textbf {\bibinfo
  {volume} {106}},\ \bibinfo {pages} {052423} (\bibinfo {year} {2022})},\
  \Eprint {https://arxiv.org/abs/2204.03029} {arXiv:2204.03029} \BibitemShut
  {NoStop}%
\bibitem [{\citenamefont {Gschwendtner}\ and\ \citenamefont
  {Winter}(2021)}]{gschwendtner2021infinite}%
  \BibitemOpen
  \bibfield  {author} {\bibinfo {author} {\bibfnamefont {M.}~\bibnamefont
  {Gschwendtner}}\ and\ \bibinfo {author} {\bibfnamefont {A.}~\bibnamefont
  {Winter}},\ }\bibfield  {title} {\bibinfo {title} {{Infinite-Dimensional
  Programmable Quantum Processors}},\ }\href
  {https://doi.org/10.1103/PRXQuantum.2.030308} {\bibfield  {journal} {\bibinfo
   {journal} {PRX Quantum}\ }\textbf {\bibinfo {volume} {2}},\ \bibinfo {pages}
  {030308} (\bibinfo {year} {2021})},\ \Eprint
  {https://arxiv.org/abs/2012.00736} {arXiv:2012.00736} \BibitemShut {NoStop}%
\bibitem [{\citenamefont {Miyadera}\ and\ \citenamefont
  {Takakura}(2023)}]{miyadera2023programming}%
  \BibitemOpen
  \bibfield  {author} {\bibinfo {author} {\bibfnamefont {T.}~\bibnamefont
  {Miyadera}}\ and\ \bibinfo {author} {\bibfnamefont {R.}~\bibnamefont
  {Takakura}},\ }\bibfield  {title} {\bibinfo {title} {Programming of channels
  in generalized probabilistic theories},\ }\href
  {https://doi.org/10.1063/5.0101198} {\bibfield  {journal} {\bibinfo
  {journal} {Journal of Mathematical Physics}\ }\textbf {\bibinfo {volume}
  {64}} (\bibinfo {year} {2023})},\ \Eprint {https://arxiv.org/abs/2205.08940}
  {arXiv:2205.08940} \BibitemShut {NoStop}%
\bibitem [{\citenamefont {Kim}\ \emph {et~al.}(2025)\citenamefont {Kim},
  \citenamefont {Chitambar},\ and\ \citenamefont {Leditzky}}]{kim2025resource}%
  \BibitemOpen
  \bibfield  {author} {\bibinfo {author} {\bibfnamefont {C.}~\bibnamefont
  {Kim}}, \bibinfo {author} {\bibfnamefont {E.}~\bibnamefont {Chitambar}},\
  and\ \bibinfo {author} {\bibfnamefont {F.}~\bibnamefont {Leditzky}},\
  }\bibfield  {title} {\bibinfo {title} {A resource theory of asynchronous
  quantum information processing},\ }\Eprint {https://arxiv.org/abs/2504.12945}
  {arXiv:2504.12945}  (\bibinfo {year} {2025})\BibitemShut {NoStop}%
\bibitem [{\citenamefont {Beigi}\ and\ \citenamefont
  {K{\"o}nig}(2011)}]{beigi2011simplified}%
  \BibitemOpen
  \bibfield  {author} {\bibinfo {author} {\bibfnamefont {S.}~\bibnamefont
  {Beigi}}\ and\ \bibinfo {author} {\bibfnamefont {R.}~\bibnamefont
  {K{\"o}nig}},\ }\bibfield  {title} {\bibinfo {title} {Simplified
  instantaneous non-local quantum computation with applications to
  position-based cryptography},\ }\href
  {https://doi.org/10.1088/1367-2630/13/9/093036} {\bibfield  {journal}
  {\bibinfo  {journal} {New Journal of Physics}\ }\textbf {\bibinfo {volume}
  {13}},\ \bibinfo {pages} {093036} (\bibinfo {year} {2011})},\ \Eprint
  {https://arxiv.org/abs/1101.1065} {arXiv:1101.1065} \BibitemShut {NoStop}%
\bibitem [{\citenamefont {Yang}\ and\ \citenamefont
  {Hayashi}(2021)}]{yang2021representation}%
  \BibitemOpen
  \bibfield  {author} {\bibinfo {author} {\bibfnamefont {Y.}~\bibnamefont
  {Yang}}\ and\ \bibinfo {author} {\bibfnamefont {M.}~\bibnamefont {Hayashi}},\
  }\bibfield  {title} {\bibinfo {title} {{Representation Matching For Remote
  Quantum Computing}},\ }\href {https://doi.org/10.1103/PRXQuantum.2.020327}
  {\bibfield  {journal} {\bibinfo  {journal} {PRX Quantum}\ }\textbf {\bibinfo
  {volume} {2}},\ \bibinfo {pages} {020327} (\bibinfo {year} {2021})},\ \Eprint
  {https://arxiv.org/abs/2009.06667} {arXiv:2009.06667} \BibitemShut {NoStop}%
\bibitem [{\citenamefont {Ac\'{\i}n}\ \emph {et~al.}(2001)\citenamefont
  {Ac\'{\i}n}, \citenamefont {Jan\'e},\ and\ \citenamefont
  {Vidal}}]{acin2001optimal}%
  \BibitemOpen
  \bibfield  {author} {\bibinfo {author} {\bibfnamefont {A.}~\bibnamefont
  {Ac\'{\i}n}}, \bibinfo {author} {\bibfnamefont {E.}~\bibnamefont {Jan\'e}},\
  and\ \bibinfo {author} {\bibfnamefont {G.}~\bibnamefont {Vidal}},\ }\bibfield
   {title} {\bibinfo {title} {Optimal estimation of quantum dynamics},\ }\href
  {https://doi.org/10.1103/PhysRevA.64.050302} {\bibfield  {journal} {\bibinfo
  {journal} {Phys. Rev. A}\ }\textbf {\bibinfo {volume} {64}},\ \bibinfo
  {pages} {050302(R)} (\bibinfo {year} {2001})},\ \Eprint
  {https://arxiv.org/abs/quant-ph/0012015} {arXiv:quant-ph/0012015}
  \BibitemShut {NoStop}%
\bibitem [{\citenamefont {Arora}\ and\ \citenamefont
  {Barak}(2009)}]{arora2009computational}%
  \BibitemOpen
  \bibfield  {author} {\bibinfo {author} {\bibfnamefont {S.}~\bibnamefont
  {Arora}}\ and\ \bibinfo {author} {\bibfnamefont {B.}~\bibnamefont {Barak}},\
  }\href {https://doi.org/10.1017/CBO9780511804090} {\emph {\bibinfo {title}
  {{Computational Complexity: A Modern Approach}}}}\ (\bibinfo  {publisher}
  {Cambridge University Press},\ \bibinfo {year} {2009})\BibitemShut {NoStop}%
\bibitem [{\citenamefont {Huang}\ \emph {et~al.}(2021)\citenamefont {Huang},
  \citenamefont {Kueng},\ and\ \citenamefont
  {Preskill}}]{huang2021information}%
  \BibitemOpen
  \bibfield  {author} {\bibinfo {author} {\bibfnamefont {H.-Y.}\ \bibnamefont
  {Huang}}, \bibinfo {author} {\bibfnamefont {R.}~\bibnamefont {Kueng}},\ and\
  \bibinfo {author} {\bibfnamefont {J.}~\bibnamefont {Preskill}},\ }\bibfield
  {title} {\bibinfo {title} {{Information-Theoretic Bounds on Quantum Advantage
  in Machine Learning}},\ }\href
  {https://doi.org/10.1103/PhysRevLett.126.190505} {\bibfield  {journal}
  {\bibinfo  {journal} {Physical Review Letters}\ }\textbf {\bibinfo {volume}
  {126}},\ \bibinfo {pages} {190505} (\bibinfo {year} {2021})},\ \Eprint
  {https://arxiv.org/abs/2101.02464} {arXiv:2101.02464} \BibitemShut {NoStop}%
\bibitem [{\citenamefont {Aharonov}\ \emph {et~al.}(2022)\citenamefont
  {Aharonov}, \citenamefont {Cotler},\ and\ \citenamefont
  {Qi}}]{aharonov2022quantum}%
  \BibitemOpen
  \bibfield  {author} {\bibinfo {author} {\bibfnamefont {D.}~\bibnamefont
  {Aharonov}}, \bibinfo {author} {\bibfnamefont {J.}~\bibnamefont {Cotler}},\
  and\ \bibinfo {author} {\bibfnamefont {X.-L.}\ \bibnamefont {Qi}},\
  }\bibfield  {title} {\bibinfo {title} {Quantum algorithmic measurement},\
  }\href {https://doi.org/10.1038/s41467-021-27922-0} {\bibfield  {journal}
  {\bibinfo  {journal} {Nature communications}\ }\textbf {\bibinfo {volume}
  {13}},\ \bibinfo {pages} {887} (\bibinfo {year} {2022})},\ \Eprint
  {https://arxiv.org/abs/2101.04634} {arXiv:2101.04634} \BibitemShut {NoStop}%
\bibitem [{\citenamefont {Huang}\ \emph {et~al.}(2022)\citenamefont {Huang},
  \citenamefont {Broughton}, \citenamefont {Cotler}, \citenamefont {Chen},
  \citenamefont {Li}, \citenamefont {Mohseni}, \citenamefont {Neven},
  \citenamefont {Babbush}, \citenamefont {Kueng}, \citenamefont {Preskill}
  \emph {et~al.}}]{huang2022quantum}%
  \BibitemOpen
  \bibfield  {author} {\bibinfo {author} {\bibfnamefont {H.-Y.}\ \bibnamefont
  {Huang}}, \bibinfo {author} {\bibfnamefont {M.}~\bibnamefont {Broughton}},
  \bibinfo {author} {\bibfnamefont {J.}~\bibnamefont {Cotler}}, \bibinfo
  {author} {\bibfnamefont {S.}~\bibnamefont {Chen}}, \bibinfo {author}
  {\bibfnamefont {J.}~\bibnamefont {Li}}, \bibinfo {author} {\bibfnamefont
  {M.}~\bibnamefont {Mohseni}}, \bibinfo {author} {\bibfnamefont
  {H.}~\bibnamefont {Neven}}, \bibinfo {author} {\bibfnamefont
  {R.}~\bibnamefont {Babbush}}, \bibinfo {author} {\bibfnamefont
  {R.}~\bibnamefont {Kueng}}, \bibinfo {author} {\bibfnamefont
  {J.}~\bibnamefont {Preskill}}, \emph {et~al.},\ }\bibfield  {title} {\bibinfo
  {title} {Quantum advantage in learning from experiments},\ }\href
  {https://doi.org/10.1126/science.abn7293} {\bibfield  {journal} {\bibinfo
  {journal} {Science}\ }\textbf {\bibinfo {volume} {376}},\ \bibinfo {pages}
  {1182} (\bibinfo {year} {2022})},\ \Eprint {https://arxiv.org/abs/2112.00778}
  {arXiv:2112.00778} \BibitemShut {NoStop}%
\bibitem [{\citenamefont {Bisio}\ \emph {et~al.}(2010)\citenamefont {Bisio},
  \citenamefont {Chiribella}, \citenamefont {D’Ariano}, \citenamefont
  {Facchini},\ and\ \citenamefont {Perinotti}}]{bisio2010optimal}%
  \BibitemOpen
  \bibfield  {author} {\bibinfo {author} {\bibfnamefont {A.}~\bibnamefont
  {Bisio}}, \bibinfo {author} {\bibfnamefont {G.}~\bibnamefont {Chiribella}},
  \bibinfo {author} {\bibfnamefont {G.~M.}\ \bibnamefont {D’Ariano}},
  \bibinfo {author} {\bibfnamefont {S.}~\bibnamefont {Facchini}},\ and\
  \bibinfo {author} {\bibfnamefont {P.}~\bibnamefont {Perinotti}},\ }\bibfield
  {title} {\bibinfo {title} {Optimal quantum learning of a unitary
  transformation},\ }\href {https://doi.org/10.1103/PhysRevA.81.032324}
  {\bibfield  {journal} {\bibinfo  {journal} {Physical Review A}\ }\textbf
  {\bibinfo {volume} {81}},\ \bibinfo {pages} {032324} (\bibinfo {year}
  {2010})},\ \Eprint {https://arxiv.org/abs/0903.0543} {arXiv:0903.0543}
  \BibitemShut {NoStop}%
\bibitem [{\citenamefont {Mo}\ and\ \citenamefont
  {Chiribella}(2019)}]{mo2019quantum}%
  \BibitemOpen
  \bibfield  {author} {\bibinfo {author} {\bibfnamefont {Y.}~\bibnamefont
  {Mo}}\ and\ \bibinfo {author} {\bibfnamefont {G.}~\bibnamefont
  {Chiribella}},\ }\bibfield  {title} {\bibinfo {title} {Quantum-enhanced
  learning of rotations about an unknown direction},\ }\href
  {https://doi.org/10.1088/1367-2630/ab4d9a} {\bibfield  {journal} {\bibinfo
  {journal} {New Journal of Physics}\ }\textbf {\bibinfo {volume} {21}},\
  \bibinfo {pages} {113003} (\bibinfo {year} {2019})},\ \Eprint
  {https://arxiv.org/abs/1906.01300} {arXiv:1906.01300} \BibitemShut {NoStop}%
\bibitem [{\citenamefont {Studzi{\'n}ski}\ \emph {et~al.}(2017)\citenamefont
  {Studzi{\'n}ski}, \citenamefont {Strelchuk}, \citenamefont {Mozrzymas},\ and\
  \citenamefont {Horodecki}}]{studzinski2018port}%
  \BibitemOpen
  \bibfield  {author} {\bibinfo {author} {\bibfnamefont {M.}~\bibnamefont
  {Studzi{\'n}ski}}, \bibinfo {author} {\bibfnamefont {S.}~\bibnamefont
  {Strelchuk}}, \bibinfo {author} {\bibfnamefont {M.}~\bibnamefont
  {Mozrzymas}},\ and\ \bibinfo {author} {\bibfnamefont {M.}~\bibnamefont
  {Horodecki}},\ }\bibfield  {title} {\bibinfo {title} {Port-based
  teleportation in arbitrary dimension},\ }\href
  {https://doi.org/10.1038/s41598-017-10051-4} {\bibfield  {journal} {\bibinfo
  {journal} {Scientific Reports}\ }\textbf {\bibinfo {volume} {7}},\ \bibinfo
  {pages} {10871} (\bibinfo {year} {2017})},\ \Eprint
  {https://arxiv.org/abs/1612.09260} {arXiv:1612.09260} \BibitemShut {NoStop}%
\bibitem [{\citenamefont {Kahn}(2007)}]{kahn2007fast}%
  \BibitemOpen
  \bibfield  {author} {\bibinfo {author} {\bibfnamefont {J.}~\bibnamefont
  {Kahn}},\ }\bibfield  {title} {\bibinfo {title} {{Fast rate estimation of a
  unitary operation in $\mathrm{SU}(d)$}},\ }\href
  {https://doi.org/10.1103/PhysRevA.75.022326} {\bibfield  {journal} {\bibinfo
  {journal} {Phys. Rev. A}\ }\textbf {\bibinfo {volume} {75}},\ \bibinfo
  {pages} {022326} (\bibinfo {year} {2007})},\ \Eprint
  {https://arxiv.org/abs/quant-ph/0603115} {arXiv:quant-ph/0603115}
  \BibitemShut {NoStop}%
\bibitem [{\citenamefont {Haah}\ \emph {et~al.}(2023)\citenamefont {Haah},
  \citenamefont {Kothari}, \citenamefont {O’Donnell},\ and\ \citenamefont
  {Tang}}]{haah2023query}%
  \BibitemOpen
  \bibfield  {author} {\bibinfo {author} {\bibfnamefont {J.}~\bibnamefont
  {Haah}}, \bibinfo {author} {\bibfnamefont {R.}~\bibnamefont {Kothari}},
  \bibinfo {author} {\bibfnamefont {R.}~\bibnamefont {O’Donnell}},\ and\
  \bibinfo {author} {\bibfnamefont {E.}~\bibnamefont {Tang}},\ }\bibfield
  {title} {\bibinfo {title} {Query-optimal estimation of unitary channels in
  diamond distance},\ }in\ \href {https://doi.org/10.1109/FOCS57990.2023.00028}
  {\emph {\bibinfo {booktitle} {2023 IEEE 64th Annual Symposium on Foundations
  of Computer Science (FOCS)}}}\ (\bibinfo {organization} {IEEE},\ \bibinfo
  {year} {2023})\ pp.\ \bibinfo {pages} {363--390},\ \Eprint
  {https://arxiv.org/abs/2302.14066} {arXiv:2302.14066} \BibitemShut {NoStop}%
\bibitem [{\citenamefont {Yoshida}\ \emph {et~al.}(2024)\citenamefont
  {Yoshida}, \citenamefont {Koizumi}, \citenamefont {Studzi{\'n}ski},
  \citenamefont {Quintino},\ and\ \citenamefont {Murao}}]{yoshida2024one}%
  \BibitemOpen
  \bibfield  {author} {\bibinfo {author} {\bibfnamefont {S.}~\bibnamefont
  {Yoshida}}, \bibinfo {author} {\bibfnamefont {Y.}~\bibnamefont {Koizumi}},
  \bibinfo {author} {\bibfnamefont {M.}~\bibnamefont {Studzi{\'n}ski}},
  \bibinfo {author} {\bibfnamefont {M.~T.}\ \bibnamefont {Quintino}},\ and\
  \bibinfo {author} {\bibfnamefont {M.}~\bibnamefont {Murao}},\ }\bibfield
  {title} {\bibinfo {title} {{One-to-one Correspondence between Deterministic
  Port-Based Teleportation and Unitary Estimation}},\ }\Eprint
  {https://arxiv.org/abs/2408.11902} {arXiv:2408.11902}  (\bibinfo {year}
  {2024})\BibitemShut {NoStop}%
\bibitem [{\citenamefont {Yoshida}\ \emph
  {et~al.}(2025{\natexlab{a}})\citenamefont {Yoshida}, \citenamefont
  {Yoshida},\ and\ \citenamefont {Murao}}]{yoshida2025asymptotically}%
  \BibitemOpen
  \bibfield  {author} {\bibinfo {author} {\bibfnamefont {S.}~\bibnamefont
  {Yoshida}}, \bibinfo {author} {\bibfnamefont {H.}~\bibnamefont {Yoshida}},\
  and\ \bibinfo {author} {\bibfnamefont {M.}~\bibnamefont {Murao}},\ }\bibfield
   {title} {\bibinfo {title} {{Asymptotically optimal unitary estimation in
  $\mathrm{SU}(3)$ by the analysis of graph Laplacian}},\ }\Eprint
  {https://arxiv.org/abs/2509.20608} {arXiv:2509.20608}  (\bibinfo {year}
  {2025}{\natexlab{a}})\BibitemShut {NoStop}%
\bibitem [{\citenamefont {Nielsen}\ and\ \citenamefont
  {Chuang}(2010)}]{nielsen2010quantum}%
  \BibitemOpen
  \bibfield  {author} {\bibinfo {author} {\bibfnamefont {M.~A.}\ \bibnamefont
  {Nielsen}}\ and\ \bibinfo {author} {\bibfnamefont {I.}~\bibnamefont
  {Chuang}},\ }\href {https://doi.org/10.1017/CBO9780511976667} {\emph
  {\bibinfo {title} {{Quantum Computation and Quantum Information}}}}\
  (\bibinfo  {publisher} {Cambridge University Press},\ \bibinfo {year}
  {2010})\BibitemShut {NoStop}%
\bibitem [{\citenamefont {Grover}(1996)}]{grover1996fast}%
  \BibitemOpen
  \bibfield  {author} {\bibinfo {author} {\bibfnamefont {L.~K.}\ \bibnamefont
  {Grover}},\ }\bibfield  {title} {\bibinfo {title} {A fast quantum mechanical
  algorithm for database search},\ }in\ \href
  {https://doi.org/10.1145/237814.237866} {\emph {\bibinfo {booktitle}
  {Proceedings of the Twenty-Eighth Annual ACM Symposium on Theory of
  Computing}}},\ \bibinfo {series and number} {STOC '96}\ (\bibinfo
  {publisher} {Association for Computing Machinery},\ \bibinfo {address} {New
  York, NY, USA},\ \bibinfo {year} {1996})\ pp.\ \bibinfo {pages} {212--219},\
  \Eprint {https://arxiv.org/abs/quant-ph/9605043} {arXiv:quant-ph/9605043}
  \BibitemShut {NoStop}%
\bibitem [{\citenamefont {Harrow}\ \emph {et~al.}(2009)\citenamefont {Harrow},
  \citenamefont {Hassidim},\ and\ \citenamefont {Lloyd}}]{harrow2009quantum}%
  \BibitemOpen
  \bibfield  {author} {\bibinfo {author} {\bibfnamefont {A.~W.}\ \bibnamefont
  {Harrow}}, \bibinfo {author} {\bibfnamefont {A.}~\bibnamefont {Hassidim}},\
  and\ \bibinfo {author} {\bibfnamefont {S.}~\bibnamefont {Lloyd}},\ }\bibfield
   {title} {\bibinfo {title} {{Quantum Algorithm for Linear Systems of
  Equations}},\ }\href {https://doi.org/10.1103/PhysRevLett.103.150502}
  {\bibfield  {journal} {\bibinfo  {journal} {Phys. Rev. Lett.}\ }\textbf
  {\bibinfo {volume} {103}},\ \bibinfo {pages} {150502} (\bibinfo {year}
  {2009})},\ \Eprint {https://arxiv.org/abs/0811.3171} {arXiv:0811.3171}
  \BibitemShut {NoStop}%
\bibitem [{\citenamefont {Wilde}(2013)}]{wilde2013quantum}%
  \BibitemOpen
  \bibfield  {author} {\bibinfo {author} {\bibfnamefont {M.~M.}\ \bibnamefont
  {Wilde}},\ }\href {https://doi.org/10.1017/CBO9781139525343} {\emph {\bibinfo
  {title} {Quantum Information Theory}}}\ (\bibinfo  {publisher} {Cambridge
  University Press},\ \bibinfo {address} {Cambridge},\ \bibinfo {year} {2013})\
  \Eprint {https://arxiv.org/abs/1106.1445} {arXiv:1106.1445} \BibitemShut
  {NoStop}%
\bibitem [{\citenamefont {Watrous}(2018)}]{watrous2018theory}%
  \BibitemOpen
  \bibfield  {author} {\bibinfo {author} {\bibfnamefont {J.}~\bibnamefont
  {Watrous}},\ }\href {https://doi.org/10.1017/9781316848142} {\emph {\bibinfo
  {title} {The Theory of Quantum Information}}}\ (\bibinfo  {publisher}
  {Cambridge University Press},\ \bibinfo {address} {Cambridge, UK},\ \bibinfo
  {year} {2018})\BibitemShut {NoStop}%
\bibitem [{\citenamefont {Chen}\ \emph {et~al.}(2024)\citenamefont {Chen},
  \citenamefont {Wang},\ and\ \citenamefont {Zhang}}]{chen2024local}%
  \BibitemOpen
  \bibfield  {author} {\bibinfo {author} {\bibfnamefont {K.}~\bibnamefont
  {Chen}}, \bibinfo {author} {\bibfnamefont {Q.}~\bibnamefont {Wang}},\ and\
  \bibinfo {author} {\bibfnamefont {Z.}~\bibnamefont {Zhang}},\ }\bibfield
  {title} {\bibinfo {title} {Local test for unitarily invariant properties of
  bipartite quantum states},\ }\Eprint {https://arxiv.org/abs/2404.04599}
  {arXiv:2404.04599}  (\bibinfo {year} {2024})\BibitemShut {NoStop}%
\bibitem [{\citenamefont {Tang}\ \emph {et~al.}(2025)\citenamefont {Tang},
  \citenamefont {Wright},\ and\ \citenamefont {Zhandry}}]{tang2025conjugate}%
  \BibitemOpen
  \bibfield  {author} {\bibinfo {author} {\bibfnamefont {E.}~\bibnamefont
  {Tang}}, \bibinfo {author} {\bibfnamefont {J.}~\bibnamefont {Wright}},\ and\
  \bibinfo {author} {\bibfnamefont {M.}~\bibnamefont {Zhandry}},\ }\bibfield
  {title} {\bibinfo {title} {Conjugate queries can help},\ }\Eprint
  {https://arxiv.org/abs/2510.07622} {arXiv:2510.07622}  (\bibinfo {year}
  {2025})\BibitemShut {NoStop}%
\bibitem [{\citenamefont {Pelecanos}\ \emph {et~al.}(2025)\citenamefont
  {Pelecanos}, \citenamefont {Spilecki}, \citenamefont {Tang},\ and\
  \citenamefont {Wright}}]{pelecanos2025mixed}%
  \BibitemOpen
  \bibfield  {author} {\bibinfo {author} {\bibfnamefont {A.}~\bibnamefont
  {Pelecanos}}, \bibinfo {author} {\bibfnamefont {J.}~\bibnamefont {Spilecki}},
  \bibinfo {author} {\bibfnamefont {E.}~\bibnamefont {Tang}},\ and\ \bibinfo
  {author} {\bibfnamefont {J.}~\bibnamefont {Wright}},\ }\bibfield  {title}
  {\bibinfo {title} {Mixed state tomography reduces to pure state tomography},\
  }\Eprint {https://arxiv.org/abs/2511.15806} {arXiv:2511.15806}  (\bibinfo
  {year} {2025})\BibitemShut {NoStop}%
\bibitem [{\citenamefont {Mele}\ and\ \citenamefont
  {Bittel}(2025)}]{mele2025optimal}%
  \BibitemOpen
  \bibfield  {author} {\bibinfo {author} {\bibfnamefont {A.~A.}\ \bibnamefont
  {Mele}}\ and\ \bibinfo {author} {\bibfnamefont {L.}~\bibnamefont {Bittel}},\
  }\bibfield  {title} {\bibinfo {title} {Optimal learning of quantum channels
  in diamond distance},\ }\Eprint {https://arxiv.org/abs/2512.10214}
  {arXiv:2512.10214}  (\bibinfo {year} {2025})\BibitemShut {NoStop}%
\bibitem [{\citenamefont {Girardi}\ \emph
  {et~al.}(2025{\natexlab{a}})\citenamefont {Girardi}, \citenamefont {Mele},\
  and\ \citenamefont {Lami}}]{girardi2025random}%
  \BibitemOpen
  \bibfield  {author} {\bibinfo {author} {\bibfnamefont {F.}~\bibnamefont
  {Girardi}}, \bibinfo {author} {\bibfnamefont {F.~A.}\ \bibnamefont {Mele}},\
  and\ \bibinfo {author} {\bibfnamefont {L.}~\bibnamefont {Lami}},\ }\bibfield
  {title} {\bibinfo {title} {Random purification channel made simple},\
  }\Eprint {https://arxiv.org/abs/2511.23451} {arXiv:2511.23451}  (\bibinfo
  {year} {2025}{\natexlab{a}})\BibitemShut {NoStop}%
\bibitem [{\citenamefont {Walter}\ and\ \citenamefont
  {Witteveen}(2025)}]{walter2025random}%
  \BibitemOpen
  \bibfield  {author} {\bibinfo {author} {\bibfnamefont {M.}~\bibnamefont
  {Walter}}\ and\ \bibinfo {author} {\bibfnamefont {F.}~\bibnamefont
  {Witteveen}},\ }\bibfield  {title} {\bibinfo {title} {A random purification
  channel for arbitrary symmetries with applications to fermions and bosons},\
  }\Eprint {https://arxiv.org/abs/2512.15690} {arXiv:2512.15690}  (\bibinfo
  {year} {2025})\BibitemShut {NoStop}%
\bibitem [{\citenamefont {Mele}\ \emph {et~al.}(2025)\citenamefont {Mele},
  \citenamefont {Girardi}, \citenamefont {Chen}, \citenamefont {Fanizza},\ and\
  \citenamefont {Lami}}]{mele2025random}%
  \BibitemOpen
  \bibfield  {author} {\bibinfo {author} {\bibfnamefont {F.~A.}\ \bibnamefont
  {Mele}}, \bibinfo {author} {\bibfnamefont {F.}~\bibnamefont {Girardi}},
  \bibinfo {author} {\bibfnamefont {S.}~\bibnamefont {Chen}}, \bibinfo {author}
  {\bibfnamefont {M.}~\bibnamefont {Fanizza}},\ and\ \bibinfo {author}
  {\bibfnamefont {L.}~\bibnamefont {Lami}},\ }\bibfield  {title} {\bibinfo
  {title} {{Random purification channel for passive Gaussian bosons}},\
  }\Eprint {https://arxiv.org/abs/2512.16878} {arXiv:2512.16878}  (\bibinfo
  {year} {2025})\BibitemShut {NoStop}%
\bibitem [{\citenamefont {Chen}\ \emph {et~al.}(2025)\citenamefont {Chen},
  \citenamefont {Yu},\ and\ \citenamefont {Zhang}}]{chen2025quantum}%
  \BibitemOpen
  \bibfield  {author} {\bibinfo {author} {\bibfnamefont {K.}~\bibnamefont
  {Chen}}, \bibinfo {author} {\bibfnamefont {N.}~\bibnamefont {Yu}},\ and\
  \bibinfo {author} {\bibfnamefont {Z.}~\bibnamefont {Zhang}},\ }\bibfield
  {title} {\bibinfo {title} {Quantum channel tomography and estimation by local
  test},\ }\Eprint {https://arxiv.org/abs/2512.13614} {arXiv:2512.13614}
  (\bibinfo {year} {2025})\BibitemShut {NoStop}%
\bibitem [{\citenamefont {Girardi}\ \emph
  {et~al.}(2025{\natexlab{b}})\citenamefont {Girardi}, \citenamefont {Mele},
  \citenamefont {Zhao}, \citenamefont {Fanizza},\ and\ \citenamefont
  {Lami}}]{girardi2025random2}%
  \BibitemOpen
  \bibfield  {author} {\bibinfo {author} {\bibfnamefont {F.}~\bibnamefont
  {Girardi}}, \bibinfo {author} {\bibfnamefont {F.~A.}\ \bibnamefont {Mele}},
  \bibinfo {author} {\bibfnamefont {H.}~\bibnamefont {Zhao}}, \bibinfo {author}
  {\bibfnamefont {M.}~\bibnamefont {Fanizza}},\ and\ \bibinfo {author}
  {\bibfnamefont {L.}~\bibnamefont {Lami}},\ }\bibfield  {title} {\bibinfo
  {title} {{Random Stinespring superchannel: converting channel queries into
  dilation isometry queries}},\ }\Eprint {https://arxiv.org/abs/2512.20599}
  {arXiv:2512.20599}  (\bibinfo {year} {2025}{\natexlab{b}})\BibitemShut
  {NoStop}%
\bibitem [{\citenamefont {Yoshida}\ \emph
  {et~al.}(2025{\natexlab{b}})\citenamefont {Yoshida}, \citenamefont {Niwa},\
  and\ \citenamefont {Murao}}]{yoshida2025random}%
  \BibitemOpen
  \bibfield  {author} {\bibinfo {author} {\bibfnamefont {S.}~\bibnamefont
  {Yoshida}}, \bibinfo {author} {\bibfnamefont {R.}~\bibnamefont {Niwa}},\ and\
  \bibinfo {author} {\bibfnamefont {M.}~\bibnamefont {Murao}},\ }\bibfield
  {title} {\bibinfo {title} {Random dilation superchannel},\ }\Eprint
  {https://arxiv.org/abs/2512.21260} {arXiv:2512.21260}  (\bibinfo {year}
  {2025}{\natexlab{b}})\BibitemShut {NoStop}%
\bibitem [{\citenamefont {Bru{\ss}}\ and\ \citenamefont
  {Macchiavello}(1999)}]{bruss1999optimal}%
  \BibitemOpen
  \bibfield  {author} {\bibinfo {author} {\bibfnamefont {D.}~\bibnamefont
  {Bru{\ss}}}\ and\ \bibinfo {author} {\bibfnamefont {C.}~\bibnamefont
  {Macchiavello}},\ }\bibfield  {title} {\bibinfo {title} {Optimal state
  estimation for d-dimensional quantum systems},\ }\href
  {https://doi.org/10.1016/S0375-9601%2899%2900099-7} {\bibfield  {journal}
  {\bibinfo  {journal} {Physics Letters A}\ }\textbf {\bibinfo {volume}
  {253}},\ \bibinfo {pages} {249} (\bibinfo {year} {1999})},\ \Eprint
  {https://arxiv.org/abs/quant-ph/9812016} {arXiv:quant-ph/9812016}
  \BibitemShut {NoStop}%
\bibitem [{\citenamefont {Mele}(2024)}]{mele2024introduction}%
  \BibitemOpen
  \bibfield  {author} {\bibinfo {author} {\bibfnamefont {A.~A.}\ \bibnamefont
  {Mele}},\ }\bibfield  {title} {\bibinfo {title} {{Introduction to Haar
  Measure Tools in Quantum Information: A Beginner's Tutorial}},\ }\href
  {https://doi.org/10.22331/q-2024-05-08-1340} {\bibfield  {journal} {\bibinfo
  {journal} {Quantum}\ }\textbf {\bibinfo {volume} {8}},\ \bibinfo {pages}
  {1340} (\bibinfo {year} {2024})},\ \Eprint {https://arxiv.org/abs/2307.08956}
  {arXiv:2307.08956} \BibitemShut {NoStop}%
\bibitem [{\citenamefont {Raginsky}(2001)}]{raginsky2001fidelity}%
  \BibitemOpen
  \bibfield  {author} {\bibinfo {author} {\bibfnamefont {M.}~\bibnamefont
  {Raginsky}},\ }\bibfield  {title} {\bibinfo {title} {A fidelity measure for
  quantum channels},\ }\href
  {https://doi.org/10.1016/S0375-9601%2801%2900640-5} {\bibfield  {journal}
  {\bibinfo  {journal} {Physics Letters A}\ }\textbf {\bibinfo {volume}
  {290}},\ \bibinfo {pages} {11} (\bibinfo {year} {2001})},\ \Eprint
  {https://arxiv.org/abs/quant-ph/0107108} {arXiv:quant-ph/0107108}
  \BibitemShut {NoStop}%
\bibitem [{\citenamefont {Holevo}(2011)}]{holevo2011probabilistic}%
  \BibitemOpen
  \bibfield  {author} {\bibinfo {author} {\bibfnamefont {A.~S.}\ \bibnamefont
  {Holevo}},\ }\href {https://doi.org/10.1007/978-88-7642-378-9} {\emph
  {\bibinfo {title} {Probabilistic and statistical aspects of quantum
  theory}}},\ Vol.~\bibinfo {volume} {1}\ (\bibinfo  {publisher} {Springer
  Science \& Business Media},\ \bibinfo {year} {2011})\BibitemShut {NoStop}%
\bibitem [{\citenamefont {Chiribella}\ \emph {et~al.}(2005)\citenamefont
  {Chiribella}, \citenamefont {D'Ariano},\ and\ \citenamefont
  {Sacchi}}]{chiribella2005optimal}%
  \BibitemOpen
  \bibfield  {author} {\bibinfo {author} {\bibfnamefont {G.}~\bibnamefont
  {Chiribella}}, \bibinfo {author} {\bibfnamefont {G.~M.}\ \bibnamefont
  {D'Ariano}},\ and\ \bibinfo {author} {\bibfnamefont {M.~F.}\ \bibnamefont
  {Sacchi}},\ }\bibfield  {title} {\bibinfo {title} {Optimal estimation of
  group transformations using entanglement},\ }\href
  {https://doi.org/10.1103/PhysRevA.72.042338} {\bibfield  {journal} {\bibinfo
  {journal} {Phys. Rev. A}\ }\textbf {\bibinfo {volume} {72}},\ \bibinfo
  {pages} {042338} (\bibinfo {year} {2005})},\ \Eprint
  {https://arxiv.org/abs/quant-ph/0506267} {arXiv:quant-ph/0506267}
  \BibitemShut {NoStop}%
\bibitem [{\citenamefont {Zyczkowski}\ and\ \citenamefont
  {Sommers}(2000)}]{zyczkowski2000truncations}%
  \BibitemOpen
  \bibfield  {author} {\bibinfo {author} {\bibfnamefont {K.}~\bibnamefont
  {Zyczkowski}}\ and\ \bibinfo {author} {\bibfnamefont {H.-J.}\ \bibnamefont
  {Sommers}},\ }\bibfield  {title} {\bibinfo {title} {Truncations of random
  unitary matrices},\ }\href {https://doi.org/10.1088/0305-4470/33/10/307}
  {\bibfield  {journal} {\bibinfo  {journal} {Journal of Physics A:
  Mathematical and General}\ }\textbf {\bibinfo {volume} {33}},\ \bibinfo
  {pages} {2045} (\bibinfo {year} {2000})},\ \Eprint
  {https://arxiv.org/abs/chao-dyn/9910032} {arXiv:chao-dyn/9910032}
  \BibitemShut {NoStop}%
\bibitem [{\citenamefont {Kukulski}\ \emph {et~al.}(2021)\citenamefont
  {Kukulski}, \citenamefont {Nechita}, \citenamefont {Pawela}, \citenamefont
  {Pucha{\l}a},\ and\ \citenamefont {{\.Z}yczkowski}}]{kukulski2021generating}%
  \BibitemOpen
  \bibfield  {author} {\bibinfo {author} {\bibfnamefont {R.}~\bibnamefont
  {Kukulski}}, \bibinfo {author} {\bibfnamefont {I.}~\bibnamefont {Nechita}},
  \bibinfo {author} {\bibfnamefont {{\L}.}~\bibnamefont {Pawela}}, \bibinfo
  {author} {\bibfnamefont {Z.}~\bibnamefont {Pucha{\l}a}},\ and\ \bibinfo
  {author} {\bibfnamefont {K.}~\bibnamefont {{\.Z}yczkowski}},\ }\bibfield
  {title} {\bibinfo {title} {Generating random quantum channels},\ }\href
  {https://doi.org/10.1063/5.0038838} {\bibfield  {journal} {\bibinfo
  {journal} {Journal of Mathematical Physics}\ }\textbf {\bibinfo {volume}
  {62}} (\bibinfo {year} {2021})},\ \Eprint {https://arxiv.org/abs/2011.02994}
  {arXiv:2011.02994} \BibitemShut {NoStop}%
\bibitem [{sup()}]{supple}%
  \BibitemOpen
  \href@noop {} {}\bibinfo {note} {See the Supplemental Material for the detail
  of the proof, which includes Refs.~\cite{fulton1997young, georgi2000lie,
  ceccherini2010representation, itzykson1966unitary, taranto2025higher,
  chiribella2008quantum, wechs2021quantum, hardy2007towards,
  oreshkov2012quantum, chiribella2013quantum, matsumoto2012input,
  fuchs2002cryptographic, horn2012matrix, degroot2012probability,
  chiribella2009optimal, chiribella2008memory,
  bavaresco2022unitary}.}\BibitemShut {Stop}%
\bibitem [{\citenamefont {Kitaev}(1997)}]{kitaev1997quantum}%
  \BibitemOpen
  \bibfield  {author} {\bibinfo {author} {\bibfnamefont {A.~Y.}\ \bibnamefont
  {Kitaev}},\ }\bibfield  {title} {\bibinfo {title} {Quantum computations:
  algorithms and error correction},\ }\href
  {https://doi.org/10.1070/RM1997v052n06ABEH002155} {\bibfield  {journal}
  {\bibinfo  {journal} {Russian Mathematical Surveys}\ }\textbf {\bibinfo
  {volume} {52}},\ \bibinfo {pages} {1191} (\bibinfo {year}
  {1997})}\BibitemShut {NoStop}%
\bibitem [{\citenamefont {Watrous}(2005)}]{watrous2005notes}%
  \BibitemOpen
  \bibfield  {author} {\bibinfo {author} {\bibfnamefont {J.}~\bibnamefont
  {Watrous}},\ }\bibfield  {title} {\bibinfo {title} {{Notes on super-operator
  norms induced by Schatten norms}},\ }\href
  {https://doi.org/10.26421/QIC5.1-6} {\bibfield  {journal} {\bibinfo
  {journal} {Quantum Info. Comput.}\ }\textbf {\bibinfo {volume} {5}},\
  \bibinfo {pages} {58–68} (\bibinfo {year} {2005})},\ \Eprint
  {https://arxiv.org/abs/quant-ph/0411077} {arXiv:quant-ph/0411077}
  \BibitemShut {NoStop}%
\bibitem [{\citenamefont {Helstrom}(1969)}]{helstrom1969quantum}%
  \BibitemOpen
  \bibfield  {author} {\bibinfo {author} {\bibfnamefont {C.~W.}\ \bibnamefont
  {Helstrom}},\ }\bibfield  {title} {\bibinfo {title} {Quantum detection and
  estimation theory},\ }\href {https://doi.org/10.1007/BF01007479} {\bibfield
  {journal} {\bibinfo  {journal} {Journal of Statistical Physics}\ }\textbf
  {\bibinfo {volume} {1}},\ \bibinfo {pages} {231} (\bibinfo {year}
  {1969})}\BibitemShut {NoStop}%
\bibitem [{\citenamefont {Zhou}\ and\ \citenamefont
  {Jiang}(2021)}]{zhou2021asymptotic}%
  \BibitemOpen
  \bibfield  {author} {\bibinfo {author} {\bibfnamefont {S.}~\bibnamefont
  {Zhou}}\ and\ \bibinfo {author} {\bibfnamefont {L.}~\bibnamefont {Jiang}},\
  }\bibfield  {title} {\bibinfo {title} {Asymptotic theory of quantum channel
  estimation},\ }\href {https://doi.org/10.1103/PRXQuantum.2.010343} {\bibfield
   {journal} {\bibinfo  {journal} {PRX Quantum}\ }\textbf {\bibinfo {volume}
  {2}},\ \bibinfo {pages} {010343} (\bibinfo {year} {2021})},\ \Eprint
  {https://arxiv.org/abs/2003.10559} {arXiv:2003.10559} \BibitemShut {NoStop}%
\bibitem [{\citenamefont {Van~Trees}(2001)}]{van2001detection}%
  \BibitemOpen
  \bibfield  {author} {\bibinfo {author} {\bibfnamefont {H.~L.}\ \bibnamefont
  {Van~Trees}},\ }\href {https://doi.org/10.1002/0471221082} {\emph {\bibinfo
  {title} {{Detection, Estimation, and Modulation Theory, Part I: Detection,
  Estimation, and Linear Modulation Theory}}}}\ (\bibinfo  {publisher} {John
  Wiley \& Sons},\ \bibinfo {year} {2001})\BibitemShut {NoStop}%
\bibitem [{\citenamefont {Bagan}\ \emph {et~al.}(2004)\citenamefont {Bagan},
  \citenamefont {Baig},\ and\ \citenamefont {{{R.
  Mu\~noz-Tapia}}}}]{bagan2004entanglement}%
  \BibitemOpen
  \bibfield  {author} {\bibinfo {author} {\bibfnamefont {E.}~\bibnamefont
  {Bagan}}, \bibinfo {author} {\bibfnamefont {M.}~\bibnamefont {Baig}},\ and\
  \bibinfo {author} {\bibnamefont {{{R. Mu\~noz-Tapia}}}},\ }\bibfield  {title}
  {\bibinfo {title} {Entanglement-assisted alignment of reference frames using
  a dense covariant coding},\ }\href
  {https://doi.org/10.1103/PhysRevA.69.050303} {\bibfield  {journal} {\bibinfo
  {journal} {Phys. Rev. A}\ }\textbf {\bibinfo {volume} {69}},\ \bibinfo
  {pages} {050303(R)} (\bibinfo {year} {2004})},\ \Eprint
  {https://arxiv.org/abs/quant-ph/0303019} {arXiv:quant-ph/0303019}
  \BibitemShut {NoStop}%
\bibitem [{\citenamefont {Choi}(1975)}]{choi1975completely}%
  \BibitemOpen
  \bibfield  {author} {\bibinfo {author} {\bibfnamefont {M.-D.}\ \bibnamefont
  {Choi}},\ }\bibfield  {title} {\bibinfo {title} {Completely positive linear
  maps on complex matrices},\ }\href
  {https://doi.org/10.1016/0024-3795(75)90075-0} {\bibfield  {journal}
  {\bibinfo  {journal} {Linear algebra and its applications}\ }\textbf
  {\bibinfo {volume} {10}},\ \bibinfo {pages} {285} (\bibinfo {year}
  {1975})}\BibitemShut {NoStop}%
\bibitem [{\citenamefont {Stinespring}(1955)}]{stinespring1955positive}%
  \BibitemOpen
  \bibfield  {author} {\bibinfo {author} {\bibfnamefont {W.~F.}\ \bibnamefont
  {Stinespring}},\ }\bibfield  {title} {\bibinfo {title} {{Positive functions
  on $C^*$-algebras}},\ }\href {https://doi.org/10.2307/2032342} {\bibfield
  {journal} {\bibinfo  {journal} {Proceedings of the American Mathematical
  Society}\ }\textbf {\bibinfo {volume} {6}},\ \bibinfo {pages} {211} (\bibinfo
  {year} {1955})}\BibitemShut {NoStop}%
\bibitem [{\citenamefont {Yoshida}\ \emph {et~al.}(2023)\citenamefont
  {Yoshida}, \citenamefont {Soeda},\ and\ \citenamefont
  {Murao}}]{yoshida2023universal}%
  \BibitemOpen
  \bibfield  {author} {\bibinfo {author} {\bibfnamefont {S.}~\bibnamefont
  {Yoshida}}, \bibinfo {author} {\bibfnamefont {A.}~\bibnamefont {Soeda}},\
  and\ \bibinfo {author} {\bibfnamefont {M.}~\bibnamefont {Murao}},\ }\bibfield
   {title} {\bibinfo {title} {Universal construction of decoders from encoding
  black boxes},\ }\href {https://doi.org/10.22331/q-2023-03-20-957} {\bibfield
  {journal} {\bibinfo  {journal} {Quantum}\ }\textbf {\bibinfo {volume} {7}},\
  \bibinfo {pages} {957} (\bibinfo {year} {2023})},\ \Eprint
  {https://arxiv.org/abs/2110.00258} {arXiv:2110.00258} \BibitemShut {NoStop}%
\bibitem [{\citenamefont {Yoshida}\ \emph
  {et~al.}(2025{\natexlab{c}})\citenamefont {Yoshida}, \citenamefont {Soeda},\
  and\ \citenamefont {Murao}}]{yoshida2025universal}%
  \BibitemOpen
  \bibfield  {author} {\bibinfo {author} {\bibfnamefont {S.}~\bibnamefont
  {Yoshida}}, \bibinfo {author} {\bibfnamefont {A.}~\bibnamefont {Soeda}},\
  and\ \bibinfo {author} {\bibfnamefont {M.}~\bibnamefont {Murao}},\ }\bibfield
   {title} {\bibinfo {title} {Universal adjointation of isometry operations
  using conversion of quantum supermaps},\ }\href
  {https://doi.org/10.22331/q-2025-05-20-1750} {\bibfield  {journal} {\bibinfo
  {journal} {Quantum}\ }\textbf {\bibinfo {volume} {9}},\ \bibinfo {pages}
  {1750} (\bibinfo {year} {2025}{\natexlab{c}})},\ \Eprint
  {https://arxiv.org/abs/2401.10137} {arXiv:2401.10137} \BibitemShut {NoStop}%
\bibitem [{\citenamefont {Strelchuk}\ \emph {et~al.}(2013)\citenamefont
  {Strelchuk}, \citenamefont {Horodecki},\ and\ \citenamefont
  {Oppenheim}}]{sergii2013generalized}%
  \BibitemOpen
  \bibfield  {author} {\bibinfo {author} {\bibfnamefont {S.}~\bibnamefont
  {Strelchuk}}, \bibinfo {author} {\bibfnamefont {M.}~\bibnamefont
  {Horodecki}},\ and\ \bibinfo {author} {\bibfnamefont {J.}~\bibnamefont
  {Oppenheim}},\ }\bibfield  {title} {\bibinfo {title} {{Generalized
  Teleportation and Entanglement Recycling}},\ }\href
  {https://doi.org/10.1103/PhysRevLett.110.010505} {\bibfield  {journal}
  {\bibinfo  {journal} {Phys. Rev. Lett.}\ }\textbf {\bibinfo {volume} {110}},\
  \bibinfo {pages} {010505} (\bibinfo {year} {2013})},\ \Eprint
  {https://arxiv.org/abs/1209.2683} {arXiv:1209.2683} \BibitemShut {NoStop}%
\bibitem [{\citenamefont {Mozrzymas}\ \emph {et~al.}(2021)\citenamefont
  {Mozrzymas}, \citenamefont {Studzi{\'n}ski},\ and\ \citenamefont
  {Kopszak}}]{mozrzymas2021optimal}%
  \BibitemOpen
  \bibfield  {author} {\bibinfo {author} {\bibfnamefont {M.}~\bibnamefont
  {Mozrzymas}}, \bibinfo {author} {\bibfnamefont {M.}~\bibnamefont
  {Studzi{\'n}ski}},\ and\ \bibinfo {author} {\bibfnamefont {P.}~\bibnamefont
  {Kopszak}},\ }\bibfield  {title} {\bibinfo {title} {{Optimal Multi-port-based
  Teleportation Schemes}},\ }\href {https://doi.org/10.22331/q-2021-06-17-477}
  {\bibfield  {journal} {\bibinfo  {journal} {Quantum}\ }\textbf {\bibinfo
  {volume} {5}},\ \bibinfo {pages} {477} (\bibinfo {year} {2021})},\ \Eprint
  {https://arxiv.org/abs/2011.09256} {arXiv:2011.09256} \BibitemShut {NoStop}%
\bibitem [{\citenamefont {Kopszak}\ \emph {et~al.}(2021)\citenamefont
  {Kopszak}, \citenamefont {Mozrzymas}, \citenamefont {Studzi{\'n}ski},\ and\
  \citenamefont {Horodecki}}]{kopszak2021multiport}%
  \BibitemOpen
  \bibfield  {author} {\bibinfo {author} {\bibfnamefont {P.}~\bibnamefont
  {Kopszak}}, \bibinfo {author} {\bibfnamefont {M.}~\bibnamefont {Mozrzymas}},
  \bibinfo {author} {\bibfnamefont {M.}~\bibnamefont {Studzi{\'n}ski}},\ and\
  \bibinfo {author} {\bibfnamefont {M.}~\bibnamefont {Horodecki}},\ }\bibfield
  {title} {\bibinfo {title} {Multiport based teleportation--transmission of a
  large amount of quantum information},\ }\href
  {https://doi.org/10.22331/q-2021-11-11-576} {\bibfield  {journal} {\bibinfo
  {journal} {Quantum}\ }\textbf {\bibinfo {volume} {5}},\ \bibinfo {pages}
  {576} (\bibinfo {year} {2021})},\ \Eprint {https://arxiv.org/abs/2008.00856}
  {arXiv:2008.00856} \BibitemShut {NoStop}%
\bibitem [{\citenamefont {Studziński}\ \emph {et~al.}(2022)\citenamefont
  {Studziński}, \citenamefont {Mozrzymas}, \citenamefont {Kopszak},\ and\
  \citenamefont {Horodecki}}]{studzinski2022efficient}%
  \BibitemOpen
  \bibfield  {author} {\bibinfo {author} {\bibfnamefont {M.}~\bibnamefont
  {Studziński}}, \bibinfo {author} {\bibfnamefont {M.}~\bibnamefont
  {Mozrzymas}}, \bibinfo {author} {\bibfnamefont {P.}~\bibnamefont {Kopszak}},\
  and\ \bibinfo {author} {\bibfnamefont {M.}~\bibnamefont {Horodecki}},\
  }\bibfield  {title} {\bibinfo {title} {{Efficient Multi Port-Based
  Teleportation Schemes}},\ }\href {https://doi.org/10.1109/TIT.2022.3187852}
  {\bibfield  {journal} {\bibinfo  {journal} {IEEE Transactions on Information
  Theory}\ }\textbf {\bibinfo {volume} {68}},\ \bibinfo {pages} {7892}
  (\bibinfo {year} {2022})},\ \Eprint {https://arxiv.org/abs/2008.00984}
  {arXiv:2008.00984} \BibitemShut {NoStop}%
\bibitem [{\citenamefont {Chiribella}\ \emph
  {et~al.}(2008{\natexlab{a}})\citenamefont {Chiribella}, \citenamefont
  {D'Ariano},\ and\ \citenamefont {Perinotti}}]{chiribella2008optimal}%
  \BibitemOpen
  \bibfield  {author} {\bibinfo {author} {\bibfnamefont {G.}~\bibnamefont
  {Chiribella}}, \bibinfo {author} {\bibfnamefont {G.~M.}\ \bibnamefont
  {D'Ariano}},\ and\ \bibinfo {author} {\bibfnamefont {P.}~\bibnamefont
  {Perinotti}},\ }\bibfield  {title} {\bibinfo {title} {{Optimal Cloning of
  Unitary Transformation}},\ }\href
  {https://doi.org/10.1103/PhysRevLett.101.180504} {\bibfield  {journal}
  {\bibinfo  {journal} {Phys. Rev. Lett.}\ }\textbf {\bibinfo {volume} {101}},\
  \bibinfo {pages} {180504} (\bibinfo {year} {2008}{\natexlab{a}})},\ \Eprint
  {https://arxiv.org/abs/0804.0129} {arXiv:0804.0129} \BibitemShut {NoStop}%
\bibitem [{\citenamefont {Bisio}\ \emph {et~al.}(2014)\citenamefont {Bisio},
  \citenamefont {D'Ariano}, \citenamefont {Perinotti},\ and\ \citenamefont
  {Sedl{\'a}k}}]{bisio2014optimal}%
  \BibitemOpen
  \bibfield  {author} {\bibinfo {author} {\bibfnamefont {A.}~\bibnamefont
  {Bisio}}, \bibinfo {author} {\bibfnamefont {G.~M.}\ \bibnamefont {D'Ariano}},
  \bibinfo {author} {\bibfnamefont {P.}~\bibnamefont {Perinotti}},\ and\
  \bibinfo {author} {\bibfnamefont {M.}~\bibnamefont {Sedl{\'a}k}},\ }\bibfield
   {title} {\bibinfo {title} {Optimal processing of reversible quantum
  channels},\ }\href {https://doi.org/10.1016/j.physleta.2014.04.042}
  {\bibfield  {journal} {\bibinfo  {journal} {Physics Letters A}\ }\textbf
  {\bibinfo {volume} {378}},\ \bibinfo {pages} {1797} (\bibinfo {year}
  {2014})},\ \Eprint {https://arxiv.org/abs/1308.3254} {arXiv:1308.3254}
  \BibitemShut {NoStop}%
\bibitem [{\citenamefont {Fulton}(1997)}]{fulton1997young}%
  \BibitemOpen
  \bibfield  {author} {\bibinfo {author} {\bibfnamefont {W.}~\bibnamefont
  {Fulton}},\ }\href {https://doi.org/10.1017/CBO9780511626241} {\emph
  {\bibinfo {title} {{Young Tableaux: With Applications to Representation
  Theory and Geometry}}}},\ \bibinfo {number} {35}\ (\bibinfo  {publisher}
  {Cambridge University Press},\ \bibinfo {year} {1997})\BibitemShut {NoStop}%
\bibitem [{\citenamefont {Georgi}(2000)}]{georgi2000lie}%
  \BibitemOpen
  \bibfield  {author} {\bibinfo {author} {\bibfnamefont {H.}~\bibnamefont
  {Georgi}},\ }\href {https://doi.org/10.1201/9780429499210} {\emph {\bibinfo
  {title} {{Lie Algebras In Particle Physics: From Isospin To Unified
  Theories}}}}\ (\bibinfo  {publisher} {CRC Press, Boca Raton},\ \bibinfo
  {year} {2000})\BibitemShut {NoStop}%
\bibitem [{\citenamefont {Ceccherini-Silberstein}\ \emph
  {et~al.}(2010)\citenamefont {Ceccherini-Silberstein}, \citenamefont
  {Scarabotti},\ and\ \citenamefont {Tolli}}]{ceccherini2010representation}%
  \BibitemOpen
  \bibfield  {author} {\bibinfo {author} {\bibfnamefont {T.}~\bibnamefont
  {Ceccherini-Silberstein}}, \bibinfo {author} {\bibfnamefont {F.}~\bibnamefont
  {Scarabotti}},\ and\ \bibinfo {author} {\bibfnamefont {F.}~\bibnamefont
  {Tolli}},\ }\href {https://doi.org/10.1017/CBO9781139192361} {\emph {\bibinfo
  {title} {{Representation Theory of the Symmetric Groups: The Okounkov-Vershik
  Approach, Character Formulas, and Partition Algebras}}}},\ Vol.\ \bibinfo
  {volume} {121}\ (\bibinfo  {publisher} {Cambridge University Press},\
  \bibinfo {year} {2010})\BibitemShut {NoStop}%
\bibitem [{\citenamefont {Itzykson}\ and\ \citenamefont
  {Nauenberg}(1966)}]{itzykson1966unitary}%
  \BibitemOpen
  \bibfield  {author} {\bibinfo {author} {\bibfnamefont {C.}~\bibnamefont
  {Itzykson}}\ and\ \bibinfo {author} {\bibfnamefont {M.}~\bibnamefont
  {Nauenberg}},\ }\bibfield  {title} {\bibinfo {title} {{Unitary Groups:
  Representations and Decompositions}},\ }\href
  {https://doi.org/10.1103/RevModPhys.38.95} {\bibfield  {journal} {\bibinfo
  {journal} {Rev. Mod. Phys.}\ }\textbf {\bibinfo {volume} {38}},\ \bibinfo
  {pages} {95} (\bibinfo {year} {1966})}\BibitemShut {NoStop}%
\bibitem [{\citenamefont {Taranto}\ \emph {et~al.}(2025)\citenamefont
  {Taranto}, \citenamefont {Milz}, \citenamefont {Murao}, \citenamefont
  {Quintino},\ and\ \citenamefont {Modi}}]{taranto2025higher}%
  \BibitemOpen
  \bibfield  {author} {\bibinfo {author} {\bibfnamefont {P.}~\bibnamefont
  {Taranto}}, \bibinfo {author} {\bibfnamefont {S.}~\bibnamefont {Milz}},
  \bibinfo {author} {\bibfnamefont {M.}~\bibnamefont {Murao}}, \bibinfo
  {author} {\bibfnamefont {M.~T.}\ \bibnamefont {Quintino}},\ and\ \bibinfo
  {author} {\bibfnamefont {K.}~\bibnamefont {Modi}},\ }\bibfield  {title}
  {\bibinfo {title} {{Higher-Order Quantum Operations}},\ }\Eprint
  {https://arxiv.org/abs/2503.09693} {arXiv:2503.09693}  (\bibinfo {year}
  {2025})\BibitemShut {NoStop}%
\bibitem [{\citenamefont {Chiribella}\ \emph
  {et~al.}(2008{\natexlab{b}})\citenamefont {Chiribella}, \citenamefont
  {D’Ariano},\ and\ \citenamefont {Perinotti}}]{chiribella2008quantum}%
  \BibitemOpen
  \bibfield  {author} {\bibinfo {author} {\bibfnamefont {G.}~\bibnamefont
  {Chiribella}}, \bibinfo {author} {\bibfnamefont {G.~M.}\ \bibnamefont
  {D’Ariano}},\ and\ \bibinfo {author} {\bibfnamefont {P.}~\bibnamefont
  {Perinotti}},\ }\bibfield  {title} {\bibinfo {title} {{Quantum Circuit
  Architecture}},\ }\href {https://doi.org/10.1103/PhysRevLett.101.060401}
  {\bibfield  {journal} {\bibinfo  {journal} {Phys. Rev. Lett.}\ }\textbf
  {\bibinfo {volume} {101}},\ \bibinfo {pages} {060401} (\bibinfo {year}
  {2008}{\natexlab{b}})},\ \Eprint {https://arxiv.org/abs/0712.1325}
  {arXiv:0712.1325} \BibitemShut {NoStop}%
\bibitem [{\citenamefont {Wechs}\ \emph {et~al.}(2021)\citenamefont {Wechs},
  \citenamefont {Dourdent}, \citenamefont {Abbott},\ and\ \citenamefont
  {Branciard}}]{wechs2021quantum}%
  \BibitemOpen
  \bibfield  {author} {\bibinfo {author} {\bibfnamefont {J.}~\bibnamefont
  {Wechs}}, \bibinfo {author} {\bibfnamefont {H.}~\bibnamefont {Dourdent}},
  \bibinfo {author} {\bibfnamefont {A.~A.}\ \bibnamefont {Abbott}},\ and\
  \bibinfo {author} {\bibfnamefont {C.}~\bibnamefont {Branciard}},\ }\bibfield
  {title} {\bibinfo {title} {{Quantum Circuits with Classical Versus Quantum
  Control of Causal Order}},\ }\href
  {https://doi.org/10.1103/PRXQuantum.2.030335} {\bibfield  {journal} {\bibinfo
   {journal} {PRX Quantum}\ }\textbf {\bibinfo {volume} {2}},\ \bibinfo {pages}
  {030335} (\bibinfo {year} {2021})},\ \Eprint
  {https://arxiv.org/abs/2101.08796} {arXiv:2101.08796} \BibitemShut {NoStop}%
\bibitem [{\citenamefont {Hardy}(2007)}]{hardy2007towards}%
  \BibitemOpen
  \bibfield  {author} {\bibinfo {author} {\bibfnamefont {L.}~\bibnamefont
  {Hardy}},\ }\bibfield  {title} {\bibinfo {title} {Towards quantum gravity: a
  framework for probabilistic theories with non-fixed causal structure},\
  }\href {https://doi.org/10.1088/1751-8113/40/12/S12} {\bibfield  {journal}
  {\bibinfo  {journal} {Journal of Physics A: Mathematical and Theoretical}\
  }\textbf {\bibinfo {volume} {40}},\ \bibinfo {pages} {3081} (\bibinfo {year}
  {2007})},\ \Eprint {https://arxiv.org/abs/gr-qc/0608043}
  {arXiv:gr-qc/0608043} \BibitemShut {NoStop}%
\bibitem [{\citenamefont {Oreshkov}\ \emph {et~al.}(2012)\citenamefont
  {Oreshkov}, \citenamefont {Costa},\ and\ \citenamefont
  {Brukner}}]{oreshkov2012quantum}%
  \BibitemOpen
  \bibfield  {author} {\bibinfo {author} {\bibfnamefont {O.}~\bibnamefont
  {Oreshkov}}, \bibinfo {author} {\bibfnamefont {F.}~\bibnamefont {Costa}},\
  and\ \bibinfo {author} {\bibfnamefont {{\v{C}}.}~\bibnamefont {Brukner}},\
  }\bibfield  {title} {\bibinfo {title} {Quantum correlations with no causal
  order},\ }\href {https://doi.org/10.1038/ncomms2076} {\bibfield  {journal}
  {\bibinfo  {journal} {Nature communications}\ }\textbf {\bibinfo {volume}
  {3}},\ \bibinfo {pages} {1092} (\bibinfo {year} {2012})},\ \Eprint
  {https://arxiv.org/abs/1105.4464} {arXiv:1105.4464} \BibitemShut {NoStop}%
\bibitem [{\citenamefont {Chiribella}\ \emph {et~al.}(2013)\citenamefont
  {Chiribella}, \citenamefont {D'Ariano}, \citenamefont {Perinotti},\ and\
  \citenamefont {Valiron}}]{chiribella2013quantum}%
  \BibitemOpen
  \bibfield  {author} {\bibinfo {author} {\bibfnamefont {G.}~\bibnamefont
  {Chiribella}}, \bibinfo {author} {\bibfnamefont {G.~M.}\ \bibnamefont
  {D'Ariano}}, \bibinfo {author} {\bibfnamefont {P.}~\bibnamefont
  {Perinotti}},\ and\ \bibinfo {author} {\bibfnamefont {B.}~\bibnamefont
  {Valiron}},\ }\bibfield  {title} {\bibinfo {title} {Quantum computations
  without definite causal structure},\ }\href
  {https://doi.org/10.1103/PhysRevA.88.022318} {\bibfield  {journal} {\bibinfo
  {journal} {Phys. Rev. A}\ }\textbf {\bibinfo {volume} {88}},\ \bibinfo
  {pages} {022318} (\bibinfo {year} {2013})},\ \Eprint
  {https://arxiv.org/abs/0912.0195} {arXiv:0912.0195} \BibitemShut {NoStop}%
\bibitem [{\citenamefont {Matsumoto}(2012)}]{matsumoto2012input}%
  \BibitemOpen
  \bibfield  {author} {\bibinfo {author} {\bibfnamefont {K.}~\bibnamefont
  {Matsumoto}},\ }\bibfield  {title} {\bibinfo {title} {When is an input state
  always better than the others?: universally optimal input states for
  statistical inference of quantum channels},\ }\Eprint
  {https://arxiv.org/abs/1209.2392} {arXiv:1209.2392}  (\bibinfo {year}
  {2012})\BibitemShut {NoStop}%
\bibitem [{\citenamefont {Fuchs}\ and\ \citenamefont {Van
  De~Graaf}(2002)}]{fuchs2002cryptographic}%
  \BibitemOpen
  \bibfield  {author} {\bibinfo {author} {\bibfnamefont {C.~A.}\ \bibnamefont
  {Fuchs}}\ and\ \bibinfo {author} {\bibfnamefont {J.}~\bibnamefont {Van
  De~Graaf}},\ }\bibfield  {title} {\bibinfo {title} {Cryptographic
  distinguishability measures for quantum-mechanical states},\ }\href
  {https://doi.org/10.1109/18.761271} {\bibfield  {journal} {\bibinfo
  {journal} {IEEE transactions on information theory}\ }\textbf {\bibinfo
  {volume} {45}},\ \bibinfo {pages} {1216} (\bibinfo {year} {2002})},\ \Eprint
  {https://arxiv.org/abs/quant-ph/9712042} {arXiv:quant-ph/9712042}
  \BibitemShut {NoStop}%
\bibitem [{\citenamefont {Horn}\ and\ \citenamefont
  {Johnson}(2012)}]{horn2012matrix}%
  \BibitemOpen
  \bibfield  {author} {\bibinfo {author} {\bibfnamefont {R.~A.}\ \bibnamefont
  {Horn}}\ and\ \bibinfo {author} {\bibfnamefont {C.~R.}\ \bibnamefont
  {Johnson}},\ }\href {https://doi.org/10.1017/CBO9780511810817} {\emph
  {\bibinfo {title} {Matrix analysis}}}\ (\bibinfo  {publisher} {Cambridge
  university press},\ \bibinfo {year} {2012})\BibitemShut {NoStop}%
\bibitem [{\citenamefont {DeGroot}\ and\ \citenamefont
  {Schervish}(2010)}]{degroot2012probability}%
  \BibitemOpen
  \bibfield  {author} {\bibinfo {author} {\bibfnamefont {M.~H.}\ \bibnamefont
  {DeGroot}}\ and\ \bibinfo {author} {\bibfnamefont {M.~J.}\ \bibnamefont
  {Schervish}},\ }\href@noop {} {\emph {\bibinfo {title} {{Probability and
  Statistics}}}},\ \bibinfo {edition} {4th}\ ed.\ (\bibinfo  {publisher}
  {Addison-Wesley},\ \bibinfo {year} {2010})\BibitemShut {NoStop}%
\bibitem [{\citenamefont {Chiribella}\ \emph {et~al.}(2009)\citenamefont
  {Chiribella}, \citenamefont {D’Ariano},\ and\ \citenamefont
  {Perinotti}}]{chiribella2009optimal}%
  \BibitemOpen
  \bibfield  {author} {\bibinfo {author} {\bibfnamefont {G.}~\bibnamefont
  {Chiribella}}, \bibinfo {author} {\bibfnamefont {G.~M.}\ \bibnamefont
  {D’Ariano}},\ and\ \bibinfo {author} {\bibfnamefont {P.}~\bibnamefont
  {Perinotti}},\ }\bibfield  {title} {\bibinfo {title} {Optimal covariant
  quantum networks},\ }in\ \href {https://doi.org/10.1063/1.3131375} {\emph
  {\bibinfo {booktitle} {AIP Conference Proceedings}}},\ Vol.\ \bibinfo
  {volume} {1110}\ (\bibinfo {organization} {American Institute of Physics},\
  \bibinfo {year} {2009})\ pp.\ \bibinfo {pages} {47--56},\ \Eprint
  {https://arxiv.org/abs/0812.3922} {arXiv:0812.3922} \BibitemShut {NoStop}%
\bibitem [{\citenamefont {Chiribella}\ \emph
  {et~al.}(2008{\natexlab{c}})\citenamefont {Chiribella}, \citenamefont
  {D’Ariano},\ and\ \citenamefont {Perinotti}}]{chiribella2008memory}%
  \BibitemOpen
  \bibfield  {author} {\bibinfo {author} {\bibfnamefont {G.}~\bibnamefont
  {Chiribella}}, \bibinfo {author} {\bibfnamefont {G.~M.}\ \bibnamefont
  {D’Ariano}},\ and\ \bibinfo {author} {\bibfnamefont {P.}~\bibnamefont
  {Perinotti}},\ }\bibfield  {title} {\bibinfo {title} {{Memory Effects in
  Quantum Channel Discrimination}},\ }\href
  {https://doi.org/10.1103/PhysRevLett.101.180501} {\bibfield  {journal}
  {\bibinfo  {journal} {Phys. Rev. Lett.}\ }\textbf {\bibinfo {volume} {101}},\
  \bibinfo {pages} {180501} (\bibinfo {year} {2008}{\natexlab{c}})},\ \Eprint
  {https://arxiv.org/abs/0803.3237} {arXiv:0803.3237} \BibitemShut {NoStop}%
\bibitem [{\citenamefont {Bavaresco}\ \emph {et~al.}(2022)\citenamefont
  {Bavaresco}, \citenamefont {Murao},\ and\ \citenamefont
  {Quintino}}]{bavaresco2022unitary}%
  \BibitemOpen
  \bibfield  {author} {\bibinfo {author} {\bibfnamefont {J.}~\bibnamefont
  {Bavaresco}}, \bibinfo {author} {\bibfnamefont {M.}~\bibnamefont {Murao}},\
  and\ \bibinfo {author} {\bibfnamefont {M.~T.}\ \bibnamefont {Quintino}},\
  }\bibfield  {title} {\bibinfo {title} {Unitary channel discrimination beyond
  group structures: Advantages of sequential and indefinite-causal-order
  strategies},\ }\href {https://doi.org/10.1063/5.0075919} {\bibfield
  {journal} {\bibinfo  {journal} {Journal of Mathematical Physics}\ }\textbf
  {\bibinfo {volume} {63}},\ \bibinfo {pages} {042203} (\bibinfo {year}
  {2022})},\ \Eprint {https://arxiv.org/abs/2105.13369} {arXiv:2105.13369}
  \BibitemShut {NoStop}%
\end{thebibliography}%

\end{document}